\documentclass[acmsmall,screen,dvipsnames, authorversion,nonacm]{acmart}
\usepackage{graphicx}
\usepackage{hyperref}
\usepackage{amsfonts}
\usepackage{amsmath}
\usepackage{mathtools}
\usepackage{caption}
\usepackage{mathrsfs}
\usepackage{stmaryrd}
\usepackage{scalerel}
\usepackage{dsfont}
\usepackage[capitalize]{cleveref}
\usepackage{quiver}
\usepackage{tikz}
\usepackage{cancel}
\usepackage{todonotes}
\usepackage{stackengine}
\usetikzlibrary{automata,
arrows,
positioning,
shapes.geometric,
shapes,
fit,
calc,
positioning,
commutative-diagrams,
backgrounds}
\usepackage{thmtools}
\usepackage{mathpartir}
\newcommand \acro[1]	{\(\mathsf{#1}\)} 				

\def\rho{\varrho}
\def\theta{\vartheta}
\def\phi{\varphi}

\def\Nat{\mathds{N}}

\newcommand{\cat}[1]{\mathscr{#1}}
\def\A{\cat{A}}
\def\C{\cat{C}}
\def\D{\cat{D}}

\def\zero{{\mathsf{0}}}
\def\one{{\mathsf{1}}}

\def\Set{{\mathsf{Set}}}
\def\Rel{{\mathsf{Rel}}}
\def\pca{{\mathsf{PCA}}}

\def\coa{\mathop{\mathsf{Coalg}}}
\def\alg{\mathop{\mathsf{Alg}}}
\def\coaf{\mathop{\mathsf{Coalg}_{\mathsf{fp}}}}

\def\coafr{\mathop{\mathsf{Coalg}_{\mathsf{free}}}}

\def\Id{\mathsf{Id}}
\def\powf{\mathcal{P}_f}
\def\distf{\mathcal{D}}

\def\id{\mathsf{id}}
\renewcommand{\o}{\ensuremath{\circ}}

\def\ol#1{\overline{#1}}

\def\colim{\mathop{\mathrm{colim}}}

\newcommand\supp{\operatorname{supp}}
\newcommand\cl{\operatorname{cl}}

\newcommand     \Exp {\mathsf{Exp}}
\newcommand     \seq {\mathbin{;}}
\newcommand 	\ex {\operatorname{exp}}
\newcommand   \uaequiv {\mathrel{\dot\equiv}}
\DeclareMathOperator*{\bigboxplus}{\scalerel*{\boxplus}{\bigoplus}}

\def\DoubleFill@{\arrowfill@\Relbar\Relbar\Relbar}

\def\eqgap{.2ex}
\def\overgap{.4ex}
\def\inferrulerule{.2pt}

\newlength\rulelength
\newlength\toplength
\newlength\bottomlength

\newcommand\myinferrule[2]{%
  \stackMath%
  \setlength\bottomlength{\widthof{$#1$}}%
  \setlength\toplength{\widthof{$#2$}}%
  \ifdim\toplength>\bottomlength%
    \setlength\rulelength{\the\toplength}%
  \else%
    \setlength\rulelength{\the\bottomlength}%
  \fi%
  \mathrel{%
  \stackon[\overgap]{\stackon[\eqgap]{\stackon[\overgap]{#1}%
    {\rule{\the\rulelength}{\inferrulerule}}}%
    {\rule{\the\rulelength}{\inferrulerule}}}{#2}%
  }%
}
\usepackage[most]{tcolorbox}
\usepackage[all,cmtip]{xy}
\usepackage{multicol}
\usepackage{enumitem}

\AtBeginDocument{%
  }

\copyrightyear{2024} 
\acmYear{2024} 
\setcopyright{rightsretained} 
\acmConference[LICS '24]{39th Annual ACM/IEEE Symposium on Logic in Computer Science}{July 8--11, 2024}{Tallinn, Estonia}
\acmBooktitle{39th Annual ACM/IEEE Symposium on Logic in Computer Science (LICS '24), July 8--11, 2024, Tallinn, Estonia}\acmDOI{10.1145/3661814.3662084}
\acmISBN{979-8-4007-0660-8/24/07}

\makeatletter
\def\@acmplainindent{0pt}
\def\@acmdefinitionindent{0pt}
\def\@proofindent{\noindent}
\makeatother

\begin{document}

\title{A Completeness Theorem for Probabilistic Regular Expressions (Full Version)}

\author{Wojciech Różowski}
\orcid{0000-0002-8241-7277}
\affiliation{%
  \institution{Department of Computer Science}
  \institution{University College London}
  \city{London}
  \country{UK}
}
\email{w.rozowski@cs.ucl.ac.uk}

\author{Alexandra Silva}
\orcid{0000-0001-5014-9784}
\affiliation{%
  \institution{Department of Computer Science}
  \institution{Cornell University}
  \city{Ithaca}
  \country{USA}
}
\email{alexandra.silva@cornell.edu}

\renewcommand{\shortauthors}{Wojciech Różowski and Alexandra Silva}

\begin{abstract}
 We introduce Probabilistic Regular Expressions (\acro{PRE}), a probabilistic analogue of regular expressions denoting probabilistic languages in which every word is assigned a probability of being generated. We present and prove the completeness of an inference system for reasoning about probabilistic language equivalence of \acro{PRE} based on Salomaa's axiomatisation of Kleene Algebra. 

\end{abstract}



\maketitle

\section{Introduction}
    Kleene~\cite{Kleene:1951:Representation} introduced regular expressions and proved that these denote exactly the languages accepted by deterministic finite automata. In his seminal paper, Kleene left open a completeness question: are there a finite number of equations that enable reasoning about language equivalence of regular expressions? Since then, the pursuit of inference systems for equational reasoning about the equivalence of regular expressions has been  subject of extensive study~\cite{Salomaa:1966:Two,Krob:1990:Complete,Boffa:1990:Une,Kozen:1994:Completeness}. The first proposal is due to Salomaa~\cite{Salomaa:1966:Two}, who introduced a non-algebraic axiomatisation of regular expressions and proved its completeness.
    
    Deterministic automata are a particular type of transition system: simply put, an automaton is an object with a finite set of states and a {\em deterministic} transition function that assigns every state and every action of the input alphabet {\em exactly} one next state. By varying the type of transition function one gets different systems: e.g. if the function assigns to every state and every action of the input alphabet {\em a set} of the next states, the resulting system is said to be non-deterministic; if the transition function assigns the next state based on any sort of {\em probability distribution} then the system is said to be probabilistic. Probabilistic systems appear in a range of applications, including modelling randomised algorithms, cryptographic protocols, and probabilistic programs. In this paper, we focus on generative probabilistic transition systems (\acro{GPTS}) with explicit termination~\cite{Glabbeek:1995:Reactive}, and study the questions that Kleene and Salomaa answered for deterministic automata. 
    
    Our motivation to look at probabilistic extensions of regular expressions and axiomatic reasoning is two-fold:  first, regular expressions and extensions thereof have been used in the verification of uninterpreted imperative programs, including network policies~\cite{Kozen:2000:Certification,Kot:2005:Kleene,Anderson:2014:NetKAT}; second, reasoning about exact behaviour of probabilistic imperative programs is subtle~\cite{Chen:2022:Does}, in particular in the presence of loops. By studying the semantics and axiomatisations of regular expressions featuring probabilistic primitives, we want to enable axiomatic reasoning for randomised programs and provide a basis to develop further verification techniques.

    We start by introducing the syntax of Probabilistic Regular Expressions (\acro{PRE}), inspired by work from the probabilistic pattern matching literature~\cite{Ross:2000:Probabilistic}. \acro{PRE} are formed through constants from an alphabet and \emph{regular} operations of probabilistic choice, sequential composition, probabilistic Kleene star, identity and emptiness. We define the probabilistic analogue of {\em language semantics} of \acro{PRE} as exactly the behaviours of \acro{GPTS}. We achieve this by endowing \acro{PRE} with operational semantics in the form of \acro{GPTS} via a construction reminiscent of Antimirov derivatives of regular expressions~\cite{Antimirov:1996:Partial}. We also give a converse construction, allowing us to describe languages accepted by finite-state \acro{GPTS} in terms of \acro{PRE}, thus establishing an analogue of Kleene's theorem.

    The main contribution of this paper is presenting an inference system for reasoning about probabilistic language equivalence of \acro{PRE} and proving its completeness. The technical core of the paper is devoted to the completeness proof, which relies on technical tools convex algebra, arising from the rich structure of probabilistic languages. While being in the spirit of classic results from automata theory, our development relies on the more abstract approach enabled via the theory of universal coalgebra~\cite{Rutten:2000:Universal}. As much as a concrete completeness proof ought to be possible, our choice to use the coalgebraic approach was fueled by wanting to reuse recent abstract results on the algebraic structure of probabilistic languages. A concrete proof would have to deal with fixpoints of probabilistic languages and would therefore require highly combinatorial and syntactic proofs about these. Instead, we reuse a range of hard results on convex algebras and fixpoints that Milius~\cite{Milius:2018:Proper}, Sokolova and Woracek~\cite{Sokolova:2015:Congruences,Sokolova:2018:Proper} proved in the last 5 years. In particular, we rely on the theory of rational fixpoints (for proper functors~\cite{Milius:2018:Proper}), which can be seen as a categorical generalisation of regular languages. 
    
    Our completeness proof provides further evidence that the use of coalgebras over proper functors provides a good abstraction for completeness theorems, where general steps can be abstracted away leaving as a domain-specific task to achieve completeness a construction to syntactically build solutions to systems of equations. Proving the uniqueness of such solutions is ultimately the most challenging step in the proof. By leveraging the theory of proper functors, our proof of completeness, which depends on establishing an abstract universal property, boils down to an argument that can be viewed as a natural extension of the work by Salomaa~\cite{Salomaa:1966:Two} and Brzozowski~\cite{Brzozowski:1964:Derivatives} from the 1960s.
    
    In a nutshell, the contributions of this paper are as follows. 

    \begin{itemize}[leftmargin=*,topsep=5pt]
    	\item We introduce Probabilistic Regular Expressions (\acro{PRE}), an analogue of Kleene's regular expressions denoting probabilistic languages and propose an inference system for reasoning about language equivalence of \acro{PRE} (\Cref{sec:overview}). 
    	\item We provide a small-step semantics of \acro{PRE} through an analogue of Antimirov derivatives (\Cref{sec:operational_semantics}) endowing expressions with a structure of Generative Probabilistic Transition Systems (\acro{GPTS}).
    	\item We give an overview of the coalgebraic setup for our main results (\Cref{sec:roadmap,sec:language_semantics}) exposing which specific technical lemmas are needed for our concrete proof. 
	\item We prove soundness (\Cref{sec:soundness}) and completeness (\Cref{sec:completeness}) results for our axiomatisation. Due to our use of proper functors, the proof boils down to a generalisation of a known proof of Salomaa for regular expressions~\cite{Salomaa:1966:Two} exposing the connection to a classical result. We also obtain an analogue of Kleene's theorem allowing the conversion of finite-state \acro{GPTS} to expressions through an analogue of Brzozowski's method~\cite{Brzozowski:1964:Derivatives}.
    \end{itemize}
To improve the accessibility of the paper, we recall key definitions from the literature within subsections encased in coloured boxes \begin{tcolorbox}[colback=Periwinkle!5!white,colframe=Periwinkle!75!black,on line,hbox, boxsep=0pt, top=2pt,
left=2pt, bottom=2pt, right=2pt] with a light blue background\end{tcolorbox}. 
    Proofs appear in the appendix.
    
\section{Probabilistic Regular Expressions}\label{sec:overview}
In this section, we will introduce the syntax and the language semantics of probabilistic regular expressions (\acro{PRE}), as well as a candidate inference system to reason about the equivalence of \acro{PRE}. 

\noindent
\textbf{Syntax.} Given a finite alphabet $A$, the syntax of \acro{PRE} is given by:
$$e,f \in \Exp ::= \zero \mid \one \mid a \in A \mid e \oplus_p f \mid e\seq f \mid e^{[p]} \quad\quad\quad p \in [0,1]$$
We denote the expressions that immediately abort and successfully terminate by $\zero$ and $\one$ respectively. For every letter $a \in A$ in the alphabet, there is a corresponding expression representing an atomic action. Given two expressions $e, f \in {\Exp}$ and $p \in [0,1]$, probabilistic choice $e \oplus_p f$ denotes an expression that performs $e$ with probability $p$ and performs $f$ with probability $1-p$. One can think of $ \oplus_p$ as the probabilistic analogue of the
 plus operator ($e + f$) in Kleene's regular expressions. $e \seq f$ represents sequential composition, while $e^{[p]}$ is a probabilistic analogue of Kleene star: it successfully terminates with probability $1-p$ or with probability $p$ performs $e$ and then iterates $e^{[p]}$ again.  In terms of the notational convention, the sequential composition $(\seq)$ has higher precedence than the probabilistic choice $(\oplus_p)$.

\begin{example}
    The expression $a \seq a^{[\frac{1}{4}]}$ first performs action $a$ with probability $1$ and then enters a loop which successfully terminates with probability $\frac{3}{4}$ or performs action $a$ with probability $\frac{1}{4}$ and then repeats the loop again. Intuitively, if we think of the action $a$ as observable, the expression above denotes a probability associated with a non-empty sequence of $a$'s. For example, the sequence $aaa$ would be observed with probability $1\times (1/4)^2 \times 3/4=3/64$.
\end{example}

\noindent
\textbf{Language semantics.} \acro{PRE} denote probabilistic languages $A^* \to [0,1]$. For instance, the expression $\zero$ denotes a function that assigns $0$ to every word, whereas $\one$ and $a$ respectively assign probability $1$ to the empty word and the word containing a single letter $a$ from the alphabet. The probabilistic choice $e \oplus_p f$ denotes a language in which the probability of each word is the total sum of its probability in $e$ scaled by $p$ and its probability in $f$ scaled by $1-p$. Describing the semantics of sequential composition and loops inductively is more involved. In particular, the semantics of loops would require a fixpoint calculation, which does not have as clear and straightforward (closed-form) formula, as the asterate of regular languages. Instead, we take an \emph{operational approach}, and we formally define the language semantics of \acro{PRE} in \Cref{sec:language_semantics} through a small-step operational semantics, using a specific type of probabilistic transition system, which we introduce next.

\vspace{2px}
\noindent
\textbf{Generative probabilistic transition systems.} A \acro{GPTS} consists of a set of states $Q$ and a transition function that maps each state $q\in Q$ to finitely many distinct outgoing arrows of the form:
\begin{itemize}[leftmargin=*]
    \item \emph{successful termination} with probability $t$ (denoted $q \xRightarrow[]{t} \checkmark$), or 
    \item to another state $r$, via an \emph{$a$-labelled transition}, with probability $s \in [0,1] $ (denoted $q \xrightarrow[]{a \mid s} r$).
\end{itemize} 
We require that, for each state, the total sum of probabilities appearing on outgoing arrows sums up to less or equal to one. The remaining probability mass is used to model unsuccessful termination, hence the state with no outgoing arrows can be thought of as exposing deadlock behaviour. 

Given a word $w \in A^*$ the probability of it being generated by a state $q\in Q$ (denoted $\llbracket q\rrbracket (w) \in [0,1]$) is defined inductively:
\begin{equation}\label{language}
    \llbracket q\rrbracket (\epsilon)=t \quad\text{if } q \xRightarrow[]{t} \checkmark \qquad\qquad
    \llbracket q \rrbracket(av) = \sum_{q \xrightarrow[]{a \mid s} r} s \cdot \llbracket r \rrbracket(v)
\end{equation}
We say that two states $q$ and $q'$ are \emph{language equivalent} if for all words $w \in A^*$, we have that $\llbracket q \rrbracket(w)=\llbracket q' \rrbracket(w)$.
\begin{example}\label{ex:systems}
 Consider the following \acro{GPTS}:
 
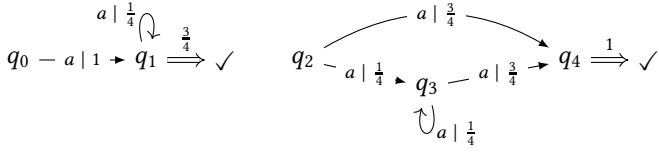
\begin{gather*}
\begin{tikzpicture}[baseline=-3ex]
        \node (0) {$q_0$};
        \node (1) [right= 1.15 cm of 0]{$q_1$};
        \draw (0) edge[-latex] node[ fill=white] {\scriptsize\( {a} \mid {1}\)} (1);
        \draw (1) edge[loop above] node[left] {\scriptsize\( {a} \mid {\frac{1}{4}}\)} (1);
        \node (o1) [right= 0.5 cm of 1]{$\checkmark$};
        \draw (1) edge[-implies, double, double distance=0.5mm] node[above] {\scriptsize\( \frac{3}{4}\)} (o1);
        \node (2) [right= 0.5 cm of o1] {$q_2$};
         \node (2a) [right= 1.25 cm of 2] {};
        \node (3) [below= -.01 cm of 2a] {$q_3$};
        \draw (2) edge[-latex] node[ fill=white] {\scriptsize\( {a} \mid {\frac{1}{4}}\)} (3);
        \draw (3) edge[loop below] node[right] {\scriptsize\( {a} \mid {\frac{1}{4}}\)} (3);
        \node (4) [right= 3 cm of 2] {$q_4$};
        \draw (3) edge[-latex] node[ fill=white] {\scriptsize\( {a} \mid {\frac{3}{4}}\)} (4);
        \draw (2) edge[-latex, bend left] node[ fill=white] {\scriptsize\( {a} \mid {\frac{3}{4}}\)} (4);
        \node (o2) [right= 0.5 cm of 4]{$\checkmark$};
        \draw (4) edge[-implies, double, double distance=0.5mm] node[above] {\scriptsize\( 1\)} (o2);
\end{tikzpicture}
\end{gather*}
States $q_0$ and $q_2$ both assign probability $0$ to the empty word $\epsilon$ and each word $a^{n+1}$ is mapped to the probability $\left(\frac{1}{4}\right)^n \times \frac{3}{4}$. Later in the paper, we show that the languages generated by states $q_0$ and $q_2$ can be specified using expressions $a\seq a^{\left[ \frac{1}{4} \right]}$ and $a \oplus_{\frac{3}{4}} ( a\seq a^{\left[ \frac{1}{4} \right]} \seq a)$ respectively.
\end{example}
In \Cref{sec:operational_semantics}, we will associate to each \acro{PRE} $e$ an operational semantics or, more precisely, a state $q_e$ in a \acro{GPTS}. The language semantics of $e$ will then be the language $\llbracket q_e \rrbracket \colon A^* \to [0,1]$ generated by $q_e$. Two \acro{PRE} $e$ and $f$ are language equivalent if $\llbracket q_e \rrbracket=\llbracket q_f \rrbracket$. One of our main goals is to present a complete inference system to reason about language equivalence. In a nutshell, we want to present a system of (quasi-)equations of the form $e\equiv f$ such that:
\[
e\equiv f \Leftrightarrow \llbracket q_e \rrbracket=\llbracket q_f \rrbracket
\]
Such an inference system will have to contain rules to reason about all constructs of \acro{PRE}, including probabilistic choice and loops. We describe next the system, with some intuition for the inclusion of each group of rules. 

\vspace{2px}
\noindent
\textbf{Axiomatisation of language equivalence of \acro{PRE}.} We define ${\equiv} \subseteq \Exp \times \Exp$ to be the least congruence relation closed under the axioms shown on \cref{fig:axioms}. We will show in \Cref{sec:completeness} that these axioms are complete wrt. language semantics. 

\begin{figure*}
        \centering
        \begin{multicols}{2}
        	 \begin{flushleft}\underline{\bf Probabilistic Choice}  	 \end{flushleft}
        \begin{alignat*}{4}
		\textbf{(C1)} & \;\, 
		& e &{} \equiv e \oplus_p e\\
				\textbf{(C2)} & \;\, 
		& e &{} \equiv e \oplus_1 f\\
		\textbf{(C3)} & \;\, 
		& e \oplus_p  f &{} \equiv f \oplus_{\overline{p}} e\\
		\textbf{(C4)} & \;\, 
		& (e \oplus_{p} f) \oplus_{q} g & {} \equiv e \oplus_{pq} (f \oplus_{\frac{\overline{p}q}{1-pq}} g)\\
		\textbf{(D1)} & \;\, 
		& (e \oplus_{p} f) \seq g & {} \equiv e \seq g \oplus_{p} f \seq g\\[-0.25ex]
		\textbf{(D2)} & \;\, 
		& e \seq (f \oplus_{p} g) & {} \equiv e \seq g \oplus_{p} e \seq f\\[-0.25ex]\\\\
        \end{alignat*}
          \begin{flushleft}   \underline{\bf Sequencing} \end{flushleft}
        \begin{alignat*}{4}
        \textbf{(0S)} & \;\, 
		& 0 \seq e & {} \equiv 0 & \\[-0.25ex]
		\textbf{(S0)} & \;\, 
		& e \seq 0 & {} \equiv 0 & \\[-0.25ex]
		\textbf{(1S)} & \;\, 
		& 1 \seq e & {} \equiv e & \\[-0.25ex]
		\textbf{(S1)} & \;\, 
		& e \seq 1 & {} \equiv e \\[-0.25ex]
		\textbf{(S)} & \;\, 
		& e \seq (f \seq g) & {} \equiv (e \seq f ) \seq g \\[-0.25ex]
        \end{alignat*}
        \end{multicols}
        \vspace{-2cm}
         \begin{multicols}{2}
     \begin{flushleft}    \underline{\bf Loops} \end{flushleft}
        \begin{alignat*}{4}
       	\textbf{(Unroll)} & \;\, 
		& e \seq e^{[p]} \oplus_p 1& {} \equiv e^{[p]} \\
		\textbf{(Tight)} & \;\, 
		& (e \oplus_p 1)^{[q]} \seq 1 & {} \equiv e^{\left[\frac{pq}{1-\ol{p}q}\right]} \\[-0.25ex]
		\textbf{(Div)} & \;\, 
		& 1^{[1]} & {} \equiv 0 \\
        \end{alignat*}
          \begin{flushleft}   \underline{\bf (Unique) fixpoint rule} \end{flushleft}
		$\inferrule{g \equiv e\seq g \oplus_p f \quad \fcolorbox{red}{white}{$E(e)=0$}}{g \equiv e^{[p]}\seq f}$
\end{multicols}
        \vspace{-.5cm}
     \begin{flushleft}    \underline{\bf Termination cond. $E : \Exp \to [0,1]$} \end{flushleft}
        \fcolorbox{red}{white}{$
        \begin{array}{l}
            \quad E(\one)=1 \qquad E(\zero)=E(a) = 0 \qquad
            E \left( e \oplus_p f \right)=pE(e) + \ol{p}E(f)\qquad 
            E(e\seq f)=E(e)E(f)\quad \\[1.4ex]
            \quad E\left(e^{[p]}\right)=
                \begin{cases}
                0 &  \text{\scriptsize $E(e)=1 \wedge p=1$} \\ 
                \frac{1-p}{1-pE(e)} & \text{\scriptsize otherwise}
                \end{cases}   
                \end{array}$
     }
        \caption{\label{fig:axioms} For $p \in [0,1]$, we write $\ol{p}=1-p$. The rules involving the division of probabilities are defined only when the denominator is non-zero. The function $E(-)$ provides a termination side condition to the \textbf{Unique} fixpoint axiom.}
            \vspace{-.4cm}
    \end{figure*}

The first group of axioms capture properties of the probabilistic choice operator $\oplus_p$ (\textbf{C1-C4}) and its interaction with sequential composition (\textbf{D1-D2}). Intuitively, \textbf{C1-C4} are the analogue of the semilattice axioms governing the behaviour of $+$ in regular expressions. These four axioms are reminiscent of the axioms of barycentric algebras~\cite{Stone:1949:Postulates}. \textbf{D1} and \textbf{D2} are \emph{right and left distributivity} rules of $\oplus$ over $\seq$. The sequencing axioms \textbf{1S, S1, S}  state \acro{PRE} have the structure of a monoid (with neutral element $\one$) with absorbent element $\zero$ (\textbf{0S, S0}). The loop axioms contain respectively \emph{unrolling}, \emph{tightening}, and \emph{divergency} axioms plus a \emph{unique fixpoint} rule. The \textbf{Unroll} axiom associates loops with their intuitive behaviour of choosing, at each step, probabilistically between successful termination and executing the loop body once. \textbf{Tight} and \textbf{Div} are the probabilistic analogues of the identity $(e + \one)^{*} \equiv e^{*}$ from regular expressions. In the case of \acro{PRE}, we need two axioms: \textbf{Tight} states that the probabilistic loop whose body might instantly terminate, causing the next loop iteration to be executed immediately is provably equivalent to a different loop, whose body does not contain immediate termination; \textbf{Div} takes care of the edge case of a no-exit loop and identifies it with failure. Finally, the unique fixpoint rule is a re-adaptation of the analogous axiom from Salomaa's axiomatisation and provides a partial converse to the loop unrolling axiom, given the loop body is productive -- i.e. cannot immediately terminate. This productivity property is formally written using the side condition $E(e) = 0$, which can be thought of as the probabilistic analogue of empty word property from Salomaa’s axiomatisation. Consider an expression $a^{\left[\frac{1}{2}\right]} \seq (b \oplus_{\frac{1}{2}} 1)$. The only way it can accept the empty word is to leave the loop with the probability of $\frac{1}{2}$ and then perform $1$, which also can happen with probability $\frac{1}{2}$. In other words, $\llbracket a^{\left[\frac{1}{2}\right]} \seq (b \oplus_{\frac{1}{2}} 1)\rrbracket(\epsilon) = \frac{1}{4}$. A simple calculation allows to verify that $E( a^{\left[\frac{1}{2}\right]} \seq (b \oplus_{\frac{1}{2}} 1)) = \frac{1}{4}$.

\begin{example} We revisit the expressions from \cref{ex:systems} and show their equivalence via axiomatic reasoning.
\begin{align*}
     a \seq a^{[\frac{1}{4}]} &\stackrel\dagger\equiv   a \seq (a^{[\frac{1}{4}]} \seq a \oplus_{\frac{1}{4}} \one) \stackrel{\textbf{D2}}\equiv  (a \seq a^{[\frac{1}{4}]}\seq a) \oplus_{\frac{1}{4}} a \seq \one\\  &\stackrel{\textbf{S1}}\equiv (a \seq a^{[\frac{1}{4}]}\seq a) \oplus_{\frac{1}{4}} a \stackrel{\textbf{C3}}\equiv a \oplus_{\frac{3}{4}} (a \seq a^{[\frac{1}{4}]}\seq a)
\end{align*}

The $\dagger$ step of the proof above relies on the equivalence $e^{[p]}\seq e \oplus_p \one \equiv e$ derivable from other axioms under the assumption $E(e)=0$ through a following line of reasoning:
\begin{align*}
    e^{[p]} \seq e \oplus_p \one  &\equiv (e \seq e^{[p]} \oplus_p \one)\seq e\oplus_p \one \tag{\textbf{Unroll}}\\
    &\equiv (e \seq e^{[p]}\seq e \oplus_p \one\seq e) \oplus_p \one \tag{\textbf{D1}}\\
    &\equiv (e \seq e^{[p]}\seq e \oplus_p e) \oplus_p \one \tag{\textbf{1S}}\\
        &\equiv (e \seq e^{[p]}\seq e \oplus_p e\seq \one) \oplus_p \one \tag{\textbf{S1}}\\
    &\equiv e \seq (e^{[p]}\seq e \oplus_p \one) \oplus_p \one \tag{\textbf{D2}}
\end{align*}
Since $E(e)=0$,  we then have: 
\(
     e^{[p]} \seq e \oplus_p \one \stackrel{(\textbf{Unique})}\equiv e^{[p]}\seq \one 
     \stackrel{(\textbf{S1})}\equiv e^{[p]} .
\)
\end{example}

\section{Operational Semantics of \acro{PRE}}\label{sec:operational_semantics}
In this section, we recall the formal definition of \acro{GPTS} and then provide a procedure, inspired by Antimirov derivatives~\cite{Antimirov:1996:Partial}, to associate with each \acro{PRE} a \acro{GPTS}, effectively endowing \acro{PRE} with a small-step operational semantics. 

The definition of \acro{GPTS} uses probability (sub)distributions, so we define these first together with the notation that we will need for the derivative structure. 
%
\begin{tcolorbox}[breakable, enhanced,title=Subdistributions---definitions and operations,colback=Periwinkle!5!white,colframe=Periwinkle!75!black,]
 A function $\nu : X \to [0,1]$ is called a \emph{subprobability distribution} or \emph{subdistribution}, if it satisfies $\sum_{x \in X} \nu(x) \leq 1$. A subdistribution $\nu$ is \emph{finitely supported}  if the set $\supp(\nu) = \{x \in X \mid \nu(x) > 0\}$ is finite. We use $\distf {X}$ to denote the set of finitely supported subprobability distributions on $X$.
 \tcblower
Given $x \in X$, its \emph{Dirac} is a subdistribution $\delta_x$ which is given by $\delta_x(y)=1$ only if $x=y$, and $0$ otherwise. We will moreover write $0 \in \distf X$ for a subdistribution with an empty support. It is defined as $0(x)=0$ for all $x \in X$. When $\nu_1, \nu_2 : X \to [0,1]$ are subprobability distributions and $p \in [0,1]$, we write $p\nu_1 + (1-p)\nu_2$ for the convex combination of $\nu_1$ and $\nu_2$, which is the probability distribution given by $(p\nu_1 + (1-p)\nu_2)(x) = p\nu_1(x) + (1-p)\nu_2(x)$; this operation preserves finite support. 

Given a map $f : X \to \distf Y$, there exists a unique map $f^\star : \distf X \to \distf Y$ satisfying $f = f^\star \circ \delta$ called the \emph{convex extension of $f$}, and explicitly given by $f^\star(\nu)(y) = \sum_{x \in X} \nu(x)f(x)(y)$.
\end{tcolorbox}

A {\em generative probabilistic transition system} (\acro{GPTS})~\cite{Glabbeek:1995:Reactive} is a pair $(X,\beta)$ consisting of the set of states $X$ and a transition function $\beta \colon X \to \distf(1+A\times X)$, where $1 = \{\checkmark\}$ is a singleton set\footnote{Our definition of \acro{GPTS} slightly differs from~\cite{Glabbeek:1995:Reactive} as we do not explicitly include an initial state.}. The transition function $\beta$ encompasses the two types of transitions we informally described in \Cref{sec:overview}: $$\beta(x)(\checkmark) = t \iff x \xRightarrow[]{t} \checkmark \qquad\text{ and }\qquad\beta(x)(a,r) = s \iff q \xrightarrow[]{a \mid s} r$$

The probabilistic language generated by a state $x$ in a \acro{GPTS} $(X,\beta)$ is given by the function $\llbracket x \rrbracket \colon A^* \to [0,1]$ as defined in \Cref{language}.

\vspace{2px}
\noindent
\textbf{Antimirov derivatives.} 
We use \acro{GPTS} to equip \acro{PRE} with an operational semantics. More precisely, we will define a function $\partial : \Exp \to \distf (1+A\times \Exp)$. We refer to $\partial$ as the \emph{Antimirov derivative}, as it is reminiscent of the analogous construction
for regular expressions and nondeterminisic automata due to Antimirov~\cite{Antimirov:1996:Partial} building up on earlier work of Brzozowski~\cite{Brzozowski:1964:Derivatives}. Given $a \in A$, $e,f \in \Exp$ and $p \in [0,1]$ we define:
\begin{gather*}
    \partial(\zero) = 0 \quad\quad \partial(\one)=\delta_{\checkmark} \quad\quad \partial(a)=\delta_{(a,\one)} \\ \partial(e \oplus_p f)=p\partial(e) + (1-p)\partial(f)
\end{gather*}
The expression $\zero$ has no outgoing transitions, intuitively representing a deadlock and is therefore mapped to the empty support subdistribution. On the other hand, the expression $\one$ represents immediate acceptance, that is it transitions to $\checkmark$ with probability $1$. For any letter $a \in A$ in the alphabet, the expression $a$ performs $a$-labelled transition to $\one$ with probability $1$. The outgoing transitions of the probabilistic choice $e \oplus_p f$ consist of the outgoing transitions of $e$ with
probabilities scaled by $p$ and the outgoing transitions of $f$ scaled by $1-p$.

The definition of $\partial$ for sequential composition $e;f$ is slightly more involved. We need to factor in the possibility that $e$ may accept with some probability $t$, in which case the outgoing transitions of $f$ contribute to the outgoing transitions of $e \seq f$. Formally, $\partial(e \seq f) = \partial(e) \lhd f$ where for any $f \in \Exp$ the operation $( - \lhd f) : \distf (1+A\times \Exp) \to \distf (1+A\times \Exp)$ is given by $( - \lhd f) = {c_f}^{\star}$, the convex extension of $c_f : 1+A\times \Exp \to \distf (1+A \times \Exp)$ given below on the left.
\begin{gather*}
c_f(x) = \begin{cases}
    \partial(f) & x = \checkmark \\
    \delta_{(a, e' \seq f)} & x = (a,e')
\end{cases}\qquad\qquad
\begin{tikzpicture}[baseline=0ex]
        \node (0) {${e}\seq{f}$};
        \node (4) [right=1cm of 0] {${e}'{\color{Maroon} {} \seq {f}}$};
        \node (5) [left=0.5cm of 0] {{\color{gray} \bcancel\checkmark}};
            \node (5a) [left=-.2cm of 5] {{\color{Maroon} \(\partial({f})\)}};
        \draw (0) edge[-implies, double, double distance=0.5mm] node[above] {\({t}\)} (5);
        \draw (0) edge[-latex] node[above] {\footnotesize\( {a} \mid {s}\)} (4);
\end{tikzpicture}
\end{gather*}
Intuitively, $c_{f}$ reroutes the transitions coming out of ${e}$: acceptance (the first case) is replaced by the behaviour of ${f}$, and the probability mass of transitioning to ${e}'$ (the second case) is reassigned to $ e\seq f$.
A pictorial representation of the effect on the derivatives of ${e} \seq {f}$ is given above on the right.
Here, we assume that \(\partial({e})\) can perform a \({a}\)-transition to \({e'}\) with probability \({s}\); we make the same assumption in the informal descriptions of derivatives for the loops, below. 

For loops, we require $\partial\left(e^{[p]}\right)$ to be the least subdistribution satisfying $\partial\left(e^{[p]} \right) =  p \partial(e) \lhd e^{{[p]}} + (1-p) \partial(\checkmark)$. In the case when $\partial(e)(\checkmark)\neq 0$, the above becomes a fixpoint equation (as in such a case, the unrolling of the definition of $(- \lhd e^{[p]})$ involves $\partial \left( e{[p]}\right)$). We can define $\partial \left( e^{[p]} \right)$ in a neat, closed form, but we need to consider two cases. If $\partial(e)(\checkmark)=1$ and $p = 1$, then the loop body is constantly executed, but the inner expression $e$ does not perform any labelled transitions. We identify such divergent loops with deadlock behaviour and hence we set $\partial(e^{[p]})(x)=0$. Otherwise, we look at $\partial(e)$ to build $\partial\left({e}^{[p]}\right)$. First, we make sure that the loop may be skipped with probability $1-p$.
Next, we modify the branches that perform labelled transitions by adding ${e}^{[p]}$ to be executed next.
The remaining mass is $p\partial(e)(\checkmark)$, the probability that we will enter the loop and immediately exit it without performing any labelled transitions. We discard this possibility and redistribute it among the remaining branches. 
As before, we provide an informal visual depiction of the probabilistic loop semantics below, using the same conventions as before. The crossed-out checkmark along with the dashed lines denotes the redistribution of probability mass described above.
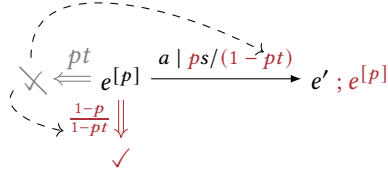
\[
	\begin{tikzpicture}
        \node (0) {${e}^{[p]}$};
        \node (6) [below=.5cm of 0, Maroon] {\checkmark};
        \node (4) [right=2cm of 0] {${e}'{\color{Maroon}{} \seq {e}^{[{p}]}}$};
        \node (5) [left=.5cm of 0] {{\color{gray} \(\bcancel{\checkmark}\)}};
        \draw (0) edge[-implies, double, double distance=0.5mm, gray, Maroon] node[left] {\footnotesize \(\frac{1-{p}}{1-{pt}}\)} (6);
        \draw (0) edge[-implies, double, double distance=0.5mm, gray, pos=0.3] node[above, gray] {\({pt}\)} (5);
        \draw (0) edge[-latex] node[above] {\footnotesize\( {a} \mid {\textcolor{Maroon}{{p}}}{s}/{\color{Maroon} (1-{pt})}\)} (4);
        \draw (5) edge[->,dashed,bend left, out=90] ($(4.west) + (-0.5,0.35)$);
        \draw (5) edge[->,dashed,bend right, out=-90] ($(6.west) + (-0.5,0.35)$);
    \end{tikzpicture}
\]
Formally speaking, the definition of $\partial\left({e}^{[{p}]}\right)$ can be given by the following:
\begin{gather*}
\partial\left(e^{[p]}\right)(x)= \begin{cases}
    \frac{1-p}{1-p\partial(e)(\checkmark)} & x = \checkmark \\[1.2ex]
    \frac{p\partial(e)(a,e')}{1-p\partial(e)(\checkmark)} & x = (a, (e' \seq e^{[p]}))\\
    0 & \text{otherwise}
\end{cases}
\end{gather*}

Having defined the Antimirov transition system, one can observe that the termination operator $E(-) : \Exp \to [0,1]$ precisely captures the probability of an expression transitioning to $\checkmark$ (successful termination) when viewed as a state in the Antimirov \acro{GPTS}.

\begin{restatable}{lemma}{exitoperatorlemma}\label{lem:exit_operator_lemma}
    For all $e \in \Exp$ it holds that $E(e)=\partial(e)(\checkmark)$.
    \end{restatable}

Given an expression $e \in \Exp$, we write $\langle e\rangle \subseteq \Exp$ for the set of states reachable from $e$ by repeatedly applying $\partial$. In \Cref{lem:locally_finite} in the appendix, we argue that this set is always finite - in other words, the operational semantics of every \acro{PRE} can be always described by a finite-state \acro{GPTS} given by $(\langle e \rangle, \partial)$.
\begin{example} Operational semantics of the expression $e = a \oplus_{\frac{3}{4}} a \seq a^{[\frac{1}{4}]}\seq a$ correspond to the following \acro{GPTS}: 
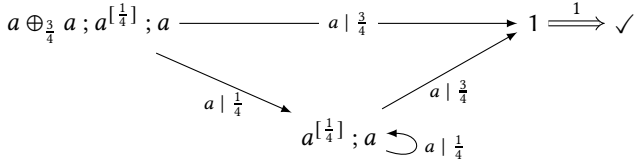
\begin{center}
\begin{tikzpicture}[baseline=0ex]
        \node (2) {$a \oplus_{\frac{3}{4}} a \seq a^{[\frac{1}{4}]}\seq a$};
        \node (2a) [ right= 2 cm of 2]{};
        \node (3) [below= 1 cm of 2a] {$a^{[\frac{1}{4}]}\seq a$};
        \draw (2) edge[-latex] node[below] {\scriptsize\( {a} \mid {\frac{1}{4}}\)} (3);
        \draw (3) edge[-latex,out=-20,in=0,looseness=6] node[right] {\scriptsize\( {a} \mid {\frac{1}{4}}\)} (3);
        \node (4) [right= 4.5 cm of 2] {$\one$};
        \draw (3) edge[-latex] node[below] {\scriptsize\( {a} \mid {\frac{3}{4}}\)} (4);
        \draw (2) edge[-latex] node[ fill=white] {\scriptsize\( {a} \mid {\frac{3}{4}}\)} (4);
        \node (o2) [right= 0.75 cm of 4]{$\checkmark$};
        \draw (4) edge[-implies, double, double distance=0.5mm] node[above] {\scriptsize\( 1\)} (o2);
\end{tikzpicture}
\end{center}
One can observe that the transition system above for $e$ is isomorphic to the one starting in $q_2$ in \cref{ex:systems}.
\end{example}

Given the finite-state \acro{GPTS} $(\langle e \rangle, \partial)$ associated with an expression $e\in \Exp$ we can define the language semantics of $e$ as the probabilistic language $\llbracket e\rrbracket \in [0,1]^{A^*}$ generated by the state $e$ in the \acro{GPTS} $(\langle e \rangle, \partial)$. 

\section{Roadmap to (Sound+Complete)ness}\label{sec:roadmap}
We want to show that the axioms in \Cref{fig:axioms} are sound and complete to reason about probabilistic language equivalence of \acro{PRE}, that is:
\[
e \equiv f \quad \begin{array}{c} {\text{\tiny Completeness}} \\ \Longleftarrow\\ \Longrightarrow\\ {\text{\tiny Soundness}} \end{array} \llbracket e\rrbracket = \llbracket f\rrbracket
\]
We now sketch the roadmap on how we will prove these two results to ease the flow into the upcoming technical sections. Perhaps not surprisingly, the completeness direction is the most involved. 

\smallskip
\noindent\textbf{High-level overview.} 
The heart of both arguments will rely on arguing that the semantics map $\llbracket - \rrbracket : \Exp \to [0,1]^{A^*}$ assigning a probabilistic language to each expression can be seen as the following composition of maps:
\[\begin{tikzcd}
	\Exp && {{\Exp}/{\equiv}} && {[0,1]^{A^*}}
	\arrow["{[-]}", from=1-1, to=1-3]
	\arrow["{\dagger d}", from=1-3, to=1-5]
	\arrow["{\llbracket-\rrbracket}"', curve={height=18pt}, from=1-1, to=1-5]
\end{tikzcd}\] 
In the picture above $[-] \colon \Exp \to {\Exp}/{\equiv}$ is a quotient map taking expressions to their equivalence class modulo the axioms of $\equiv$, while $\dagger d : {\Exp}/{\equiv} \to [0,1]^{A^*}$ is a map taking each equivalence class to the corresponding probabilistic language. In such a case, soundness follows as a sequence of three steps:
\begin{align}\label{eq:sound}
\begin{split}
		e&\equiv f\\\Rightarrow [e]&=[f]\\
	 \Rightarrow \dagger d ([e])  &= \dagger d ( [f] )\\
	 \Rightarrow \llbracket e\rrbracket &= \llbracket f\rrbracket
\end{split}
\end{align} 
 To obtain completeness we want to reverse all implications in \Cref{eq:sound}--and they all are easily reversible except $[e]=[f] \Rightarrow \dagger d ( [e] )  =  \dagger d ( [f] )$. To obtain this reverse implication we will need to show that $\dagger d $ is {\em injective}.

\vspace{2px}
\noindent
\textbf{Language semantics via determinisation.} 
Both arguments outlined above will crucially rely on the transition system structure on the set of expressions provided by the Antimirov construction and hence it will be convenient to study it through the theory of \emph{universal coalgebra}~\cite{Rutten:2000:Universal,Gumm:2000:Elements}. In general, coalgebras provide a uniform and abstract treatment of transition systems and their semantics. The theory allows us to derive the notion of coalgebra \emph{homomorphisms}, which are transition-preserving maps between transition systems. Under mild conditions, one can also derive the notion of a \emph{final coalgebra}, which provides a form of canonical domain allowing to assign meaning to states of transition systems, through unique coalgebra homomorphisms. Much of the theory can be illustrated through an example of deterministic automata, where the final coalgebra is given by the automaton structure over the set of all formal languages, namely $2^{A^*}$. In such a case, the unique final homomorphism is given by the language-assigning map taking each state of an automaton to its language. 

Similarly to deterministic automata, \acro{GPTS} can be seen as coalgebras, which also admit a final coalgebra. Unfortunately, in this case, the final coalgebra is \textbf{not} carried by the set of probabilistic languages, that is $[0,1]^{A^{*}}$, because the canonical semantics of \acro{GPTS} happens to correspond to the more restrictive notion of probabilistic bisimilarity (also known as Larsen-Skou bisimilarity~\cite{Larsen:1991:Bisimulation}). Probabilistic bisimilarity is a branching-time notion of equivalence, requiring observable behaviour of compared states to be equivalent at every step, while probabilistic language equivalence is a more liberal notion comparing sequences of observable behaviour. In general, if two states are bisimilar, then they are language equivalent, but the converse does not hold.
\begin{example}\label{ex:bisimilarity}
 Consider the following \acro{GPTS}:
 \begin{gather*}
\begin{tikzpicture}[baseline=-3ex]
        \node (0) {$q_0$};
        \node (1) [right= 1.25 cm of 0]{$q_1$};
        \node (2) [right= 1.25 cm of 1]{$q_2$};
        \node (o1) [right= 0.75 cm of 2]{$\checkmark$};
        \node (3) [right = 1.8 cm of o1]{$q_3$};
        \node (4) [right= 1.25 cm of 3]{$q_4$};
        \node (5) [right= 1.25 cm of 4]{$q_5$};
        \node (o2) [right= 0.75 cm of 5]{$\checkmark$};
        \draw (0) edge[-latex] node[ fill=white] {\scriptsize\( {a} \mid {\frac{2}{3}}\)} (1);
		\draw (1) edge[-latex] node[ fill=white] {\scriptsize\( {b} \mid {\frac{1}{2}}\)} (2);
		\draw (2) edge[-implies, double, double distance=0.5mm] node[above] {\scriptsize\( 1\)} (o1);
		\draw (3) edge[-latex] node[ fill=white] {\scriptsize\( {a} \mid {\frac{1}{2}}\)} (4);
		\draw (4) edge[-latex] node[ fill=white] {\scriptsize\( {b} \mid {\frac{2}{3}}\)} (5);
		\draw (5) edge[-implies, double, double distance=0.5mm] node[above] {\scriptsize\( 1\)} (o2);
\end{tikzpicture}
\end{gather*}
States $q_0$ and $q_3$ are language equivalent because they both accept the string $ab$ with the probability $\frac{1}{3}$, but are not bisimilar, because the state $q_0$ can make $a$ transition with the probability $\frac{2}{3}$, while $q_3$ can perform an $a$ transition with probability $\frac{1}{2}$.
\end{example}
A similar situation happens when looking at nondeterministic automata through the lenses of universal coalgebra, where again the canonical notion of equivalence is the one of bisimilarity. A known remedy is the powerset construction from classic automata theory, which converts a nondeterministic automaton to a deterministic automaton, whose states are sets of states of the original nondeterministic automaton we have started from. In such a case, the nondeterministic branching structure is factored into the state space of the determinised automaton. The language of an arbitrary state of the nondeterministic automaton corresponds to the language of the singleton set containing that state in the determinised automaton.

There is a coalgebraic generalisation of this construction, allowing to canonically obtain the language semantics of analogues of nondeterministic automata, including the weighted and probabilistic generalisations~\cite{Silva:2010:Generalizing}. Unfortunately, \acro{GPTS} do not fit into this theory. Luckily, each \acro{GPTS} can be seen as a special case of a more general kind of transition system, known as reactive probabilistic transition systems~(\acro{RPTS})~\cite{Glabbeek:1995:Reactive} or Rabin probabilistic automata~\cite{Rabin:1963:Probabilistic}. \acro{RPTS} can be intuitively viewed as a probabilistic counterpart of nondeterministic automata and they can be determinised to obtain probabilistic language semantics. In an \acro{RPTS}, each state $x$ is mapped to a pair $(o_x,n_x)$, where $o \in [0,1]$ is the acceptance probability of state $x$ and $n_x\colon A \to \distf(X)$ is the next-state function, which takes a letter $a \in A$ and returns the subprobability distribution over successor states. So, formally an \acro{RPTS} is a pair $(X,\left<o,n\right>)$ where $X$ is the set of states, $o\colon X\to [0,1]$, and $n\colon X \to  \distf(X)^A$. 
Every \acro{GPTS} can be easily seen as \acro{RPTS}, as shown in the figure below.
\begin{example}\label{ex:converting}
Fragment of a \acro{GPTS} (on the left) and of the corresponding RPTS (on the right).
 
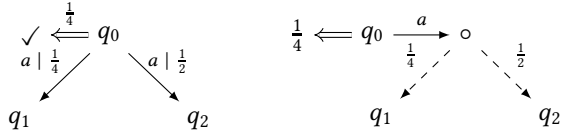
\begin{gather*}
\begin{tikzpicture}[baseline=-3ex]
        \node (o1) {$\checkmark$};
        \node (0) [right= 0.5 cm of o1]{$q_0$};
        \draw (0) edge[-implies, double, double distance=0.5mm] node[above] {\scriptsize\(\frac{1}{4}\)} (o1);
        \node (1) [below left= .9 cm of 0]{$q_1$};
        \node (2) [below right= .9 cm of 0]{$q_2$};
        \draw (0) edge[-latex] node[above, pos=0.6] {\scriptsize\( {a} \mid {\frac{1}{4}}\qquad  \)} (1);
        \draw (0) edge[-latex] node[above, pos=0.6] {\scriptsize\(\quad {a} \mid {\frac{1}{2}}\)} (2);
        \node (o2) [right= 2cm of 0]{$\frac{1}{4}$};
        \node (3) [right= 0.5 cm of o2]{$q_0$};
        \draw (3) edge[-implies, double, double distance=0.5mm] (o2);
        \node (4) [right= 0.75 cm of 3]{$\circ$};
        \draw (3) edge[-latex] node[above] {\scriptsize\( {a}\)} (4);
        \node (5) [below left= 0.9 cm of 4]{$q_1$};
        \node (6) [below right= 0.9 cm of 4]{$q_2$};
        \draw (4) edge[-latex, dashed] node[left, pos = 0.1] {\scriptsize\( {\frac{1}{4}}\quad\)} (5);
        \draw (4) edge[-latex, dashed] node[right, pos = 0.1] {\scriptsize\(\quad {\frac{1}{2}}\)} (6);
\end{tikzpicture}
\end{gather*}
In the corresponding \acro{RPTS} state $q_0$ accepts with probability $\frac{1}{4}$ and given input $a$ it transitions to subprobability distribution that has $\frac{1}{4}$ probability of going to to $q_1$ and $\frac{1}{2}$ probability of going to $q_2$. \end{example}
Determinising \acro{RPTS} yields fuzzy analogues of deterministic automata, where each state has an associated acceptance probability. Viewing those deterministic transition systems through the theory of coalgebra allows uniquely assigning their states to probabilistic languages (elements of $[0,1]^{A^*}$) through unique coalgebra homomorphisms. Additionally, state spaces of determinisations of \acro{RPTS} carry an additional structure of positive convex algebras, in the same way as state spaces of deterministic automata obtained via powerset construction can be seen as semilattices. The set of all probabilistic languages also carries an algebra structure and the unique language-assigning coalgebra homomorphism, happens to be an algebra homomorphism.  This additional algebraic structure will play a crucial role in our soundness and completeness arguments. The language semantics of a state of \acro{GPTS} that we gave in \Cref{sec:overview}, can be equivalently seen as determinising the corresponding \acro{RPTS} and asking about the probabilistic language corresponding to the Dirac distribution of the state we are interested in. In particular, we will apply this point of view to the Antimirov automaton, which is used to give operational semantics to the expressions. We formally phrase the construction of language semantics of \acro{GPTS} via generalised determinisation in \Cref{sec:language_semantics}.
	
\vspace{2px}
\noindent
\textbf{Roadmap to Soundness.} 
The proof of soundness will crucially rely on viewing the language semantics of \acro{PRE} through generalised determinisation construction. Thanks to that point of view, the language $\llbracket e \rrbracket \in [0,1]^{A^*}$ of an arbitrary expression $e \in \Exp$, can be alternatively seen as the language of the state $\delta_e$ in the determinisation of the Antimirov transition system (after being seen as a more general \acro{RPTS}). We will show that the quotient ${\Exp}/{\equiv}$ can be equipped with a structure of deterministic transition system, such that there is language preserving map (deterministic transition system homomorphism) $h : \distf {\Exp} \to {\Exp}/{\equiv}$ from the determinisation of the Antimirov \acro{GPTS}, additionally satisfying that $h(\delta_e) = [e] \in {\Exp}/{\equiv}$. This will allow us to conclude that the language of the state $\delta_e$ in the determinisation of the Antimirov \acro{GPTS} is the same as the one of the state $[e]$ in the deterministic transition structure on ${\Exp}/{\equiv}$, thus establishing soundness. Equivalently, if we write $\dagger d : {\Exp}/{\equiv} \to [0,1]^{A^*}$ for the unique map taking states of that deterministic transition system to corresponding probabilistic languages, then this will allow us to prove that $\llbracket - \rrbracket = \dagger d \circ [-]$, as outlined at the beginning of this section. 

As much as our proof of soundness is not a straightforward inductive argument like in ordinary regular expressions, it immediately sets up the stage for the completeness argument. Essentially, we will need to show that $\dagger d$ is injective, by showing that the deterministic transition system structure on ${\Exp}/{\equiv}$ satisfies a certain universal property. In general, obtaining the appropriate transition system structure on $\Exp/_\equiv$ needs a couple of intermediate steps, which then lead to soundness:
\begin{enumerate}
\item We first prove the soundness of a subset of the axioms of \Cref{fig:axioms}: omitting \textbf{S0} and \textbf{D2} yields a sound inference system, which we call $\equiv_0$, wrt a finer equivalence--probabilistic bisimilarity as defined by Larsen and Skou~\cite{Larsen:1991:Bisimulation} (\Cref{lem:soundness_bisim}). As a consequence, there exists a deterministic transition system structure on the set $\distf{{\Exp}/{\equiv_0}}$, such that $\distf [-]_{\equiv_0} : \distf \Exp \to \distf {{\Exp}/{\equiv_0}}$ is a transition system homomorphism.
\item We then prove that the set of expressions modulo the bisimilarity axioms, that is $\Exp/_{\equiv_0}$, has the structure of a positive convex algebra $\alpha_{\equiv_0}: \distf{{\Exp}/{\equiv_0}} \to {\Exp}/{\equiv_0}$. This allows us to flatten a distribution over equivalence classes into a single equivalence class. This proof makes use of a {\em fundamental theorem} decomposing expressions (\Cref{thm:fundamental_theorem}).
\item With the above result, we equip the set $\Exp/{\equiv_0}$ with a structure of deterministic transition system and show that the positive convex algebra structure map is also a transition system homomorphism from $\distf {\Exp}/{\equiv_0}$. 
\item Through a simple inductive argument (this step encapsulates the key part of the soundness argument), we show that there exists a unique deterministic transition system structure on the coarser quotient, that is ${\Exp}/{\equiv}$, that makes further identification using axioms \textbf{S0} and \textbf{D2} (denoted $[-]_{\equiv} : {\Exp}/{\equiv_0} \to {\Exp}/{\equiv}$) into a homomorphism. We compose all homomorphisms into $h : \distf {\Exp} \to {\Exp}/{\equiv}$ and show the correspondence of the probabilistic language of the state $[e]$ in the above transition system with the one of $\delta_{e}$ in the determinisation of Antimirov \acro{GPTS}, thus establishing soundness.
\end{enumerate}

\vspace{2px}
\noindent
\textbf{Roadmap to Completeness.} 
As mentioned above, we will show that the deterministic transition system on the quotient ${\Exp}/{\equiv}$ possesses a certain universal property allowing us to conclude that the unique language-assigning homomorphism $\dagger d : {\Exp}/{\equiv} \to [0,1]^{A^*}$ is injective. The universal property we care about is the one of \emph{rational fixpoint}~\cite{Milius:2010:Sound}, which is a coalgebraic generalisation of the concept of regular languages providing semantics to finite-state deterministic automata. Unfortunately, as we will see in \Cref{sec:completeness}, determinising a finite-state \acro{GPTS} can lead to an infinite state deterministic transition system. However, each determinisation of a finite-state \acro{GPTS} carries a structure of a positive convex algebra, that is \emph{free finitely generated}. Thanks to the work of Milius~\cite{Milius:2018:Proper} and Sokolova \& Woracek~\cite{Sokolova:2018:Proper}, we will see that establishing that ${\Exp}/{\equiv}$ is isomorphic to the rational fixpoint boils down to showing uniqueness of homomorphisms from determinisations of finite-state \acro{GPTS}. We will reduce this problem to converting \acro{GPTS} to language equivalent expressions through the means of axiomatic manipulation using a procedure reminiscent of Brzozowski's equation solving method~\cite{Brzozowski:1964:Derivatives} for converting DFAs to regular expressions. As a corollary, we will obtain an analogue of (one direction of) Kleene's theorem for \acro{GPTS} and \acro{PRE}. To sum up, the completeness result is obtained in $4$ steps:
\begin{enumerate}
\item We show that the deterministic transition system on the quotient ${\Exp}/{\equiv}$ has an additional algebraic structure, which allows us to rely on the theory of the rational fixpoint.
\item Following a traditional pattern in completeness proofs~\cite{Salomaa:1966:Two,Backhouse:1976:Closure,Milner:1984:Complete}, we represent  \acro{GPTS} as \emph{left-affine} systems of equations within the calculus and show that these systems have {\em unique} solutions up to provable equivalence (\Cref{thm:unique_solutions}).
\item We then show that these solutions are in 1-1 correspondence with well-behaved maps from algebraically structured transition systems obtained from finite-state \acro{GPTS} into the deterministic transition system on ${\Exp}/{\equiv}$ (\Cref{lem:determinisations_correspond_to_df_coalgebras}).
\item Finally, we use this correspondence together with an abstract categorical argument to show that a transition system structure on ${\Exp}/{\equiv}$ has a universal property (\Cref{cor:rational}) that implies injectivity of  $\llbracket-\rrbracket$.
\end{enumerate}

\section{(Co)algebras and Determinisations}\label{sec:language_semantics}

This section contains a number of technical definitions and results that we need to execute the roadmap to soundness described above, crucially relying on the theory of {\em universal coalgebra}. We assume the familiarity of the reader with the basic notions of category theory, such as endofunctors, monads, and natural transformations; for a detailed introduction, we refer to~\cite{Abramsky:2010:Introduction}. Below, we use the fact that subdistributions $\distf$ are in fact a functor on the category $\Set$ of sets and functions: $\distf$ maps each set $X$ to the set to $\distf X$ and maps each arrow $f : X \to Y$ to the function $\distf f : \distf X \to \distf Y$ given by $$\distf f (\nu)(x) = \sum_{y \in f^{-1}(x)} \nu(y)$$ Moreover, $\distf$ also  carries a monad structure with unit  $\eta_X(x) = \delta_x$ and multiplication given by $\mu_X(\Phi)(x) = \sum_{\varphi \in \distf X }\Phi(\varphi)\varphi(x)$ for $\Phi \in \distf \distf X$.

\begin{tcolorbox}[breakable, enhanced,title=Coalgebras---definitions and notation,colback=Periwinkle!5!white,colframe=Periwinkle!75!black,]
A coalgebra for an endofunctor $F : \A \to \A$ is a pair $(X,\beta)$, where $X$ is an object in the category $\A$ representing the state space of the coalgebra and $\beta : X \to F X$ is a morphism representing the transition structure. A map $f : X \to Y$ is a homomorphisms between coalgebras $(X, \beta)$ and $(Y, \gamma)$ if $ \gamma \circ f = F f \circ \beta$. We use $\coa {F}$ to denote the category of $F$-coalgebras and their homomorphisms.

\tcblower
An $F$-coalgebra $(\nu F,t)$ is final if for any $F$-coalgebra $(X, \beta)$ there exists a unique homomorphism $(X,\beta) \to (\nu F,t)$. We write $\dagger c : X \to \nu F$ for the unique homomorphism into the final coalgebra (if it exists). Sometimes, we will abuse the notation and write $\nu F$ instead of $(\nu F , t)$. 
\end{tcolorbox}

\acro{GPTS} can be viewed as coalgebras for the composite functor $\distf F : \Set \to \Set$, where $F = {1} +(-) \times A$ and  \acro{RPTS} can be modelled as $G\distf$-coalgebras, where $G : \Set \to \Set$ is the functor given by $G = [0,1] \times (-)^{A}$. The functor $G$ has a final coalgebra with carrier set $[0,1]^{A^*}$ -- this characterises probabilistic languages as a canonical object and every $G$-coalgebra can be mapped into this final one (that is, every state in a $G$-coalgebra can be assigned a probabilistic language in a unique way). \acro{RPTS} are $G\distf$-coalgebras and in order to exploit the finality of $G$ we need to describe a transformation that maps $G\distf$-coalgebras to $G$-coalgebras (with a richer state space in a category of algebras for a monad)--this will be an instance of a \emph{generalised determinisation} construction~\cite{Silva:2010:Generalizing}, which we describe below. 
\begin{tcolorbox}[breakable, enhanced,title=Algebras for a monad---definitions and notation,colback=Periwinkle!5!white,colframe=Periwinkle!75!black,]
Let $(T, \eta^T, \mu^T)$ be a monad on $\Set$. The category of Eilenberg-Moore algebras for $T$, denoted $\Set^T$, consists of $T$-algebras $(X,\alpha : T X \to X)$ such that $\alpha \circ \eta^T_X = \id_X$ and $\alpha \circ T\alpha = \alpha \circ \mu^T_X$ and their homomorphisms. A 
homomorphism between two Eilenberg-Moore algebras $(X, \alpha)$ and $(Y, \beta)$ is a map $f : X \to Y$ satisfying $\beta \o T f = f \o \alpha $.\tcblower
 For every set $X$, we call the algebra $(T X, \mu^T_X)$
\emph{free}. If $X$ is finite, we call $(T X, \mu^T_X)$ \emph{free finitely generated}. Free algebras satisfy the universal property that given any map $f : X \to Y$, such that $(Y, \beta)$ is an Eilenberg-Moore algebra, there is a unique $T$-algebra homomorphism $f^\star : (TX, \mu^T_X) \to (Y, \beta)$ satisfying $f = f^\star \circ \eta^T_X$.
\end{tcolorbox}

For concrete monads, Eilenberg-Moore categories have familiar algebraic descriptions: for instance, for the powerset monad $\mathcal P$ it is isomorphic to join semilattices and for the $\distf$ monad it is isomorphic to positive convex algebras. 

\begin{tcolorbox}[breakable, enhanced,title=Positive Convex Algebras---definitions and notation,colback=Periwinkle!5!white,colframe=Periwinkle!75!black,]
A \emph{positive convex algebra} is an algebra with an abstract convex sum operation $\bigboxplus_{i \in I} p_i \cdot (-)_i$ for $I$ finite, $p_i \in [0,1]$ and $\sum_{i \in I} p_i \leq 1$ satisfying:
\begin{enumerate}
    \item (Projection) \(\bigboxplus_{i \in I} p_i \cdot x_i = x_j\) if \(p_j=1\)
    \item (Barycenter) \(\bigboxplus_{i \in I} p_i \cdot \left(\bigboxplus_{j \in J}{q_{i,j}} \cdot {x_j}\right) = \bigboxplus_{j \in J} \left(\sum_{i \in I} p_i q_{i,j} \right) \cdot x_j\)
\end{enumerate}
\tcblower
Positive convex algebras and their homomorphisms form a category $\pca$. This category is isomorphic to $\Set^\distf$~\cite{Jacobs:2010:Convexity,Doberkat:2008:Erratum}. In terms of notation, we denote the unary sum by $p_0 \cdot x_0$. Throughout the paper we will we abuse the notation by writing $\bigboxplus_{i \in I} p_i \cdot e_i \boxplus \bigboxplus_{i \in J} q_j \cdot f_j$ for a single convex sum $\bigboxplus_{k \in I + J} r_k \cdot g_k$, where $r_k = p_k$ and $g_k = e_k$ for $k \in I$ and similarly $r_k = q_k$ and $g_k = f_k$ for $k \in J$. This is well-defined only if $\sum_{i \in I} p_i + \sum_{j \in J} r_j \leq 1$.
\end{tcolorbox}

\vspace{2px}
\noindent
\textbf{Generalised determinisation.} Language acceptance of nondeterministic automata (\acro{NDA}) can be captured via determinisation. \acro{NDA} can be viewed as coalgebras for the functor $N = 2 \times {\powf (-)}^A$, where $\powf$ is the finite powerset monad. Determinisation converts a \acro{NDA} $(X, \beta : X \to 2 \times {\powf X}^A)$ into a deterministic automaton $({\powf X}, \beta^{\star} : {\powf X} \to 2 \times {\powf X}^A )$, where $\beta^{\star}$ is the unique extension of the map $\beta$. A language of the state $x \in X$ of \acro{NDA}, is given by the language accepted by the state $\{x\}$ in the determinised automaton $({\powf X}, \beta^\star)$. 

This construction can be generalised~\cite{Silva:2010:Generalizing} to $HT$-coalgebras, where $T : \Set \to \Set$ is a finitary monad and $H : \Set \to \Set$ an endofunctor that admits a final coalgebra that can be lifted to the functor $\ol{H} : \Set^T \to \Set^T$. Generalised determinisation turns $HT$-coalgebras $(X, \beta : X \to HT X)$ into $H$-coalgebras $(T X , \beta^\star : T X \to H T X)$. The language of a state $x \in X$ is given by $\dagger \beta^\star \o \eta^T_X : X \to \nu H$, where $\eta^T$ is the unit of the monad $T$. Each determinisation $(TX, \beta^{\star})$ can be viewed as an $\ol{H}$-coalgebra $((TX, \mu^T_X), \beta^{\star})$. The carrier of the final $H$-coalgebra can be canonically equipped with $T$-algebra structure, yielding the final $\ol{H}$-coalgebra. In such a case, the unique final homomorphism from any determinisation viewed as an $H$-coalgebra is also the final $\ol{H}$-coalgebra homomorphism. 

 \acro{RPTS} fit into the framework of generalised determinisation, as there exists a natural transformation $\lambda : \distf G \to G \distf$ that allows to lift $G$ to an endofunctor on $\pca$~\cite{Silva:2011:Sound}. We will abuse notation and write $G : \pca \to \pca$ for that lifting. Because of that, we make no notational distinction between the final coalgebras (and thus the final homomorphisms) of $G : \Set \to \Set$ and $G : \pca \to \pca$. 

\vspace{2px}
\noindent
\noindent \textbf{Language semantics of \acro{GPTS} as a map into a final coalgebra.}
Every \acro{GPTS} can be transformed into a \acro{RPTS}: in particular, there exists a natural transformation $\gamma : \distf F \to G \distf$ with monomorphic components~\cite{Silva:2011:Sound} given by:
$$
\gamma_X (\zeta) = \langle\zeta(\checkmark), \lambda a . \lambda x. \zeta (a,x)\rangle \qquad\qquad\qquad (\zeta \in \distf F X)
$$
Using that one can convert each \acro{GPTS} $(X, \beta)$ into the corresponding \acro{RPTS} $(X, \gamma_X \circ \beta)$. It is worth pointing out, that not every reactive system can be obtained in that way. For example, any \acro{RPTS} having a state accepting with probability one and at least one letter for which the successor distribution has nonempty support would require the source \acro{GPTS} to have the total probability mass greater than one. 

We are now ready to state the full construction of the probabilistic language semantics of \acro{GPTS}.  Let $(X, \beta : X \to \distf F X)$ be a \acro{GPTS}. We can convert it into a reactive system $(X, \gamma_X \o \beta)$ and then determinise it into a $G$-coalgebra $(\distf X, (\gamma_X \o \beta)^\star)$. We will abuse the terminology and sometimes call that coalgebra a determinisation of \acro{GPTS} $(X, \beta)$. The language semantics given by generalised determinisation yields a map $\dagger (\gamma_X \o \beta)^\star \o \eta_X : X \to {A^*} \to [0,1]$. This map coincides with the explicit definition of $\llbracket - \rrbracket$ we gave in \Cref{language} (this is a consequence of a result in~\cite{Silva:2011:Sound}). 
Throughout this paper, we will be mainly interested in the probabilistic language semantics of one particular \acro{GPTS}: the Antimirov transition system, which is used to give operational semantics of \acro{PRE}. 

\section{Soundness}\label{sec:soundness}
We are now ready to execute the roadmap to soundness described in \Cref{sec:roadmap}. 
\vspace{2px}
\newline
\textbf{Step 1: Soundness wrt bisimilarity.}
Let ${\equiv_0} \subseteq {\Exp \times \Exp}$ denote the least congruence relation closed under the axioms on \cref{fig:axioms} except \textbf{S0} and \textbf{D2}. 
A straightforward induction on the length derivation of $\equiv_0$ allows us to show that this relation is a bisimulation equivalence on $(\Exp, \partial)$--- an equivalence relation that is a carrier of a coalgebra for which the canonical projection maps $\pi_1(e,f)=e$ and $\pi_2(e,f)=f$ are coalgebra homomorphisms to the Antimirov transition system. In the case of \acro{GPTS} this notion corresponds to bisimulation equivalences in the sense of Larsen and Skou~\cite{Larsen:1991:Bisimulation}.
\begin{restatable}{lemma}{soundnessbisim}\label{lem:soundness_bisim}
    The relation ${\equiv_0} \subseteq {\Exp \times \Exp}$ is a bisimulation equivalence.
\end{restatable}
As a consequence of \Cref{lem:soundness_bisim}~\cite{Rutten:2000:Universal}, there exists a unique coalgebra structure $\ol{\partial} : {\Exp}/{\equiv_0} \to \distf F {\Exp}/{\equiv_0}$, which makes the quotient map $[-]_{\equiv_0} : \Exp \to {\Exp}/{\equiv_0}$ into a $\distf F$-coalgebra homomorphism from $(\Exp, \partial)$ to $({{\Exp}/{\equiv_0}}, \ol{\partial})$. It turns out, that upon converting those $\distf F$-coalgebras to $G \distf$-coalgebras using the natural transformation $\gamma : \distf F \to G \distf$ and determinising them, $\distf [-]_{\equiv_0} : \distf \Exp \to \distf {\Exp}/{\equiv_0}$ becomes a homomorphism between the determinisations. This situation can be summed up by the following commutative diagram:
\[\begin{tikzcd}
	& {\distf \Exp} && {\distf {{\Exp}/{\equiv_0}}} \\
	\Exp && {\Exp/{\equiv_0}} \\
	{\distf F \Exp} && {\distf F {\Exp}/{\equiv_0}} \\
	{G\distf \Exp} && {G\distf {\Exp}/{\equiv_0}}
	\arrow["\partial"', from=2-1, to=3-1]
	\arrow["{\gamma_{\Exp}}"', from=3-1, to=4-1]
	\arrow["{\eta_{\Exp}}"{pos=0.4}, from=2-1, to=1-2]
	\arrow["{\eta_{{\Exp}/{\equiv_0}}}"{pos=0.1}, from=2-3, to=1-4]
	\arrow["{\distf [-]_{\equiv_0}}", from=1-2, to=1-4]
	\arrow["{[-]_{\equiv_0}}"', from=2-1, to=2-3]
	\arrow["{\ol {\partial}}"', from=2-3, to=3-3]
	\arrow["{\gamma_{{\Exp}/{\equiv_0}}}"', from=3-3, to=4-3]
	\arrow["{G \distf [-]_{\equiv_0}}", from=4-1, to=4-3]
	\arrow["{(\gamma_{\Exp} \o \partial)^\star}", from=1-2, to=4-1]
	\arrow["{(\gamma_{{\Exp}/{\equiv_0}} \o \ol{\partial})^\star}", from=1-4, to=4-3]
\end{tikzcd}\]

\vspace{2px}
\noindent
\textbf{Step 2a: Fundamental theorem.} We show that every \acro{PRE} is provably equivalent to a decomposition involving sub-expressions obtained in its small-step semantics. This property, often referred to as the fundamental theorem (in analogy with the fundamental theorem of calculus) is useful in proving soundness. 

\begin{restatable}{theorem}{fundamentaltheorem}\label{thm:fundamental_theorem}
    For all $e \in \Exp$ we have that 
    $$
    e \equiv_0 \bigoplus_{d \in \supp(\partial(e))} \partial(e)(d) \cdot \ex(d)
    $$
    where $\ex: F\Exp \to \Exp$ is the map $\exp(\checkmark)=\one$ and $\exp{(a,e')}=a\seq e'$.
\end{restatable}
The definition above makes use of an $n$-ary probabilistic choice defined in terms of the binary one. In the appendix (\Cref{prop:binary}), we give the full definition and show that this $n$-ary choice operator obeys the axioms of positive convex algebras. 


\noindent
\textbf{Step 2b: Algebra structure.} 
As a consequence of the fundamental theorem, one can show that $\ol{\partial} : {\Exp}/{\equiv_0} \to \distf F {\Exp}/{\equiv_0}$ is an isomorphism -- see \Cref{cor:quotient_coalgebra_inverse}. This allows to define a map $\alpha_{\equiv_0} : \distf {\Exp}/{\equiv_0} \to {\Exp}/{\equiv_0}$ as the following composition of morphisms:
\[
\begin{tikzcd}
	{\distf{\Exp}/{\equiv_0}} & {\distf\distf F{\Exp}/{\equiv_0}} & {\distf F{\Exp}/{\equiv_0}} & {{\Exp}/{\equiv_0}}
	\arrow["{\distf\ol{\partial}}", from=1-1, to=1-2]
	\arrow["\mu_{F {\Exp}/{\equiv}}", from=1-2, to=1-3]
	\arrow["{{\ol{\partial}}^{-1}}", from=1-3, to=1-4]
\end{tikzcd}\]
It turns out that $({\Exp}/{\equiv_0}, \alpha_{\equiv_0})$ is an Eilenberg-Moore algebra for the subdistribution monad -- see \Cref{lem:quotient_is_an_em_algebra}. Moreover, using the isomorphism between $\pca$ and $\Set^\distf$ one can calculate the concrete formula for $\pca$ structure on ${\Exp}/{\equiv_0}$:
$$
\bigboxplus_{i \in I} p_i \cdot [e_i]_{\equiv_0} = \left[\bigoplus_{i \in I} p_i \cdot e_i \right]_{\equiv_0}
$$

\noindent
\textbf{Step 3: Coalgebra structure.}
Having established the necessary algebraic structure, we move on to showing how we can equip the quotient ${\Exp}/{\equiv}$ with a structure of coalgebra for the functor $G : \Set \to \Set$. First, we focus on the $G$-coalgebra structure on the finer quotient $\Exp/{\equiv_0}$, given by the following composition of maps:
\[\begin{tikzcd}
	{{\Exp}/{\equiv_0}} & {\distf F{\Exp}/{\equiv_0}} & {G \distf {\Exp}/{\equiv_0}} & {G {\Exp}/{\equiv_0}}
	\arrow["{\ol{\partial}}", from=1-1, to=1-2]
	\arrow["\gamma_{{\Exp}/{\equiv}}", from=1-2, to=1-3]
	\arrow["G \alpha_{\equiv_0}", from=1-3, to=1-4]
\end{tikzcd}\]
Intuitively, we take a $\distf F$-coalgebra structure on the quotient ${\Exp}/{\equiv_0}$, which exists as a consequence of soundness of ${\equiv_0}$ wrt. bisimilarity and then make it a $G \distf$-coalgebra using the natural transformation $\gamma : \distf F \to G \distf$. Instead of determinising it, we flatten each reachable subdistribution using the algebra map $\alpha_{\equiv_0} : \distf F {\Exp}/{\equiv_0} \to {\Exp}/{\equiv_0}$, thus obtaining a $G$-coalgebra structure. It turns out that this coalgebra is inherently related to the determinisation of $\left({{\Exp}/{\equiv_0}},{\gamma_{{\Exp}/{\equiv_0}} \o \ol{\partial}}\right)$. More precisely, the following holds:
	\begin{restatable}{lemma}{coalgebramapalgebrahomomorphism}\label{lem:algebra_map_coalgebra_homomorphism}
    The $\pca$ structure map $\alpha_{\equiv_0} : \distf {\Exp}/{\equiv_0} \to {\Exp}/{\equiv_0}$ is a $G$-coalgebra homomorphism $$\alpha_{\equiv_0} : (\distf {\Exp}/{\equiv_0}, (\gamma_{{\Exp}/{\equiv_0}} \o \ol{\partial} )^\star) \to ({\Exp}/{\equiv_0}, G \alpha_{\equiv_0} \o \gamma_{{\Exp}/{\equiv_0}} \o \ol{\partial})$$
\end{restatable}
Secondly, we can use that $G$-coalgebra structure on ${\Exp}/{\equiv_0}$ in order to endow the coarser quotient ${\Exp}/{\equiv}$ with a $G$-coalgebra structure.
 \begin{restatable}{lemma}{quotienttrace}\label{lem:coalg_on_quotient_left_dist}
    There exists a unique $G$-coalgebra structure $d : {\Exp}/{\equiv} \to G{\Exp}/{\equiv}$ such that the following diagram commutes
\[\begin{tikzcd}[ampersand replacement=\&]
	\Exp \&\& {{\Exp}/{\equiv_0}} \&\& {{\Exp}/{\equiv}} \\
	{\distf F \Exp} \&\& {\distf F {\Exp}/{\equiv_0}} \\
	\&\& {G \distf {\Exp}/{{\equiv}_0}} \\
	\&\& {G{Exp}/{\equiv_0}} \&\& {G{\Exp}/{\equiv}}
	\arrow["\partial"', from=1-1, to=2-1]
	\arrow["{[-]_{\equiv_0}}", from=1-1, to=1-3]
	\arrow["{\overline{\partial}}"', from=1-3, to=2-3]
	\arrow["{\distf F [-]_{\equiv_0}}"', from=2-1, to=2-3]
	\arrow["{[-]_{\equiv}}"', from=1-3, to=1-5]
	\arrow["{\gamma_{{\Exp}/{\equiv_0}}}"', from=2-3, to=3-3]
	\arrow["{G\alpha_{\equiv_0}}"', from=3-3, to=4-3]
	\arrow["{G[-]_{\equiv}}"', from=4-3, to=4-5]
	\arrow["d"{left}, from=1-5, to=4-5, dashed]
\end{tikzcd}\]
\end{restatable} 
The above lemma encapsulates the key step of the soundness proof. We argue the uniqueness through induction on the length of derivation of $\equiv$ using the diagonal fill-in property~\cite[Lemma~3.17]{Gumm:2000:Elements}. Since the proof involves going through the quotient ${\Exp}/{\equiv_0}$, which identifies the expressions modulo the axioms of the finer relation $\equiv_0$, it just remains to verify the soundness of axioms \textbf{S0} and \textbf{D2} present in the finer relation $\equiv$.
	The structure map $d : {\Exp}/{\equiv} \to G {\Exp}/{\equiv}$ obtained through the above argument, can be shown to be concretely given by:
$$
d\left(\left[p \cdot \one \oplus \bigoplus_{i \in I} p_i \cdot a_i \seq e_i\right]\right)= \left\langle p, \lambda a . \left[\bigoplus_{a_i = a} p_i \cdot e_i \right] \right\rangle
$$
We note that as a consequence of \cref{thm:fundamental_theorem} every element of ${\Exp}/{\equiv}$ can be written in the normal form used in the left-hand side of the definition above.

\noindent
\textbf{Step 4: Soundness result.}
Through a simple finality argument (see \Cref{lem:factoring_lemma} in the appendix for details), we can show that the unique $G$-coalgebra homomorphism from the determinisation of the Antimirov transition system, can be viewed as the following composition of homomorphisms, which we have obtained in the earlier steps:

\[
\small
\begin{tikzcd}
	{\distf \Exp} & {\distf{{\Exp}/{\equiv_{0}}}} & {{\Exp}/{\equiv_0}} & {\Exp/{\equiv}} & {[0,1]^{A^*}}
	\arrow["{[-]_{\equiv}}", from=1-3, to=1-4]
	\arrow["{\alpha_{\equiv_0}}", from=1-2, to=1-3]
	\arrow["{\dagger d}", dashed, from=1-4, to=1-5]
	\arrow["{\D [-]_{\equiv_0}}", from=1-1, to=1-2]
	\arrow["{\dagger (\gamma_{\Exp} \o \partial)^{\star}}"', curve={height=18pt}, from=1-1, to=1-5, dashed]
\end{tikzcd}\]
Since the language-assigning map relies on the homomorphism described above, we can show the following:
\begin{restatable}{lemma}{factoringlemmatwo}\label{lem:factoring_lemma2} The function $\llbracket-\rrbracket : \Exp \to [0,1]^{A^*}$ assigning each expression to its semantics satisfies:
$
    \llbracket - \rrbracket = \dagger d \o [-]
$. 
\end{restatable}
 \begin{proof}
    \begin{align*}
        \llbracket - \rrbracket &= \dagger (\gamma_{\Exp} \o \partial)^{\star} \o \eta_{\Exp} \tag{Def. of $\llbracket - \rrbracket$}\\
        &=\dagger d \o [-]_{\equiv} \o \alpha_{\equiv_0} \o \distf [-]_{\equiv_0} \o \eta_{\Exp}\tag{\cref{lem:factoring_lemma}}\\
        &=\dagger d \o [-]_{\equiv} \o \alpha_{\equiv_0} \o \eta_{{\Exp}/{\equiv_0}} \o [-]_{\equiv_0} \tag{$\eta$ is a natural transformation}\\
        &=\dagger d \o [-]_{\equiv} \o [-]_{\equiv_0} \tag{$\alpha_\equiv$ is an Eilenberg-Moore algebra}\\
        &= \dagger d \o [-] \tag{$[-] = [-]_{\equiv} \circ [-]_{\equiv_0} $}\\
    \end{align*}
    \end{proof}
We can now immediately conclude that provably equivalent expressions are mapped to the same probabilistic languages, thus establishing soundness.
\begin{theorem}[Soundness]
    For all $e,f \in \Exp$, if $e \equiv f$ then $\llbracket e \rrbracket = \llbracket f \rrbracket$
\end{theorem}


\section{Completeness}\label{sec:completeness}
We now move on to showing completeness through the steps described in \Cref{sec:roadmap}. 
\vspace{2px}
\newline
\textbf{Step 1: Algebra structure.}
Throughout the soundness proof, we have shown that the semantics of any expression $e \in \Exp$ can also be seen as the language of the state corresponding to the equivalence class $[e]$ in the deterministic transition system ($G$-coalgebra) defined on the set ${\Exp}/{\equiv}$. The completeness proof will rely on establishing that this coalgebra possesses a certain universal property, that will imply completeness. A first step in arguing so is observing that ${\Exp}/{\equiv}$, besides being equipped with a coalgebra structure, also carries the algebra structure. 
\begin{restatable}{lemma}{coraserquotientisapca}\label{lem:coarser_quotient_is_a_pca}
    The set $\Exp/{\equiv}$ can be equipped with $\pca$ structure given by:
    $$\bigboxplus_{i \in I} p_i \cdot [e_i] = \left[\bigoplus_{i \in I} p_i \cdot e_i\right]$$
    Moreover $[-]_{\equiv}: {\Exp}/{\equiv_0} \to {\Exp}/{\equiv}$ is a $\pca$ homomorphism from ${\Exp/{\equiv_0}}$ to ${\Exp}/{\equiv}$.
\end{restatable}
Thanks to the fact that we can lift $G : \Set \to \Set$ to the category of $\pca$, the set $G{\Exp}/{\equiv}$ also carries the algebra structure. In such a setting, the transition function $d : {\Exp}/{\equiv} \to G {\Exp}/{\equiv}$ becomes an algebra homomorphism.
\begin{restatable}{lemma}{smappcahomom}
    $d : {\Exp}/{\equiv} \to G{\Exp}/{\equiv}$ is a $\pca$ homomorphism from $$d : ({\Exp}/{\equiv}, \alpha_{\equiv}) \to G({\Exp}/{\equiv}, \alpha_{\equiv})$$
\end{restatable}
In other words $(({\Exp}/{\equiv}, \alpha_{\equiv}), d)$ is a coalgebra for the lifted functor $G : \pca \to \pca$. 
\noindent
\vspace{2px}
\\
\textbf{Interlude: coalgebraic completeness theorems.}
The completeness proof will follow a pattern of earlier work of Jacobs~\cite{Jacobs:2006:Bialgebraic}, Silva~\cite{Silva:2010:Kleene} and Milius~\cite{Milius:2010:Sound} and show that the coalgebra structure on the ${\Exp}/{\equiv}$ is isomorphic to the subcoalgebra of an appropriate final coalgebra, ie. the unique final coalgebra homomorphism from ${\Exp}/{\equiv}$ is injective. The intuition comes from the coalgebraic modelling of deterministic automata studied in the work of Jacobs~ \cite{Jacobs:2006:Bialgebraic}. In such a case, the final coalgebra is simply the automaton structure on the set of all formal languages, while the final homomorphism is given by the map taking a state of the automaton to a language it denotes. Restricting the attention to finite-state automata only yields \emph{regular languages}. The set of regular languages can be equipped with an automaton structure, in a way that inclusion map into the final automaton on the set of all formal languages becomes a homomorphism. In such a case, Kozen's completeness proof of Kleene Algebra~\cite{Kozen:1994:Completeness} can be seen as showing isomorphism of the automaton of regular languages and the automaton structure on the regular expressions modulo \acro{KA} axioms. 

Unfortunately, we cannot immediately rely on the identical pattern. Our semantics relies on determinising \acro{GPTS}, but unfortunately, determinising a finite-state \acro{GPTS} can yield $G$-coalgebras with infinite carriers. For example, determinising a single-state \acro{GPTS} would yield a $G$-coalgebra over the set of subdistributions over a singleton set, which is infinite. Luckily, all $G$-coalgebras we work with have additional algebraic structure. This algebraic structure will allow us to rely on the generalisations of the concept of finiteness beyond the category of sets, offered by the theory of locally finite presentable categories~\cite{Adamek:1994:Locally}. Being equipped with those abstract lenses, one can immediately see that the earlier mentioned infinite set of all subdistributions over a singleton set is also a free $\pca$ generated by a single element and thus finitely presentable~\cite{Sokolova:2015:Congruences}. 
	
In particular, we will work with a \emph{rational fixpoint}; a generalisation of the idea of subcoalgebra of regular languages to coalgebras for finitary functors $B : \A \to \A$ over locally finitely presentable categories. The rational fixpoint provides a semantical domain for the behaviour of coalgebras whose carriers are finitely presentable in the same way as regular languages provide a semantic domain for all finite-state deterministic automata. The completeness proof will essentially rely on establishing that the coalgebra structure on ${\Exp}/{\equiv}$ satisfies the universal property of the rational fixpoint.
		
	We start by recalling basic notions of locally finite presentable categories~\cite{Adamek:1994:Locally} and properties of the rational fixpoint.
\begin{tcolorbox}[breakable, enhanced,title=Lfp categories and the rational fixpoint---definitions and notation,colback=Periwinkle!5!white,colframe=Periwinkle!75!black,]
$\D$ is a filtered category, if every finite subcategory $\D_0 \hookrightarrow \D$ has a cocone in $\D$. A filtered colimit is a colimit of the diagram $\D \to \C$, where $\D$ is a filtered category. A directed colimit is a colimit of the diagram $\D \to \C$, where $\D$ is a directed poset. We call a functor \emph{finitary} if it preserves filtered colimits. An object $C$ is \emph{finitely presentable (fp)} if the representable functor $\C(C,-) : \C \to \Set$ preserves filtered colimits. Similarly, an object $C$ is \emph{finitely generated (fg)} if the representable functor $\C(C,-) : \C \to \Set$ preserves directed colimits of monos. Every fp object is fg, but the converse does not hold in general. A category $\C$ is \emph{locally finitely presentable (lfp)} if it is cocomplete and there exists a set of finitely presentable objects, such that every object of $\C$ is a filtered colimit of objects from that set. $\Set$ is the prototypical example of a locally finitely presentable category, where finite sets are the fp objects.
\tcblower
Let $B : \A \to \A$ be a finitary functor. We write $\coaf{B}$ for the subcategory of $\coa{B}$ consisting only of $B$-coalgebras with finitely presentable carrier. The \emph{rational fixpoint} is defined as $(\varrho B,r) = \colim(\coaf{B} \hookrightarrow \coa{B})$ -- a colimit of the inclusion functor from the subcategory of coalgebras with finitely presentable carriers. We call it a fixpoint, as the map $r : \varrho B \to B (\varrho B) $ is an isomorphism~\cite{Adamek:2006:Iterative}. 
Following~{\cite[Corollary~3.10, Theorem~3.12]{Milius:2020:New}}, if finitely presentable and finitely generated objects coincide in $\A$ and $B : \A \to \A$ is a finitary endofunctor preserving non-empty monomorphisms, then rational fixpoint is fully abstract - ie. $(\varrho B, r)$ is a subcoalgebra of $(\nu B, t)$.
\end{tcolorbox}
\vspace{2px}
\noindent
\textbf{Proper functors.} It can be easily noticed that the generalised determinisation of coalgebras with finite carriers corresponds to algebraically structured coalgebras of a particular, well-behaved kind. Namely, their carriers are algebras which are free finitely generated. We write $\coafr{B}$ for the subcategory of $\coa{B}$ consisting only of $B$-coalgebras with free finitely generated carriers. The recent work of Milius~\cite{Milius:2018:Proper} characterised \emph{proper} functors, for which in order to establish that some $B$-coalgebra is isomorphic to the rational fixpoint it will suffice to look at coalgebras with free finitely generated carriers. To put that formally:

\begin{theorem}[{\cite[Corollary~5.9]{Milius:2018:Proper}}]\label{thm:proper}
    Let $B : \Set^T \to \Set^T$ be a proper functor. Then a $B$-coalgebra $(R,r)$ is isomorphic to the rational fixpoint if $(R,r)$ is locally finitely presentable and for every $B$-coalgebra $(TX, c)$ in $\coafr{B}$ there exists a unique homomorphism from $TX$ to $R$.
\end{theorem}
As much as $G : \pca \to \pca$ is known to be proper~\cite{Sokolova:2018:Proper}, not every $G$-coalgebra with a free finitely generated carrier corresponds to a determinisation of some \acro{GPTS}. This is simply too general, as some $G$-coalgebras with free finitely generated carriers might be determinisations of \acro{RPTS} not corresponding to any \acro{GPTS}. To circumvent that, instead of looking at all coalgebras for the functor $G : \pca \to \pca$, we can restrict our attention in a way that will exclude determinisations of \acro{RPTS} not corresponding to any \acro{GPTS}. To do so, define a functor $\hat{G} : \pca \to \pca$. Given a positive convex algebra $(X, \boxplus)$ we set:
\begin{align*}
    \hat{G}(X, \boxplus) = &\{(o, f) \in [0,1] \times X^A \mid \forall {a \in A}~\exists {p_a^i \in [0,1], x_a^i \in X}. \\ &f(a) = \bigboxplus_{i \in I} p_a^i x_a^i \text{and} \sum_{a \in A} \sum_{i \in I} p_a^i \leq 1-o\}
\end{align*}
The $\pca$ structure on $\hat{G}(X, \boxplus)$, as well as the action of $\hat{G}$ on arrows is defined to be the same as in the case of $G$. 
It can be immediately observed that $\hat{G}$ is a subfunctor of $G$. Whenever the algebra structure is clear from the context, we write $\hat{G}X$ for $\hat{G}(X, \boxplus)$. Most importantly for us, thanks to the result of Sokolova and Woracek, we know that $\hat{G}$ is also proper~\cite{Sokolova:2018:Proper}. We can now see the following correspondence.

\begin{restatable}{lemma}{onetoone}\label{lem:onetoone_sublemma}
    $\distf F$-coalgebras with finite carrier are in one-to-one correspondence with $\hat{G}$-coalgebras with free finitely generated carrier.
    $$\mprset{fraction={===}}
      \inferrule {\beta : X \to \distf F X \text{ on } \Set} {\xi : (\distf X, \mu_X) \to \hat{G} (\distf X, \mu_X) \text{ on } \pca}$$ 
      In other words, every coalgebra structure map $\xi : (\distf X, \mu_X) \to \hat{G} (\distf X, \mu_X)$ is given by $\xi = (\gamma_X \o \beta)^\star$ for some unique $\beta : X \to \distf F X$.
    \end{restatable}

	Additionally, the coalgebra structure on ${\Exp}/{\equiv}$, which is at the centre of attention of the completeness proof, happens also to be a $\hat{G}$-coalgebra.
\begin{restatable}{lemma}{quotientcoalgebraforsubfunctor}\label{lem:determinisations_correspond_to_df_coalgebras}
    $(({\Exp}/{\equiv}, \alpha_{\equiv}), d)$ is a $\hat{G}$-coalgebra.
\end{restatable}

\noindent\textbf{Step 2: Systems of equations.}
In order to establish that ${\Exp}/{\equiv}$ is isomorphic to the rational fixpoint (which is the property that will eventually imply completeness), we establish the conditions of \Cref{thm:proper}. One of the required things we need to show is that the determinisation of an arbitrary finite-state GPTS admits a unique homomorphism to the coalgebra carried by ${\Exp}/{\equiv}$. In other words, we need to convert states of an arbitrary finite-state GPTS to language equivalent expressions in a way which is unique up to the axioms of $\equiv$. This can be thought of as an abstract reformulation of one direction of the Kleene theorem to the case of PRE and GPTS. To make that possible, we give a construction inspired by Brzozowski's equation solving method~\cite{Brzozowski:1964:Derivatives} of converting a DFA to the corresponding regular expression. We start by stating the necessary definitions:
\begin{tcolorbox}[breakable, enhanced,title=Systems of equations---definitions and notation,colback=Periwinkle!5!white,colframe=Periwinkle!75!black,]
A \emph{left-affine system} on finite set \(Q\) of unknowns is a quadruple
\begin{align*}
	\mathcal{S}=\langle &M : Q \times Q \to \Exp, p : Q \times Q \to [0,1],\\ &b : Q \to \Exp, r : Q \to [0,1] \rangle
\end{align*}
such that for all $q,q' \in Q$, \(\sum_{q' \in Q}p_{q,q'} + r_q = 1\) and $E(M_{q,q'})=0$.
\tcblower
Let \({\sim} \subseteq {\Exp \times \Exp}\) be a congruence relation. A map \(h : Q \to \Exp\) is \(\sim\)-solution if for all \(q \in Q\) we have that:
\[
    h(q) \sim \left(\bigoplus_{q' \in Q} p_{q,q'} \cdot M_{q,q'}\seq h(q')\right) \oplus r_{q} \cdot b_{q}
\]
\end{tcolorbox}
A system representing the finite-state \acro{GPTS} $(X, \beta)$ is given by $\mathcal{S}(\beta) = \langle M^\beta, p^\beta, b^\beta, r^\beta\rangle $ where for all $x, x' \in X$ we have that:
\begin{gather*}
p^\beta_{x, x'} = \sum_{a ' \in A}\beta(x)(a', x')\\
M^\beta_{x, x'} = \begin{cases}
    \bigoplus_{a \in A} \frac{\beta(x)(a,x')}{p^{\beta}_{x,x'}} \cdot a & \text{if } p^\beta_{x,x'}\neq 0\\
    \zero & \text{otherwise}
\end{cases} \\
r^\beta_x = 1 - \sum_{(a',x') \in \supp (\beta(x))}\beta(x)(a',x')\\
 b^\beta_x = \begin{cases}
    \frac{\beta(x)(\checkmark)}{r^\beta_x} \cdot \one & \text{if } r^\beta_x \neq 0\\
    \zero & \text{otherwise}
\end{cases} 
\end{gather*}
Note that for all $x,x' \in X$ we have that $E\left(M^{\beta}_{x,x'}\right)=0$. This property will be particularly useful when solving the systems up to $\equiv$ using the \textbf{Unique} fixpoint axiom.

\vspace{2px}
\noindent
\textbf{Unique solution of systems.} The original Brzozowski's equation-solving method is purely semantic, as it crucially relies on Arden's rule by providing solutions up to the language equivalence to the systems of equations. As much as it would be enough for an analogue of the one direction of Kleene's theorem, for the purposes of our completeness argument, we need to argue something stronger. Namely, we show that we can uniquely solve each system purely through the means of syntactic manipulation using the axioms of $\equiv$. This is where the main complexity of the completeness proof is located. We show this property, by re-adapting the key result of Salomaa~\cite{Salomaa:1966:Two} to the systems of equations of our interest.
\begin{restatable}{theorem}{uniquesolutions}\label{thm:unique_solutions}
    Every left-affine system of equations admits a unique $\equiv$-solution.
\end{restatable}
The proof of the result above proceeds by induction on the size of the systems of equations. We first show that systems with only one unknown can be solved via the \textbf{Unique} fixpoint axiom. Systems of the size $n+1$ can be reduced to systems of size $n$, by solving for one of the unknowns and substituting the obtained equation with $n$ unknowns to the remaining $n$ equations. We note that this reduction step is highly reliant on axioms \textbf{D2} and \textbf{S0}.
\begin{example}
    Consider the transition system from the left-hand side of \cref{ex:systems}. A map $h : \{q_0, q_1\} \to \Exp $ is an $\equiv$-solution to the system associated with that transition system, if and only if:
    \[
        h(q_0) \equiv a \seq h(q_1) \qquad h(q_1) \equiv a \seq h(q_1) \oplus_{\frac{1}{4}} \one
    \]
    Since $E(a)=0$, we can apply the unique fixpoint axiom and $e \seq \one \equiv e$ to the equation on the right to deduce that 
    $
        h(q_1) \equiv a^{\left[\frac{1}{4}\right]}
    $. Substituting it into the left equation yields $h(q_0)\equiv a \seq a^{\left[\frac{1}{4}\right]}$.
\end{example}
\noindent
\textbf{Step 3: Correspondence of solutions and homomorphisms.}
We are not done yet, as in the last step we only proved properties of systems of equations and their solutions, while our main interest is in appropriate $\hat{G}$-coalgebras and their homomorphisms. As desired, it turns out that $\equiv$-solutions are in one-to-one correspondence with $\hat{G}$-coalgebra homomorphisms from determinisations of finite state \acro{GPTS} to the coalgebra structure on ${\Exp}/{\equiv}$.
\begin{restatable}{lemma}{solutionshomomorphisms}\label{lem:solutions_homomorphisms}
    For a finite set $X$, we have the following one-to-one correspondence:
    $$
 \myinferrule{\begin{gathered}
 \text{$\hat{G}$-coalgebra homomorphisms }\\ m : ((\distf X, \mu_{X}), (\gamma_X \o \beta)^\star) \to (({\Exp}/{\equiv}, \alpha_{\equiv}), d) 	
 \end{gathered}
}{\begin{gathered}
 	\text{$\equiv$-solutions } h: X \to \Exp \text{ to a system } \mathcal{S}(\beta)\\ \text{ associated with $\distf F$-coalgebra }(X, \beta)
 \end{gathered}}$$ 
\end{restatable}
Aside from the completeness argument, the above result also gives us an analogue of (one direction of) Kleene's theorem for \acro{GPTS} as a corollary. The other direction, converting \acro{PRE} to finite-state \acro{GPTS} is given by the Antimirov construction, described in \Cref{sec:operational_semantics}.
\begin{restatable}{corollary}{kleene}\label{cor:kleene}
	Let $(X, \beta)$ be a finite-state $\distf F$-coalgebra. For every state $x \in X$, there exists an expression $e_x \in \Exp$, such that the probabilistic language denoted by $x$ is the same as $\llbracket e_x \rrbracket$.
\end{restatable}
In particular, if $\mathcal{S}(\beta)$ is the system of equations associated with $(X, \beta)$ and $h : X \to \Exp$ is an $\equiv$-solution to it (which exists because of \Cref{thm:unique_solutions}), then we set $e_x = h(x)$.

\noindent
\textbf{Step 4: Establish the universal property.}
We are now ready to argue that the coalgebra structure on ${\Exp}/{\equiv}$ is isomorphic to the rational fixpoint of the functor $\hat{G}$. Recall that by \cref{lem:onetoone_sublemma}, every coalgebra in $\coafr{\hat{G}}$ is a determinisation of some finite-state \acro{GPTS}. Combining it with the correspondence from \Cref{lem:solutions_homomorphisms} and the fact that every left-affine system possesses a unique $\equiv$-solution (\Cref{thm:unique_solutions}) we can conclude that any coalgebra in $\coafr{\hat{G}}$ admits a unique $\hat{G}$-coalgebra homomorphism to $(({\Exp}/{\equiv}, \alpha_{\equiv}), d)$. This establishes the first condition of \Cref{thm:proper}. We are left with showing the remaining condition:
\begin{restatable}{lemma}{lfplemma}\label{lem:quotient_is_lfp}
    $\hat{G}$-coalgebra $(({\Exp}/{\equiv}, \alpha_{\equiv}), d)$ is locally finitely presentable.
\end{restatable}
The proof of the above claim is available in the appendix of the full version of the paper. Essentially, we argue that starting from an arbitrary convex combination of finitely many expressions, there are finitely many expressions such that reachable states in the coalgebra structure on ${\Exp}/{\equiv}$ can be seen as convex combinations of expressions from that set. Speaking more precisely, we show that every finitely generated subalgebra of ${\Exp}/{\equiv}$ is contained in the subocalgebra of ${\Exp}/{\equiv}$ with a finitely generated carrier. The proof makes use of the fact that every expression has finitely many Antimirov derivatives. 
We can now combine the above results and conclude that:
\begin{corollary}\label{cor:rational}
   $(({\Exp}/{\equiv}, \alpha_{\equiv}), d)$ is isomorphic to the rational fixpoint for the functor $\hat{G}$.
\end{corollary}

\noindent\textbf{Completeness.} At this point, all that remains is to make an additional observation that the universal property we have established implies completeness. Since $\hat{G}$ preserves non-empty monomorphisms~\cite{Milius:2020:New} and in $\pca$ finitely generated and finitely presentable objects coincide~\cite{Sokolova:2015:Congruences} the rational fixpoint is fully abstract, i.e. is a subcoalgebra of the final $\hat{G}$-coalgebra. Moreover, since $\hat{G}$ is a subfunctor of $G$, the final $\hat{G}$-coalgebra is subcoalgebra of the final $G$-coalgebra. Combining those facts, allows to conclude that $\dagger d : {{\Exp}/{\equiv}} \to [0,1]^{A^*}$ is injective. At this point, showing completeness becomes straightforward.
\begin{theorem}[Completeness]
     For all $e,f \in \Exp$, if $\llbracket e \rrbracket = \llbracket f \rrbracket$ implies $e \equiv f$
\end{theorem}
\begin{proof} Assume that $ \llbracket e \rrbracket = \llbracket f \rrbracket$. Then, by \cref{lem:factoring_lemma2} we have that: $$(\dagger d \circ [-])(e) = (\dagger d \circ [-])(f)$$ Since $\dagger d$ is injective, we can cancel it on the left to obtain $[e] = [f]$, which is equivalent to $e \equiv f$. 

\end{proof}
\section{Discussion}\label{sec:related}

We introduced probabilistic regular expressions (\acro{PRE}), a probabilistic counterpart to Kleene's regular expressions. As the main technical contribution, we presented a Salomaa-style inference system for reasoning about probabilistic language equivalence of expressions and proved it sound and complete. Additionally, we gave a two-way correspondence between languages denoted by \acro{PRE} and finite-state Generative Probabilistic Transition Systems. Our approach is coalgebraic and enabled us to reuse several recently proved results on fixpoints of functors and convex algebras. This abstract outlook guided the choice of the right formalisms and enabled us to isolate the key results we needed to prove to achieve completeness while at the same time reusing existing results and avoiding repeating complicated combinatorial proofs.  The key technical lemma, on uniqueness of solutions to certain systems of equations, is a generalisation of automata-theoretic constructions from the 60s further exposing the bridge between our probabilistic generalisation and the classical deterministic counterpart. 

\vspace{2px}
\noindent
\textbf{Related work.} Probabilistic process algebras and their axiomatisations have been widely studied~\cite{Bandini:2001:Axiomatizations,Stark:2000:Complete,Mislove:2003:Axioms,Benardo:2022:Probabilistic} with syntaxes featuring action prefixing and least fixed point operators instead of the regular operations of sequential composition and probabilistic loops. This line of research focussed on probabilistic bisimulation, while probabilistic language equivalence, which we focus on, stems from automata theory, e.g. the work of Rabin on probabilistic automata~\cite{Rabin:1963:Probabilistic}. Language equivalence of Rabin automata has been studied from an algorithmic point of view~\cite{Kiefer:2011:Language,Kiefer:2012:APEX}.

Stochastic Regular Expressions (\acro{SRE})~\cite{Ross:2000:Probabilistic,Beeh:2017:Transformations,Getir:2018:Formal}, which were one of the main inspirations for this paper, can also be used to specify probabilistic languages. The syntax of \acro{SRE} features probabilistic Kleene star and $n$-ary probabilistic choice, however, it does not include $\zero$ and $\one$. The primary context of that line of research was around genetic programming in probabilistic pattern matching, and the topic of axiomatisation was simply not tackled. 

 \acro{PRE} can be thought of as a fragment of \acro{ProbGKAT}~\cite{Rozowski:2023:Probabilistic}, a probabilistic extension of a strictly deterministic fragment of Kleene Algebra with Tests, that was studied only under the finer notion of bisimulation equivalence. The completeness \acro{ProbGKAT} was obtained through a different approach to ours, as it relied on a powerful axiom scheme to solve systems of equations. 

Our soundness result, as well as semantics via generalised determinisation, were inspired by the work of Silva and Sokolova~\cite{Silva:2011:Sound}, who introduced a two-sorted process calculus for reasoning about probabilistic language equivalence of \acro{GPTS}. Unlike \acro{PRE}, their language syntactically excludes the possibility of introducing recursion over terms which might immediately terminate. Moreover, contrary to our completeness argument, their result hinges on the subset of axioms being complete wrt bisimilarity, similarly to the complete axiomatisation of trace congruence of \acro{LTS} due to Rabinovich~\cite{Rabinovich:1993:Complete}. The use of coalgebra to model trace/language semantics is a well-studied topic~\cite{Jacobs:2015:Trace,Rot:2021:Steps} and other approaches besides generalised determinisation~\cite{Silva:2010:Generalizing,Bonchi:2017:Power} included the use of Kleisli categories~\cite{Hasuo:2007:Generic} and coalgebraic modal logic~\cite{Klin:2015:Coalgebraic}. 
We build on the vast line of work on coalgebraic completeness theorems~\cite{Jacobs:2006:Bialgebraic,Silva:2010:Kleene,Schmid:2021:Star,Milius:2010:Sound,Bonsangue:2013:Sound}, coalgebraic semantics of probabilistic systems~\cite{Vink:1999:Bisimulation,Sokolova:2005:Coalgebraic} and fixpoints of the functors~\cite{Milius:2020:New,Milius:2018:Proper,Sokolova:2018:Proper}. 

\vspace{2px}
\noindent
\textbf{Future work.} A first natural direction is exploring whether one could obtain an \emph{algebraic} axiomatisation of \acro{PRE}. Similarly to Salomaa's system, our axiomatisation is unsound under substitution of letters by arbitrary expressions in the case of the termination operator used to give side condition to the unique fixpoint axiom. We are interested if one could give an alternative inference system in the style of Kozen's axiomatisation~\cite{Kozen:1994:Completeness}, in which the Kleene star is the least fixpoint wrt the natural order induced by the $+$ operation and thus not requiring the side condition to introduce loops. In the case of \acro{PRE}, the challenge is that there is no obvious way of defining a natural order on \acro{PRE} in terms of $\oplus_p$ operation. 

Additionally, we would like to study if the finer relation $\equiv_{0}$ is complete wrt probabilistic bisimilarity. One could view it as a probabilistic analogue of the problem of completeness of Kleene Algebra modulo bisimilarity posed by Milner~\cite{Milner:1984:Complete}, which was recently answered positively by Grabmayer~\cite{Grabmayer:2022:Milner}. 

Finally, an interesting direction would be to study a more robust notion of \emph{language distance}~\cite{Baldan:2018:Coalgebraic} of \acro{GPTS} by extending our axiomatisation to a system based on quantitative equational logic~\cite{Mardare:2016:Quantitative}. Similar results were already obtained in the case of probabilistic process calculus of Stark and Smolka~\cite{Stark:2000:Complete} for bisimulation distance~\cite{Bacci:2018:Complete} and distance between infinite traces~\cite{Bacci:2018:CompleteTraces}.
\begin{acks}
	  The authors would like to thank Samson Abramsky, Robin Hirsch, Benjamin Kaminski, Tobias Kappé and Todd Schmid for helpful discussions.
	  The authors would also like to thank the anonymous referees for their valuable comments and helpful suggestions. This work is supported by the European Research Council grant Autoprobe (grant agreement 101002697).
\end{acks}
\bibliographystyle{ACM-Reference-Format}
\bibliography{bibliography}

\appendix
\onecolumn
%
\section{Antimirov derivatives are finite}
\begin{lemma}\label{lem:locally_finite}
    For all $e \in \Exp$, the set $\langle e \rangle$ is finite. In fact, the number of of states is bounded above by $\#(-): \Exp \
    \to \Nat$, where $\#(-)$ is defined recursively by:
    \begin{align*}
        \#(\zero)=\#(\one)=1 \quad \#(a) = 2 \quad \#(e \oplus_p f) = \#(e) + \#(f)\\ 
        \#(e \seq f) = \#(e) + \#(f) \quad \#(e^{[p]}) = \#(e) + 1
    \end{align*}
\end{lemma}
\begin{proof}
    We adapt the analogous proof for \acro{GKAT}~\cite{Schmid:2021:Guarded}.

    For any \({e} \in \Exp\), let \(|\langle e \rangle|\) be the cardinality of the carrier set of the least subcoalgebra of \((\Exp, \partial)\) containing \(e\). We show by induction that for all \(e \in \Exp\) it holds that \(|\langle e \rangle|\leq \#({e})\).

     For the base cases, observe that for $\zero$ and $\one$ the subcoalgebra has exactly one state. Hence, \(\#(\zero) = 1 = |\langle \zero \rangle|\). Similarly, we have \(\#(\one) = 1 = |\langle \one \rangle|\).
    For \(a \in A\), we have two states; the initial state, which transitions with probability \(1\) on \(a\) to the state which outputs \(\checkmark\) with probability \(1\).

    For the inductive cases, assume that \(|\langle e \rangle| \leq \#(e)\), \(|\langle f \rangle| \leq \#(f)\) and \(p \in [0,1]\).
     \begin{itemize}
         \item
         Every derivative of \({e} \oplus_p {f}\) is either a derivative of \({e}\) or \({f}\) and hence \(|\langle e \oplus_p f \rangle| \leq |\langle e \rangle| + |\langle f \rangle| = \#({e}) + \#({f}) = \#({e} \oplus_p {f})\).
         \item
         In the case of \({e}\seq{f}\), every derivative of this expression is either a derivative of \({f}\) or some derivative of \({e}\) followed by \({f}\). Hence, \(|\langle e \seq f \rangle| = |\langle e \rangle\times \{{f}\}| + |\langle f \rangle| \leq \#({e}) + \#({f}) = \#({e}\seq{f}) \).

         \item
         For the probabilistic loop case, observe that every derivative of \({e}^{[{p}]}\) is a derivative of \({e}\) followed by \({e}^{[{p}]}\) or it is the state that outputs $\checkmark$ with probability $1$. It can be easily observed,that \(|\langle e^{[p]}\rangle| \leq |\langle e \rangle| + 1 = \#({e}) = \#({e}^{[p]})\). 
    \end{itemize}
\end{proof}
\section{Positive convex algebras}
\begin{proposition}\label{prop:properties_of_positive_convex_algebras}
    In any positive convex algebra, the following hold:
    \begin{enumerate}
        \item $$\bigboxplus_{i \in I} p_i \cdot x_i = \bigboxplus_{x \in \bigcup_{i \in I} \{x_i\}} \left(\sum_{x_i = x} p_i\right)\cdot x$$
        \item Let ${\sim} \subseteq {X \times X}$ be a congruence relation. Then, 
        $${\bigboxplus_{i \in I} p_i \cdot x_i} \sim {\bigboxplus_{[x]_{\sim} \in \bigcup_{i \in I} \{[x_i]_{\sim}\}} \left(\sum_{x_i = x} p_i\right)\cdot x}$$
        \item All terms $\bigboxplus_{i \in I } 0 \cdot x_i $ coincide and are all provably equivalent to the empty convex sum.
        \item   Let $I$ be a finite index set, and let $\{p_i\}_{i \in I}$ and $\{x_i\}_{i \in I}$ be indexed collections. If $J \subseteq I$ and $J \supseteq \{i \in I \mid p_i \neq 0\}$, then 
        $$
        \bigboxplus_{i \in I} p_i \cdot x_i = \bigboxplus_{j \in J} p_j \cdot x_j
        $$
    \end{enumerate}
\end{proposition}
\begin{proof}
    In this lemma, we use $[\Phi]$ to denote Iverson bracket, which is $1$ if $\Phi$ is true and $0$ otherwise. For (1) we have that
    \begin{align*}
        \bigboxplus_{i \in I} p_i \cdot x_i &= \bigboxplus_{i \in I} p_i \cdot \left( \bigboxplus_{x \in \cup_{i \in I} \{x_i\}} [x_i = x] \cdot  x \right) &\tag{Projection axiom}\\
        &= \bigboxplus_{x \in \cup_{i \in I} \{x_i\}} \left( \sum_{i \in I} p_i[x_i = x] \right) \cdot x \tag{Barycenter axiom} \\
        &=  \bigboxplus_{x \in \cup_{i \in I} \{x_i\} } \left(\sum_{x_i = x} p_i\right) \cdot x
    \end{align*}
    (2) can be shown by picking a representative for each equivalence class and then using (1). For (3), by \cite[Lemma~3.4]{Sokolova:2015:Congruences} we know that all terms $\bigboxplus_{i \in I} 0 \cdot x_i$ coincide. To see that they are provably equivalent to empty convex sum, observe that
    \begin{align*}
        \bigboxplus_{i \in I} 0 \cdot x_i &= \bigboxplus_{i \in I} 0 \cdot \left(\bigboxplus_{j \in \emptyset} p_j \cdot y_j\right)\tag{\cite[Lemma~3.4]{Sokolova:2015:Congruences}}\\
        &= \bigboxplus_{j \in \emptyset} 0 \cdot y_j \tag{Barycenter axiom}\\
    \end{align*}
    (4) was also proved in \cite[Lemma~3.4]{Sokolova:2015:Congruences}.
\end{proof}
\begin{lemma}\label{lem:flattening_convex_sums}
    Let $I, J$ be finite index sets, $\{p_i\}_{i \in I}$, $\{q_{i,j}\}_{(i,j) \in I\times J}$ and $\{x_{i,j}\}_{(i,j) \in I \times J}$ indexed collections such that for all $i \in I$ and $j \in J$, $p_i, q_{i,j} \in [0,1]$ and $x_{i,j} \in X$. If $X$ carries $\pca$ structure, then:
    $$\bigboxplus_{i \in I} p_i \cdot \left(\bigboxplus_{j \in J} q_{i,j} \cdot x_{i,j}\right) = \bigboxplus_{(i,j) \in I \times J} p_iq_{i,j} \cdot x_{i,j}$$
\end{lemma}
\begin{proof}
    \begin{align*}
        \bigboxplus_{i \in I} p_i \cdot \left(\bigboxplus_{j \in J} q_{i,j} \cdot x_{i,j}\right) &= \bigboxplus_{i \in I} p_i \cdot \left(\bigboxplus_{(k,j) \in \{i\} \times J} q_{k,j} \cdot x_{k,j}\right) \\
        &= \bigboxplus_{i \in I} p_i \cdot \left(\bigboxplus_{(k,j) \in I \times J} [k = i]q_{k,j} \cdot x_{k,j}\right) \tag{\cref{prop:properties_of_positive_convex_algebras}} \\
        &= \bigboxplus_{(k,j) \in I \times J} \left(\sum_{i \in I} p_i [k=i]q_{k,j} \right) \cdot x_{k,j} \tag{Barycenter axiom} \\
        &= \bigboxplus_{(k,j) \in I \times J} p_k q_{k,j}  \cdot x_{k,j} \\
        &= \bigboxplus_{(i,j) \in I \times J} p_i q_{i,j}  \cdot x_{i,j} \tag{Renaming indices}\\
    \end{align*}
\end{proof}
\begin{lemma}\label{lem:grouping_probabilities}
    Let $I$ be a finite index set, $\{p_i\}_{i \in I}$ and $\{q_i\}_{i \in I}$ indexed collections such that $p_i,q_i \in [0,1]$ for all $i \in I$, $\sum_{i \in I} p_i + \sum_{i \in I} q_i \leq 1$ and let $\{x_i\}_{i \in I}$ and $\{y_i\}_{i \in I}$ indexed collection such that $x_i, y_i \in X$. If $X$ carries $\pca$ structure, then:
    $$
    \bigboxplus_{i \in I} p_i \cdot x_i \boxplus\bigboxplus_{i \in I} q_i \cdot y_i = \bigboxplus_{i \in I} (p_i + q_i) \cdot \left(\frac{p_i}{p_i + q_i} \cdot x_i \boxplus \frac{q_i}{p_i + q_i} \cdot y_i \right)
    $$
\end{lemma}
\begin{proof}
    Let $J = \{0,1\}$. Define indexed collections $\{r_{i,j}\}_{(i,j) \in I \times J}$ and $\{z_{i,j}\}_{(i,j) \in I \times J}$, such that $r_{i,0} = \frac{p_i}{p_i + q_i}$ and $z_{i,0} = x_i$ and $r_{i,1} = \frac{q_i}{p_i + q_i}$ and $z_{i,1} = x_i$.
    We have the following:
    \begin{align*}
        \bigboxplus_{i \in I} (p_i + q_i) \cdot \left(\frac{p_i}{p_i + q_i} \cdot x_i \boxplus \frac{q_i}{p_i + q_i} \cdot y_i\right) &= \bigboxplus_{i \in I} (p_i + q_i) \cdot \left(\bigboxplus_{j \in J} r_{i,j} \cdot z_{i,j} \right) \\
        &= \bigboxplus_{(i,j) \in I \times J} (p_i + q_i)r_{i_j} \cdot z_{i,j} \tag{\cref{lem:flattening_convex_sums}} \\
        &= \bigboxplus_{i \in I} p_i \cdot x_{i} \boxplus \bigboxplus_{i \in I} q_i \cdot y_i
    \end{align*}
\end{proof}
\begin{proposition}\label{prop:binary}
    If  $X$ is a set equipped with a binary operation $\boxplus_{p} : X \times X \to X$ for each $p \in [0,1]$ and a constant $\zero \in X$ satisfying for all $x,y,z \in X$ (when defined) the following:
    \begin{gather*}
        x \boxplus_{p} x = x \qquad x \boxplus_{1} y = x \qquad x \boxplus_{p} y = y \boxplus_{\ol{p}} x \\ (x \boxplus_{p} y) \boxplus_{q} z = x \boxplus_{pq} \left(y \boxplus_{\frac{(1-p)q}{1-pq}} z\right)
    \end{gather*}
    then $X$ carries the structure of a positive convex algebra. The interpretation of $\boxplus_{i \in I} p_i \cdot (-)_{i}$ is defined inductively by the following
    $$\bigboxplus_{i \in I} p_i \cdot x_i = \begin{cases}
        \zero & \text{if } I = \emptyset \\
        x_0 & \text{if } p_0 = 1\\
        x_n \boxplus_{p_k} \left(\bigboxplus_{i \in I \setminus \{k\}} \frac{p_i}{\ol{p_k}}\cdot x_i  \right) &\text{otherwise, for some } k \in I
    \end{cases} $$
\end{proposition}
\begin{proof}
    A straightforward re-adaptation of \cite[Proposition~7]{Bonchi:2017:Power}.
\end{proof}
\section{Syntax and axioms}
\begin{lemma}\label{lem:binary_sum_properties}
    The following facts are derivable in $\equiv_0$
    \begin{enumerate}
        \item $e \oplus_p (f 
        \oplus_q g) \equiv_0 \left( e \oplus_{\frac{p}{1-(1-p)(1-q)}} f \right) \oplus_{1-(1-p)(1-q)} g $
        \item $(e \oplus_p f) \oplus_q (g \oplus_p h) \equiv_0 (e \oplus_q g) \oplus_p (f \oplus_q h)$
    \end{enumerate}
\end{lemma}
\begin{proof}
\begin{enumerate}
    \item 
    Let $k = \frac{p}{1-(1-p)(1-q)} $ and $l = 1-(1-p)(1-q)$. We derive the following:
    \begin{align*}
        e \oplus_p (f \oplus_q g) &\equiv_0  (f \oplus_q g) \oplus_{\ol{p}} e \tag{\textbf{C3}} \\
        &\equiv_0 (g \oplus_{\ol{q}} f) \oplus_{\ol{p}} e \tag{\textbf{C3}} \\
        &\equiv_0 g \oplus_{\ol{l}} (f \oplus_{\ol{k}} e) \tag{\textbf{C4}} \\
        &\equiv_0 (f \oplus_{\ol{k}} e) \oplus_{{l}} g\tag{\textbf{C3}} \\
        &\equiv_0 (e \oplus_{{k}} f) \oplus_{{l}} g\tag{\textbf{C3}} \\
    \end{align*}
    \item We derive the following:
    \begin{align*}
        (e \oplus_p f) \oplus_q (g \oplus_p h) &\equiv_0 e \oplus_{pq} \left(f \oplus_{\frac{\ol{p}q}{1-pq}} (g \oplus_p h)\right) \tag{\textbf{C4}}\\
        &\equiv_0 e \oplus_{pq} \left((g \oplus_p h) \oplus_{\frac{\ol{q}}{1-pq}} f \right) \tag{\textbf{C3}} \\
        &\equiv_0 e \oplus_{pq} \left(g \oplus_{\frac{p\ol{q}}{1-pq}} \left(h \oplus_{\ol{q}} f\right) \right) \tag{\textbf{C4}} \\
        &\equiv_0 e \oplus_{pq} \left(g \oplus_{\frac{p\ol{q}}{1-pq}} \left(f \oplus_{{q}} h\right) \right) \tag{\textbf{C3}}\\
        &\equiv_0 \left(e \oplus_q g\right) \oplus_p \left(f \oplus_q h\right) \tag{subcase $1$ above}
    \end{align*}
\end{enumerate}
\end{proof}
\begin{lemma}\label{lem:generalised_right_distributivity}
    Let $f \in \Exp$, $I$ be a finite index set and let $\{p_i\}_{i \in I}$ and $\{e_i\}_{i \in I}$ indexed collections of probabilities and expressions respectively. Then,
    $$\left(\bigoplus_{i \in I} p_i \cdot e_i\right) \seq f \equiv_0 \bigoplus_{i \in I} p_i \cdot e_i \seq f$$
\end{lemma}
\begin{proof}
    By induction. If $I = \emptyset$, then: 
    \begin{align*}
        \left(\bigoplus_{i \in I} p_i \cdot e_i\right) \seq f &\equiv_0 \zero \seq f \\
        &\equiv_0 \zero \tag{\textbf{0S}}\\
        &\equiv_0 \bigoplus_{i \in I} p_i \cdot e_i \seq f \tag{$I = \emptyset$}
    \end{align*}
    If there exists $j \in I$, such that $p_j = 1$, then:
    \begin{align*}
        \left(\bigoplus_{i \in I} p_i \cdot e_i\right) \seq f &\equiv_0 e_j \seq f \\
        &\equiv_0     \left(\bigoplus_{i \in I} p_i \cdot e_i \seq f\right)
    \end{align*}
    Finally, for the induction step, we have that:
    \begin{align*}
         \left(\bigoplus_{i \in I} p_i \cdot e_i\right) \seq f &\equiv_0 \left( e_j \oplus_{p_j} \left(\bigoplus_{i \in I \setminus \{j\}} \frac{p_i}{\ol{p_j}} \cdot e_i\right) \right) \seq f \\
         &\equiv_0  e_j\seq f \oplus_{p_j} \left(\bigoplus_{i \in I \setminus \{j\}} \frac{p_i}{\ol{p_j}} \cdot e_i\right)\seq f \tag{\textbf{D1}} \\
         &\equiv_0  e_j\seq f \oplus_{p_j} \left(\bigoplus_{i \in I \setminus \{j\}} \frac{p_i}{\ol{p_j}} \cdot e_i\seq f\right) \tag{Induction hypothesis}\\
         &\equiv_0 \left(\bigoplus_{i \in I} p_i \cdot e_i \seq f\right)
    \end{align*}
\end{proof}
\begin{lemma}\label{lem:generalised_left_distributivity}
    Let $e \in \Exp$, $I$ be a finite index set and let $\{p_i\}_{i \in I}$ and $\{f_i\}_{i \in I}$ indexed collections of probabilities and expressions respectively. Then,
    $$e \seq \left(\bigoplus_{i \in I} p_i \cdot f_i\right) \equiv_0 \bigoplus_{i \in I} p_i \cdot e \seq f_i$$
\end{lemma}
\begin{proof}
    Analogous proof to \Cref{lem:generalised_right_distributivity} relying on axioms \textbf{S0} and \textbf{D2} instead.
\end{proof}
\section{Soundness wrt bisimilarity}
The \emph{weight} of a subdistribution $\nu : X \to [0,1]$ is a total probability of its support:
$$|\nu| = \sum_{x \in X} \nu(x)$$
Given $\nu\in\distf X$ and $Y \subseteq X$, we will write $\nu[Y] = \sum_{x \in Y} \nu(x)$. This sum is well-defined as only finitely many summands have non-zero probability. 
\begin{definition}[{\cite[Definition~2.1.2]{Hsu:2017:Probabilistic}}]\label{def:coupling}
    Given two subdistributions $\nu_1, \nu_2$ over $X$ and $Y$ respectively, a subdistribution $\nu$ over $X \times Y$ is called coupling if:
    \begin{enumerate}
        \item For all $x \in X$, $\nu_1(x) = \nu[\{x\}\times Y]$
        \item For all $y \in Y$, $\nu_2(y)= \nu[X \times \{y\}]$
    \end{enumerate} 
\end{definition}
It can be straightforwardly observed that a coupling $\nu$ of $(\nu_1, \nu_2)$ is finitely supported if and only if both $\nu_1$ and $\nu_2$ are finitely supported.
\begin{definition}[{\cite[Definition~2.1.7]{Hsu:2017:Probabilistic}}]\label{def:lifting}
    Let $\nu_1, \nu_2$ be subdistributions over $X$ and $Y$ respectively and let $R \subseteq X \times Y$ be a relation. A subdistribution $\nu$ over $X \times Y$ is a \emph{witness} for the $R$-\emph{lifting} of $(\nu_1, \nu_2)$ if:
    \begin{enumerate}
        \item $\nu$ is a coupling for $(\nu_1, \nu_2)$
        \item $\supp(\mu) \subseteq R$
    \end{enumerate}
\end{definition}
If there exists $\nu$ satisfying these two conditions, we say $\nu_1$ and $\nu_2$ are related by the \emph{lifting} of $R$ and write $\nu_1 \equiv_R \nu_2$. It can be immediately observed that $\nu_1 \equiv_R \nu_2$ implies $\nu_2 \equiv_{R^{-1}} \nu_1$.

Given a relation ${R}\subseteq{X \times Y}$ and a set $B \subseteq X$, we will write $R(B) \subseteq Y$ for the set given by $R(B)= \{y \in Y \mid (x,y) \in R \text{ and } x \in B\}$. We will write $R^{-1} \subseteq Y \times X$ for the converse relation given by $R^{-1} = \{(y,x) \in Y \times X \mid (x,y) \in R\}$. 
\begin{theorem}[{\cite[Theorem 2.1.11]{Hsu:2017:Probabilistic}}]\label{thm:coupling_theorem}
Let $\nu_1$, $\nu_2$ be subdistributions over $X$ and $Y$ respectively and let $R \subseteq X \times Y$ be a relation. Then the lifting $\nu_1 \equiv_R \nu_2$ implies $\nu_1[B] \leq \nu_2[R(B)]$ for every subset $B \subseteq X$. The converse holds if $\mu_1$ and $\mu_2$ have equal weight.
\end{theorem}
\begin{lemma}\label{lem:coupling_lemma}
Let $\nu_1$, $\nu_2$ be subdistributions over $X$ and $Y$ respectively and let $R \subseteq X \times Y$ be a relation. $\nu_1 \equiv_R \nu_2$ if and only if:
\begin{enumerate}
	\item For all $B \subseteq X$,  $\nu_1[B] \leq \nu_2[R(B)]$
	\item For all $C \subseteq Y$, $\nu_2[C] \leq \nu_1[R^{-1}(C)]$
\end{enumerate}
\end{lemma}
\begin{proof}
	Assume that $\nu_1 \equiv_R \nu_2$. Recall, that in such a case $\nu_2 \equiv_{R^{-1}} \nu_2$. Applying \cref{thm:coupling_theorem} yields (1) and (2) respectively.
	
	For the converse, assume (1) and (2) do hold. We have the following:
\begin{align*}
	|\nu_1| &= \nu_1[X] \\
	&\leq \nu_2[R(X)] \tag{1} \\
	&\leq \nu_2[Y] \tag{$R(X) \subseteq Y$}\\
	&= |\nu_2|
\end{align*}
By a symmetric reasoning involving (2), we can show that $|\nu_2| \leq |\nu_1|$ and therefore $|\nu_1| = |\nu_2|$. Since condition (1) holds, we can use \cref{thm:coupling_theorem} to conclude that $\nu_1 \equiv_R \nu_2$.
\end{proof}

\begin{definition}[{\cite[Definition~3.6.1]{Sokolova:2005:Coalgebraic}}]
    Let $R \subseteq X \times Y$ be a relation, and $B$ a $\Set$ endofunctor. The relation $R$ can be lifted to relation $\Rel(B)(R) \subseteq B X \times B Y$ defined by
    $$
        (x,y) \in \Rel(B)(R) \iff \exists z \in B R \text{ such that } B \pi_1(z)=x \text{ and } B \pi_2(z)=y
    $$
\end{definition}
\begin{lemma}[{\cite[Lemma~3.6.4]{Sokolova:2005:Coalgebraic}}]\label{lem:relation_lifting_bisim}
    Let $B : \Set \to \Set$ be an arbitrary endofunctor on $\Set$. A relation $R \subseteq X \times Y$ is a bisimulation between the $B$-coalgebras $(X, \beta)$ and $(Y, \gamma)$ if and only if: 
    $$(x,y) \in R \implies (\beta(x), \gamma(y)) \in \Rel(B)(R)$$
\end{lemma}
\begin{lemma}\label{lem:concrete_liftings}
    For any $R \subseteq {X \times Y}$
    \begin{enumerate}
        \item $\Rel(\{\checkmark\} + A \times \Id)(R) = \{(\checkmark, \checkmark)\} \cup \{(a,x),(a,y)) \mid a \in A, (x,y) \in R\}$
        \item $\Rel(\distf F)(R) = {\equiv}_{\Rel(\{\checkmark\} + A \times \Id)(R)}$
    \end{enumerate}
\end{lemma}
\begin{proof}
    Using the inductive definition of relation liftings from \cite[Lemma~3.6.7]{Sokolova:2005:Coalgebraic}.
\end{proof}
\begin{lemma}\label{lem:bisim_characterisation_sublemma}
    Let $\nu_1 \in \distf F X$, $\nu_2 \in \distf F Y$ and let $R \subseteq X \times Y$. The necessary and sufficient conditions for $(\nu_1, \nu_2) \in \Rel(\distf F)(R)$ are:
    \begin{enumerate}
        \item $\nu_1(\checkmark)=\nu_2(\checkmark)$
        \item For all $B \subseteq X$ and all $a \in A$, $\nu_1[\{a\}\times B] \leq \nu_2[\{a\}\times R(B)]$
        \item For all $C \subseteq Y$ and all $a \in A$, $\nu_2[\{a\}\times C] \leq \nu_1[\{a\}\times R^{-1}(C)]$
    \end{enumerate}
\end{lemma}
\begin{proof}
Assume that $(\nu_1, \nu_2) \in \Rel(\distf F)(R)$. By \cref{lem:concrete_liftings}, it is equivalent to $\nu_1$ and $\nu_2$ being related by the lifting of $\Rel(\{\checkmark\} + A \times \Id)(R)$. First, observe that $\Rel(\{\checkmark\} + A \times \Id)(R)(\{\checkmark\})= \{\checkmark\}$ and for all $a \in A$, and $B \subseteq X$, we have that
$\Rel(\{\checkmark\} + A \times \Id)(R)(\{a\} \times B) = \{a\} \times R(B)$. Symmetric conditions hold for $R^{-1}$.

First, we have that:
\begin{align*}
\nu_1(\checkmark) &= \nu_1[\{\checkmark\}] \\
&\leq \nu_2[R(\{\checkmark\})] \tag{\cref{lem:coupling_lemma}} \\
&\leq \nu_2[\{\checkmark\}] \\
&\leq \nu_1[R^{-1}(\{\checkmark\})] \\
&\leq \nu_1(\checkmark)
\end{align*}
which proves $\nu_1(\checkmark) = \nu_2(\checkmark)$. Now, pick an arbitrary $a \in A$ and $B \subseteq X$. We have that:
\begin{align*}
\nu_1[\{a\} \times B] &\leq \nu_2[\Rel(\{\checkmark\} + A \times \Id)(R)(\{a\} \times B)] \tag{\cref{lem:coupling_lemma}}\\
&=\nu_2[\{a\} \times R(B)]
\end{align*}
which proves (2). Symmetric reasoning involving $R^{-1}$ allows to prove (3).

For the converse, assume that conditions (1), (2) and (3) hold. Let $M \subseteq \{\checkmark\} + A \times X$. We can partition $M$ in the following way:
$$M = \{o \mid o \in \{\checkmark\} \cap M\} \cup \bigcup_{a \in A} \{a\} \times \{x \mid (a,x) \in M\}$$
We have the following:
\begin{align*}
\nu_1[B] &= \nu_1[\{o \mid o \in \{\checkmark\} \cap M\} ] + \sum_{a \in A} \nu_1[\{a\} \times \{x \mid (a,x) \in M\}] \\
&= [\checkmark \in M]~\nu_1(\checkmark) + \sum_{a \in A} \nu_1[\{a\} \times \{x \mid (a,x) \in M\}] \\
&\leq [\checkmark \in M]~\nu_2(\checkmark) + \sum_{a \in A} \nu_2[\{a\} \times R(\{x \mid (a,x) \in M\})] \tag{1 and 2}\\
&=\nu_2[\{o \mid o \in \{\checkmark\} \cap M\} ] + \sum_{a \in A} \nu_2[\{a\} \times R(\{x \mid (a,x) \in M\})] \\
&=\nu_2[\{o \mid o \in \{\checkmark\} \cap M\} ] + \nu_2[\Rel(\{\checkmark\} + A \times \Id)(R)((A \times X)\cap M)] \\
&=\nu_2[\Rel(\{\checkmark\} + A \times \Id)(R)(M)]
\end{align*}

Recall that relation liftings preserve inverse relations~\cite{Hughes:2004:Simulations} and therefore $\Rel(\{\checkmark\} + A \times \Id)(R)^{-1} = \Rel(\{\checkmark\} + A \times \Id)(R^{-1})$. A similar reasoning to the one before allows us to conclude that for all $N \subseteq \{\checkmark\} + A \times Y$ we have that 
$\nu_2[N] \leq \nu_1[\Rel(\{\checkmark\} + A \times \Id)(R)^{-1}(N)]$. Finally, we can apply \cref{lem:coupling_lemma} to conclude that $\nu_1 \equiv_{\Rel(\{\checkmark\} + A \times \Id)(R)} \nu_2$.
\end{proof}

\begin{lemma}\label{lem:characterisation_of_bisimulation}
    A relation ${R} \subseteq {X \times Y}$ is a bisimulation between $\distf F$-coalgebras $(X, \beta)$ and $(Y, \gamma)$ if and only if for all $(x,y) \in R$, we have that:
    \begin{enumerate}
        \item $\beta(x)(\checkmark) = \gamma(y)(\checkmark)$
        \item For all $a \in A$, and for all $B \subseteq X$, we have that: $$\beta(x)[\{a\} \times B] \leq \gamma(y)[\{a\} \times R(B)]  $$
        \item For all $a \in A$, and for all $C \subseteq Y$, we have that: $$\gamma(y)[\{a\} \times C] \leq \gamma(y)[\{a\} \times R^{-1}(C)]  $$ 
    \end{enumerate}
\end{lemma}
\begin{proof}
    Consequence of \cref{lem:relation_lifting_bisim} and \cref{lem:bisim_characterisation_sublemma}.
\end{proof}
\begin{lemma}\label{lem:simpler_characterisation}
    Let $(X, \beta)$ be a $\distf F$-coalgebra, $R \subseteq {X \times X}$ be an equivalence relation, $(x,y) \in R$ and $a \in A$.
    We have that: 
    \begin{enumerate}
        \item For all $G \subseteq X$, $\beta(x)[\{a\} \times G] \leq \beta(y)[\{a\} \times R(G)]$
        \item For all $H \subseteq X$, $\beta(y)[\{a\} \times H] \leq \beta(x)[\{a\} \times R^{-1}(H)]$
    \end{enumerate}
    if and only if for all equivalence classes $Q \in X/R$:
    $$
    \beta(x)[\{a\} \times Q] = \beta(y)[\{a\} \times Q]
    $$
\end{lemma}
\begin{proof}
    First assume that (1) and (2) do hold.
    We have that:
    \begin{align*}
        \beta(x)[\{a\} \times Q] &\leq \beta(y)[\{a\} \times R(Q)] \\
        &= \beta(y)[\{a\} \times Q] \tag{$R$ is an equivalence relation}\\
    \end{align*}
    Since $R=R^-1$ we can employ the symmetric reasoning and show that: 
    $$\beta(y)[\{a\} \times Q]\leq\beta(x)[\{a\} \times Q]$$ which allows us to conclude: $$\beta(x)[\{a\} \times Q]=\beta(y)[\{a\} \times Q]$$

    For the converse, let $G \subseteq X$ be an arbitrary set. 
    Let \(G / {R}\) be the quotient of \(G\) by the relation \(R\) and let \(X / {R}\) be the quotient of \(X\) by \(R\).
    Observe that \(G / {R}\) is a partition of \(G\) and \(X / {R}\) is a partition of \(X\).

    Moreover, for each equivalence class \(P \in G / {R}\), there exists an equivalence class \(Q_P \in X / {R}\), such that \(P \subseteq Q_P = R(P)\).
    Because of monotonicity, we also have that \(\beta(x)[\{a\} \times P] \leq \beta(x)[\{a\} \times Q_P]\).
    By \(\sigma\)-additivity we have that:
	\begin{align*}
		\beta(x)[\{a\} \times G] &= \beta(x)\left[\{a\} \times \bigcup_{P \in G / {R}} P \right] \\
        &= \sum_{P \in G / {R}} \beta(x)[\{a\} \times P]\\
        &\leq \sum_{P \in G / {R}} \beta(y)[\{a\} \times Q_P]\\
        &= \beta(y)\left[\{a\} \times \bigcup_{P \in G / {R}} Q_P\right]\\
        &= \beta(y)\left[\{a\} \times \bigcup_{P \in G / {R}} R(P)\right]\\
        &= \beta(y)[\{a\} \times R(G)]
	\end{align*}
    We can show (2) via symmetric line of reasoning to the one above.
\end{proof}
Given a set $Q \subseteq \Exp$ and an expression $f \in \Exp$, we will write: 
$$
{Q}/{f} = \{e \in \Exp \mid e \seq f \in Q\}
$$
\begin{lemma}\label{lem:cutting_postfixes}
    If \(R \subseteq \Exp \times \Exp\) is a congruence relation and \(({e}, {f}) \in R\) then for all $G \subseteq \Exp$, \(R(G / {{e}}) \subseteq R(G) / {f} \).
\end{lemma}
\begin{proof}
    If \({g} \in  R(G / {e})\), then there exists some \({h} \in G / {e}\) such that \(({h}, {g})\in R\).
    Since \({h} \in G / {e}\), also \({h} \seq {e} \in G\).
    Because \(R\) is a congruence relation, we have that \(({h}\seq {e}, {g}\seq {f})\in R\) and hence \({g}\seq {f} \in R(G)\), which in turn implies that \({g} \in G / {f}\).
\end{proof}

\begin{lemma}\label{lem:associativity_of_cutting}
    Let \(e,f \in \Exp\) and let \(Q \in \Exp / {\equiv_0}\).
    Now \((Q/{f})/{e} = Q / {e}\seq{f}\)
\end{lemma}
\begin{proof}
    Let \({g} \in Q\); we derive as follows
    \begin{align*}
        {g} \in (Q/ {f})/{e}
            &\iff {g} \seq {e} \in Q/{f}
            \iff ({g} \seq {e}) \seq{{f}} \in Q\\
            &\iff {g} \seq ({e} \seq {f}) \in Q
            \iff {g} \in Q / {{e} \seq {f}}
    \end{align*}
    Here, the second to last step follows by associativity (\textbf{S}).
\end{proof}
\begin{lemma}\label{lem:swapping_ends}
Let $e,f \in \Exp$, $R \subseteq {\Exp \times \Exp}$ a congruence relation and let $Q \in {\Exp}/{R}$. If $(e,f) \in R$, then  $Q/e = Q/f$
\end{lemma}
\begin{proof}
    $$
    {g} \in Q / {e} \iff g \seq e \in Q \iff g \seq f \in Q \iff g \in Q/f
    $$
\end{proof}
\begin{lemma}\label{lem:simpler_sequencing_semantics}
    For all $e,f \in \Exp$, $a \in A$, $G \subseteq \Exp$ we have that:
    $$
    \partial(e \seq f)[\{a\} \times G] = \partial(e)(\checkmark)\partial(f)[\{a\} \times G] + \partial(e)[\{a\} \times {G}/{f}]
    $$
\end{lemma}
\begin{proof}
    A straightforward calculation using the definition of the Antimirov derivative.
    \begin{align*}
        \partial(e \seq f)[\{a\} \times G] &= \sum_{g \in G} \partial(e \seq f)(a,g) \\
        &=\sum_{g \in G} \partial(e)(\checkmark)\partial(f)(a,g) + \sum_{g \seq f \in G} \partial(e)(a,g) \\
        &=\partial(e)(\checkmark)\sum_{g \in G}\partial(f)(a,g) + \sum_{g \seq f \in G} \partial(e)(a,g) \\
        &=\partial(e)(\checkmark)\partial(f)[\{a\} \times G] + \partial(e)[\{a\} \times {G}/{f}] \\
    \end{align*}
\end{proof}
\begin{lemma}\label{lem:simpler_loop_semantics}
    Let $e \in \Exp$, $r \in [0,1]$, $G \subseteq \Exp$ and $r\partial(e)(\checkmark) \neq 1$. We have that:
    $$
    \partial\left(e^{[r]}\right)[\{a\} \times G] =\frac{r\partial(e)[\{a\} \times {G}/{e^{[r]}}]}{1-r\partial(e)(\checkmark)}
    $$
\end{lemma}
\begin{proof}
    A straightforward calculation using the definition of the Antimirov derivative.
    \begin{align*}
        \partial\left(e^{[r]}\right)[\{a\} \times G] &= \sum_{g \in G}\partial\left(e^{[r]}\right)(a,g)\\
        &=\sum_{g \seq e^{[r]} \in G} \frac{r\partial(e)(a,g)}{1-r\partial(e)(\checkmark)}\\
        &=\frac{r\partial(e)[\{a\} \times {G}/{e^{[r]}}]}{1-r\partial(e)(\checkmark)}
    \end{align*}
\end{proof}
\exitoperatorlemma*  
\begin{proof}
By structural induction. The base cases $E(\zero)=0=\partial(\zero)(\checkmark)$, $E(\one)=1=\partial(\one)(\checkmark)$ and $E(a)=0=\partial(a)(\checkmark)$ hold immediately. For the convex sum case, we have that:
$$E(e \oplus_p f) = pE(e) + (1-p)E(f) = p\partial(e)(\checkmark) + (1-p)\partial(f)(\checkmark) = \partial(e \oplus_p f)(\checkmark)$$
As for the sequential composition, we have that:
$$E(e \seq f)=E(e)E(f)=\partial(e)(\checkmark)\partial(f) (\checkmark) = (\partial (e) \lhd f)(\checkmark)=\partial(e \seq f)(\checkmark)$$
Moving on to loops, first consider the case when $\partial(e)(\checkmark)=1$ and the loop probability is $1$. By induction hypothesis, also $E(e)=1$.
$$E\left(e^{[1]}\right)=0=\partial\left(e^{[1]}\right)(\checkmark)$$
Otherwise, we have that:
$$E(e^{[p]})=\frac{1-p}{1-pE(e)}=\frac{1-p}{1-p\partial(e)(\checkmark)}=\partial\left(e^{[p]}\right)(\checkmark)$$
\end{proof}
\soundnessbisim*
\begin{proof}
    By structural induction on the length derivation of $\equiv_0$, we show that all conditions of $\cref{lem:characterisation_of_bisimulation}$ are satisfied. In all of the cases, besides the last two, we will show simpler conditions from \cref{lem:simpler_characterisation}.

    For the first few cases, we will rely on even simpler characterisation. In particular, observe that if for some $e,f \in \Exp$, $\partial(e) = \partial(f)$, then immediately $\partial(e)(\checkmark)=\partial(f)(\checkmark)$ and for all $a \in A$, $Q \in {\Exp}/{\equiv_0}$ we have that: $$\partial(e)[\{a\}\times Q] = \partial(f)[\{a\}\times Q]$$
    
    In other words, equality of distributions given by applying the coalgebra structure map to some two states implies that they are bisimilar.
    \begin{itemize}
        \item[] \fbox{ $e \oplus_1 f \equiv_0 e$}
        For all $e,f \in \Exp$, $x \in F \Exp $ we have that: 
        $$
        \partial(e \oplus_1 f)(x) = 1\partial(e)(x) + 0 \partial(f)(x) = \partial(e)(x)
        $$
        Since $\partial(e \oplus_1 f) = \partial(e)$, then $e \oplus_1 f$ and $e$ are bisimilar.

        \item[]  \fbox{ $e \oplus_p f \equiv_0 f \oplus_{\ol{p}} e$}
        For all $e,f  \in \Exp$, $p \in [0,1]$ and $x \in F \Exp$ we have that
        \begin{align*}
            \partial(e \oplus_p f)(x) &= p \partial(e)(x) + (1-p) \partial(f)\\
            &= (1-p) \partial(f)(x)  + (1-(1-p))\partial(e)(x)\\ &= \partial(f \oplus_{\ol{p}} e)(x) 
        \end{align*}

        Since $\partial(e \oplus_p f) = \partial(f \oplus_{\ol{p}} e )$, then $e \oplus_p f$ and $f \oplus_{\ol{p}} e$ are bisimilar.

        \item[] \fbox{ $(e \oplus_p f) \oplus_{q} g \equiv_0 e \oplus_{pq} \left( f \oplus_{\frac{\ol{p}q}{1-pq}} g\right)$}
        For all $e,f,g  \in \Exp$, $p,q \in [0,1]$ such that $pq \neq 1$ and for all $x \in F \Exp$ we have that:
        \begin{align*}
            \partial\left((e \oplus_p f) \oplus_{q} g \right)(x) &= q\partial(e \oplus_p f)(x) + (1-q) \partial(g)(x) \\
            &=pq \partial(e)(x) + (1-p)q \partial(f)(x) + (1-q) \partial(g)(x)\\
            &=pq \partial(e)(x)\\ &\quad\quad+ (1-pq)\left(\frac{(1-p)q}{1-pq} \partial(f)(x) + \frac{1-q}{1-pq} \partial(g)(x)\right)\\
            &=pq \partial(e)(x) + (1-pq)\partial\left(f \oplus_{\frac{\ol{p}q}{1-pq}}g \right)(x)\\
            &=\partial\left(e \oplus_{pq} \left( f \oplus_{\frac{\ol{p}q}{1-pq}} g\right)\right)(x)
        \end{align*}
        Since $\partial\left((e \oplus_p f) \oplus_{q} g \right) = \partial\left(e \oplus_{pq} \left( f \oplus_{\frac{\ol{p}q}{1-pq}} g\right)\right)$, then $(e \oplus_p f) \oplus_{q} g$ and $e \oplus_{pq} \left( f \oplus_{\frac{\ol{p}q}{1-pq}} g\right)$ are bisimilar.

        \item[] \fbox{$e \seq \one \equiv_0 e$}
        For all $e \in \Exp$, we have that
        $$\partial(e \seq \one)(\checkmark) = \partial(e)(\checkmark)\delta_{\checkmark}(\checkmark)=\partial(e)(\checkmark)$$
        For all $a \in A$ and $Q \in {\Exp}/{\equiv_0}$ we have that:
        \begin{align*}
            \partial(e \seq \one)[\{a\} \times Q] &= \partial(e)[\{a\} \times {Q}/{\one}] + \partial(e)(\checkmark)\partial(\one)[\{a\} \times Q] \tag{\cref{lem:simpler_sequencing_semantics}}\\
            &= \partial(e)[\{a\} \times {Q}/{\one}] \\
            &= \sum_{q \seq \one \in Q}\partial(e)(a,q) \\
            &= \sum_{q \in Q}\partial(e)(a,q) \tag{\textbf{S1}}\\
            &= \partial(e)[\{a\} \times Q]
        \end{align*}

        \item[] \fbox{$\one \seq e \equiv_0 e$}
        For all $e \in \Exp$, we have that:
        $$\partial(\one \seq e)(\checkmark) = \delta_{\checkmark}(\checkmark)\partial(e)(\checkmark)=\partial(e)(\checkmark)$$

        For all $a \in A$ and $Q \in {\Exp}/{\equiv_0}$ we have that:
        \begin{align*}
            \partial(\one \seq e)[\{a\} \times Q] &= \partial(\one)[\{a\} \times {Q}/{e}] + \partial(\one)(\checkmark)\partial(e)[\{a\} \times Q] \tag{\cref{lem:simpler_sequencing_semantics}}\\
            &= \partial(e)[\{a\} \times {Q}]
        \end{align*}

        \item[] \fbox{$\zero\seq e \equiv_0 \zero$}
        For all $e \in \Exp$ we have that: 
        $$\partial(\zero \seq e)(\checkmark) = \partial(\zero)(\checkmark)\partial(e)(\checkmark) = 0 = \partial(\zero)(\checkmark)$$
        
        For all $a \in A$ and $Q \in {\Exp}/{\equiv_0}$ we have that:
        \begin{align*}
            \partial(\zero \seq e)[\{a\} \times Q] &= \partial(\zero)[\{a\} \times Q/{e}] + \partial(\zero)(\checkmark)\partial(e)[\{a\}\times Q]\\
            &=0=\partial(\zero)[\{a\} \times Q]
        \end{align*}

        \item[] \fbox{$e \seq (f \seq g) \equiv_0 (e \seq f) \seq g$}
        For all $e,f,g \in \Exp$ we have that:
        \begin{align*}
            \partial(e \seq (f \seq g))(\checkmark) &= \partial(e)(\checkmark)\partial(f \seq g)(\checkmark)\\
            &=\partial(e)(\checkmark)\partial(f)(\checkmark)\partial(g)(\checkmark)\\
            &=\partial(e \seq f)(\checkmark)\partial(g)(\checkmark)\\
            &=\partial((e \seq f) \seq g)(\checkmark)
        \end{align*}

        For all $a \in A$ and $Q \in {\Exp}/{\equiv_0}$ we have that:
        \begin{align*}
            \partial(e \seq (f \seq g))[\{a\} \times Q] &= \partial(e)[\{a\} \times {Q}/{f \seq g}] + \partial(e)(\checkmark)\partial(f \seq g)[\{a\} \times Q] \tag{\cref{lem:simpler_sequencing_semantics}} \\
            &= \partial(e)[\{a\} \times {{Q}/{g}}/{f}] + \partial(e)(\checkmark)\partial(f)[\{a\} \times Q/{g}]\\
            &\quad\quad+\partial(e)(\checkmark)\partial(f)(\checkmark)\partial(g)[\{a\} \times Q] \tag{\cref{lem:associativity_of_cutting}}\\
            &=\partial(e \seq f)[\{a\} \times {Q}/{f}] + \partial(e \seq f)(\checkmark)\partial(g)[\{a\} \times Q] \tag{\cref{lem:simpler_sequencing_semantics}} \\
            &=\partial((e \seq f) \seq g)[\{a\} \times Q] \tag{\cref{lem:simpler_sequencing_semantics}}
        \end{align*}

        \item[] \fbox{$(e \oplus_p f) \seq g \equiv_0 e \seq g \oplus_p f \seq g$}
        For all $e,f,g \in \Exp$ and $p \in [0,1]$ we have that:
        \begin{align*}
            \partial((e \oplus_p f) \seq g)(\checkmark) &= \partial(e \oplus_p f)(\checkmark)\partial(g)(\checkmark)\\
            &=p \partial(e)(\checkmark)\partial(g)(\checkmark) + (1-p)\partial(f)(\checkmark)\partial(g)(\checkmark)\\
            &= p \partial(e \seq g)(\checkmark) + (1-p)\partial(f \seq g)(\checkmark)\\
            &= \partial(e \seq g \oplus_p f \seq g)(\checkmark)
        \end{align*}
        For all $a \in A$ and $Q \in {\Exp}/{\equiv_0}$ we have that:
        \begin{align*}
            &\partial((e \oplus_p f) \seq g )[\{a\} \times Q]\\ &\quad\quad= \partial(e \oplus_p f)[\{a\} \times {Q}/{g}] + \partial(e \oplus_p f)(\checkmark)\partial(g)[\{a\} \times Q] \tag{\cref{lem:simpler_sequencing_semantics}}\\
            &\quad\quad=p \partial(e)[\{a\} \times Q ] + p \partial(e)(\checkmark)[\{a\} \times Q]\\
            &\quad\quad\quad\quad (1-p) \partial(f)[\{a\} \times Q ] + (1-p) \partial(f)(\checkmark)[\{a\} \times Q] \\
            &\quad\quad=p \partial(e \seq g)[\{a\} \times Q] + (1-p)\partial(f \seq g)[\{a\} \times Q] \tag{\cref{lem:simpler_sequencing_semantics}}\\
            &\quad\quad=\partial(e \seq g \oplus_p f \seq g)[\{a\} \times Q]
        \end{align*}

        \item[] \fbox{$e^{[p]} \equiv_0 e \seq e^{[p]} \oplus_p \one$} Let $e \in \Exp$ and $p \in [0,1]$. We distinguish two subcases. If $\partial(e)(\checkmark)=1$ and $p=1$, then:
        $$
        \partial(e^{[p]})(\checkmark) = 0 = \partial(e)(\checkmark)\partial(e^{[p]})(\checkmark) = \partial(e \seq e^{[p]} \oplus_p \one)(\checkmark)
        $$
        For all $a \in A$ and $Q \in {\Exp}/{\equiv_0}$ we have that:
        \begin{align*}
            \partial(e^{[p]})[\{a\} \times Q] &= 0 \\
            &= \partial(e)[\{a\} \times Q] + \partial(e^{[p]})[\{a\} \times Q] \\
            &= \partial(e)[\{a\} \times Q/{e^{[p]}}] + \partial(e)(\checkmark)\partial(e^{[p]})[\{a\} \times Q]\\
            &= \partial(e \seq e^{[p]})[\{a\} \times Q]\\
            &= \partial(e \seq e^{[p]} \oplus_p \one)[\{a\} \times Q]
        \end{align*}
        From now on, we can safely assume that $p\partial(e)(\checkmark)\neq 1$. We have that:
        \begin{align*}
            \partial(e^{[p]})(\checkmark) &= \frac{1-p}{1-p\partial(e)(\checkmark)}\\
            &=\frac{(1-p)(\one + p\partial(e)(\checkmark) - p\partial(e)(\checkmark))}{1-p\partial(e)(\checkmark)}\\
            &=\frac{(1-p)( p\partial(e)(\checkmark))}{1-p\partial(e)(\checkmark)} + \frac{(1-p)(1-p\partial(e)(\checkmark))}{1-p\partial(e)(\checkmark)}\\
            &= p\partial(e)(\checkmark)\frac{1-p}{1-p\partial(e)(\checkmark)} + (1-p)\\
            &= p\partial(e)(\checkmark)\partial(e^{[p]})(\checkmark) + (1-p)\\
            &= \partial\left(e \seq e^{[p]} \oplus_p \one\right)(\checkmark)
        \end{align*}
        Let $a \in A$ and $Q \in {\Exp}/{\equiv_0}$. We have that:
        \begin{align*}
            \partial(e^{[p]})[\{a\} \times Q] &= \frac{p\partial(e)[\{a\} \times {Q}/{e^{[p]}}]}{1-p\partial(e)(\checkmark)} \tag{\cref{lem:simpler_loop_semantics}} \\
            &= p\partial(e)[\{a\} \times {Q}/{e^{[p]}}]\frac{1}{1-p\partial(e)(\checkmark)}\\
            &= p\partial(e)[\{a\} \times {Q}/{e^{[p]}}]\frac{1-p\partial(e)(\checkmark) + p\partial(e)(\checkmark)}{1-p\partial(e)(\checkmark)}\\
            &= p \partial(e)[\{a\} \times Q/{e^{[p]}}] \left( 1 + \frac{p\partial(e)(\checkmark)}{1-p\partial(e)(\checkmark)}\right)\\
            &= p \partial(e)[\{a\} \times Q/{e^{[p]}}] + p\partial(e)(\checkmark)\frac{p\partial(e)[\{a\} \times {Q}/{e^{[p]}}]}{1-p\partial(e)(\checkmark)}\\
            &= p \partial(e)[\{a\} \times Q/{e^{[p]}}] + p\partial(e)(\checkmark)\partial(e^{[p]})[\{a\} \times Q] \tag{\cref{lem:simpler_loop_semantics}}\\
            &= p \partial(e \seq e^{[p]})[\{a\} \times Q]\tag{\cref{lem:simpler_sequencing_semantics}}\\
            &=\partial(e \seq e^{[p]} \oplus_p \one)[\{a\} \times Q]
        \end{align*}

        \item[] \fbox{$(e \oplus_p \one)^{[q]} \equiv_0 e^{[\frac{pq}{1-\ol{p}q}]}$}
        \item[] Let $e \in \Exp$ and let $p,q \in [0,1]$ such that $(1-p)q \neq 1$. Observe, that in such a situation $\frac{pq}{1-\ol{p}q}\neq 1$.
        First, consider the following:
        \begin{align*}
            \partial\left((e \oplus_p \one)^{[q]}\right)(\checkmark) &= \frac{1-q}{1-q\partial(e \oplus_p \one)(\checkmark)} \\
            &= \frac{1-q}{1-q(1-p)-pq\partial(e)(\checkmark)} \\
            &= \frac{1-q}{(1-q(1-p))\left(1-\frac{pq}{1-q(1-p)}\partial(e)(\checkmark)\right)}\\
            &=\frac{\frac{1-q}{1-q(1-p)}}{1-\frac{pq}{1-q(1-p)}\partial(e)(\checkmark)}\\
            &=\frac{1-\frac{pq}{1-q(1-p)}}{1-\frac{pq}{1-q(1-p)}\partial(e)(\checkmark)}\\
            &=\partial \left(e^{\left[\frac{pq}{1-\ol{p}q}\right]}\right)(\checkmark)
        \end{align*}

        For all $a \in A$ and $Q \in {\Exp}/{\equiv_0}$ we have that the following holds:
        \begin{align*}
            \partial \left( (e \oplus_p \one)^{[q]}\right)[\{a\} \times Q] &= \frac{q\partial(e \oplus_p \one)[\{a\} \times {Q}/{(e \oplus_p \one)^{[q]}}]}{1-q\partial(e \oplus_p \one)(\checkmark)} \tag{\cref{lem:simpler_loop_semantics}} \\
            &= \frac{pq\partial(e)[\{a\} \times {Q}/{(e \oplus_p \one)^{[q]}}]}{1-(1-p)q-pq\partial(e)(\checkmark)} \\
            &= \frac{pq\partial(e)[\{a\} \times {Q}/{(e \oplus_p 1)^{[q]}}]}{(1-(1-p)q)\left(1-\frac{pq}{1-(1-p)q}\partial(e)(\checkmark)\right)} \\
            &= \frac{\frac{pq}{1-(1-p)q}\partial(e)[\{a\} \times {Q}/{(e \oplus_p 1)^{[q]}}]}{(1-(1-p)q)\left(1-\frac{pq}{1-(1-p)q}\partial(e)(\checkmark)\right)} \\
            &= \frac{\frac{pq}{1-(1-p)q}\partial(e)[\{a\} \times {Q}/{e^{\left[\frac{pq}{1-\ol{p}q}\right]}}]}{(1-(1-p)q)\left(1-\frac{pq}{1-(1-p)q}\partial(e)(\checkmark)\right)} \tag{\cref{lem:swapping_ends}}\\
            &=\partial\left(e^{\left[\frac{pq}{1-\ol{p}q}\right]}\right)[\{a\} \times Q]
        \end{align*}

        \item[] \fbox{From $g \equiv_0 e \seq g \oplus_p f$ and $E(g)=0$ derive $g \equiv_0 e^{[p]} \seq f$} Let $e,f,g \in \Exp$, such that $g \equiv e\seq g \oplus_p f$ and $E(e) = 0$. Recall that by \cref{lem:exit_operator_lemma}, we have that $\partial(e)(\checkmark)=0$. First, observe that:
        \begin{align*}
            \partial(g)(\checkmark) &= \partial(e \seq g \oplus_p f)(\checkmark) \tag{Induction hypothesis} \\
            &=p \partial(e)(\checkmark)\partial(g)(\checkmark) + (1-p)\partial(f)(\checkmark) \\
            &= (1-p)\partial(f)(\checkmark) \\
            &= \frac{1-p}{1-p\partial(e)(\checkmark)}\partial(f)(\checkmark) \\
            &= \partial(e^{[p]})(\checkmark)\partial(f)(\checkmark)\\
            &= \partial(e^{[p]} \seq f)(\checkmark)
        \end{align*}

        For all $a \in A$ and $Q \in {\Exp}/{\equiv_0}$ we have that:
        \begin{align*}
            \partial(g)[\{a\} \times Q] &= \partial(e \seq g \oplus_p f)[\{a\} \times Q] \\\tag{Induction hypothesis} \\
            &=p \partial(e \seq g)[\{a\} \times Q] + (1-p)\partial(f)[\{a\} \times Q]\\
            &= p \partial(e)[\{a\} \times Q/{g}] + p \partial(e)(\checkmark)\partial(g)[\{a\}\times Q]\\&\quad\quad + (1-p)\partial(f)[\{a\} \times Q] \\
            &= p\partial(e)[\{a\} \times Q/g] + (1-p)\partial(f)[\{a\} \times Q] \\
            &=p\partial(e)[\{a\} \times Q/{e^{[p]}\seq f}] + (1-p)\partial(f)[\{a\} \times Q] \tag{\cref{lem:swapping_ends}}\\
            &=p\partial(e)[\{a\} \times {(Q/{f})}/{e^{[p]}}] + (1-p)\partial(f)[\{a\} \times Q] \tag{\cref{lem:associativity_of_cutting}}\\
            &=\frac{p\partial(e)[\{a\} \times {(Q/{f})}/{e^{[p]}}]}{1-p\partial(e)(\checkmark)} + \frac{1-p}{1-p\partial(e)(\checkmark)}\partial(f)[\{a\} \times Q] \\
            &= \partial(e^{[p]})[\{a\} \times Q/f] + \partial(e^{[p]})(\checkmark)\partial(f)[\{a\} \times Q] \tag{\cref{lem:simpler_loop_semantics}}\\
            &=\partial(e^{[p]}\seq f)[\{a\} \times Q] \tag{\cref{lem:simpler_sequencing_semantics}}
        \end{align*}
    \end{itemize}

    \item[] \fbox{reflexivity, transitivity and symmetry} We omit the proof, as it is trivial.

    \item \fbox{From $e \equiv_0 g$ and $f \equiv_0 h$ derive that $e \oplus_p f \equiv_0 g \oplus_p h$}
    Let $e,f,g,h \in \Exp$, such that $e \equiv_0 g$ and $f \equiv_0 h$. We have that
    \begin{align*}
        \partial(e \oplus_p f)(\checkmark) &= p\partial(e)(\checkmark) + (1-p)\partial(f)(\checkmark) \\
        &= p\partial(g)(\checkmark) + (1-p)\partial(h)(\checkmark) \tag{Induction hypothesis}\\
        &= \partial(g \oplus_p h)(\checkmark)
    \end{align*}
    For all $a \in A$ and $Q \in {\Exp}/{\equiv_0}$ we have that:
    \begin{align*}
        \partial(e \oplus_p f)[\{a\} \times Q]&= p\partial(e)[\{a\} \times Q] + (1-p)\partial(f)[\{a\} \times Q] \\
        &= p\partial(g)[\{a\} \times Q] + (1-p)\partial(h)[\{a\} \times Q] \tag{Induction hypothesis}\\
        &= \partial(g \oplus_p h)[\{a\} \times Q]
    \end{align*}
    \item \fbox{From $e \equiv_0 g$ and $f \equiv_0 h$ derive that $e \seq f \equiv_0 g \seq h$} Let $e,f,g,h \in \Exp$, such that $e \equiv_0 g$ and $f \equiv_0 h$.
    We have that:
    \begin{align*}
        \partial(e \seq f)(\checkmark) &= \partial(e)(\checkmark)\partial(f)(\checkmark) \\
        &=\partial(g)(\checkmark)\partial(h)(\checkmark) \tag{Induction hypothesis} \\
        &=\partial(g \seq h)(\checkmark)
    \end{align*}
    For all $a \in A$ and $G \subseteq \Exp$, we have that:
    \begin{align*}
        \partial(e \seq f)[\{a\} \times G] &= \partial(e)[\{a\}\times G/{f}] + \partial(e)(\checkmark)\partial(f)[\{a\} \times G] \tag{\cref{lem:simpler_sequencing_semantics}}\\
        &\leq \partial(g)[\{a\} \times R(G/f)] + \partial(g)(\checkmark)\partial(h)[\{a\} \times R(G)]\\
        &\leq \partial(g)[\{a\} \times R(G)/h] + \partial(g)(\checkmark)\partial(h)[\{a\} \times R(G)] \tag{\cref{lem:cutting_postfixes}} \\
        &= \partial(g \seq h)[\{a\} \times R(G)] \tag{\cref{lem:simpler_sequencing_semantics}}
    \end{align*}

    Condition that $\partial(g \seq h)[\{a\} \times G] \leq \partial(e \seq f)[\{a\} \times R^{-1}(G)]$ can be shown by a symmetric argument.

    \item[] \fbox{From $e \equiv_0 f$ derive $e^{[p]} \equiv_0 f^{[p]}$}.Let $e,f \in \Exp$ such that $e\equiv_0 f$. We distinguish two subcases. First, consider the situation when $p=1$ and $\partial(e)(\checkmark)=1$. Observe that by induction hypothesis we have that $\partial(f)(\checkmark)=1$.
    For all $x \in F \Exp$, we have that:
    $$
    \partial(e^{[p]})(x) = 0 = \partial(f^{[p]})(x) 
    $$
    Since in this case $\partial(e^{[p]})=\partial(f^{[p]})$, $e^{[p]}$ and $f^{[p]}$ are bisimilar.

    From now on, we can safely assume that $p\partial(e)(\checkmark)\neq 1$ and $p\partial(f)(\checkmark)\neq 1$. We have that:
    \begin{align*}
        \partial \left( e^{[p]} \right) (\checkmark) &= \frac{1-p}{1-p\partial(e)(\checkmark)} \\
        &= \frac{1-p}{1-p\partial(f)(\checkmark)} \tag{Induction hypothesis}\\
        &= \partial \left( f^{[p]} \right) (\checkmark)
    \end{align*}

    For all $a \in A$ and $G \subseteq \Exp$ we have that:
    \begin{align*}
        \partial \left( e^{[p]} \right)[\{a\} \times G] &= \frac{p\partial(e)[\{a\}\times {{G}/{e^{[p]}}}]}{1-p\partial(e)(\checkmark)} \tag{\cref{lem:simpler_loop_semantics}} \\
        &\leq \frac{p\partial(f)[\{a\}\times R({{G}/{e^{[p]}}})]}{1-p\partial(f)(\checkmark)} \tag{Induction hypothesis}\\
        &\leq \frac{p\partial(f)[\{a\}\times {{R(G)}/{f^{[p]}}}]}{1-p\partial(f)(\checkmark)} \tag{\cref{lem:cutting_postfixes}} \\
        &= \partial \left( e^{[p]} \right)[\{a\} \times R(G)]
    \end{align*}

    Condition that $\partial\left(f^{[p]}\right)[\{a\} \times G] \leq \partial\left(e^{[p]}\right)[\{a\} \times R^{-1}(G)]$ can be shown by a symmetric argument.
\end{proof}
\section{Fundamental theorem}
\begin{lemma}\label{lem:sum_unrolling}
    Let $\{p_i\}_{i \in I}$ and $\{e_i\}_{i \in I}$ be collections indexed by a finite set $I$, such that for all $i \in I$, $p_i \in [0,1]$ and $e_i \in \Exp$. For any $j \in I$ it holds that:
    \[
        \bigoplus_{i \in I} p_i \cdot e_i \equiv_0 e_j \oplus_{p_j} \left(\bigoplus_{i \in I \setminus \{j\}} \frac{p_i}{\ol{p_j}}\cdot e_i\right) 
    \]
\end{lemma}
\begin{proof}
    In the edge case when $p_j = 1$ (and therefore $I = \{j\}$) we have that $\bigoplus_{i \in I} p_i \cdot e _i \equiv_0 e_j$ and therefore: 
    \begin{align*}
        e_j &\equiv_0 e_j \oplus_1 0 \tag{\textbf{C2}} \\
        &\equiv_0 e_j \oplus_{p_j} \left( \bigoplus_{i \in \emptyset} \frac{p_i}{\ol{p_j}} \cdot e_i \right) \tag{Def. of empty $n$-ary convex sum}\\
        &\equiv_0 e_j \oplus_{p_j} \left( \bigoplus_{i \in I \setminus \{j\}} \frac{p_i}{\ol{p_j}} \cdot e_i \right) \tag{$I = \{j\}$}\\
    \end{align*}
    Note that despite the fact that $\ol{p_j} = 0$, the $n$-ary sum on the right is well-defined as it ranges over an empty index set and thus division by zero never happens. 

    The remaining case of $p_j \neq 1$ holds immediately by the definition from \Cref{prop:binary}.
\end{proof}
\begin{lemma}\label{lem:safe_unrolling_lemma}
    For all $e \in \Exp$, 
    $$
    \bigoplus_{d \in \supp(\partial(e))} \partial(e)(d) \cdot \ex(d) \equiv_0 \one \oplus_{\partial(e)(\checkmark)} \left(\bigoplus_{d \in \supp(\partial(e)) \setminus \{\checkmark\}} \frac{\partial(e)(d)}{\ol{\partial(e)(\checkmark)}} \cdot \ex(d)\right)
    $$
\end{lemma}
\begin{proof}
    If $\supp(\partial(e)) = \emptyset$, then 
    \begin{align*}
        \bigoplus_{d \in \supp(\partial(e))} \partial(e)(d) \cdot \ex(d) &\equiv_0 \zero \tag{Def. of empty $n$-ary convex sum}\\
        &\equiv_0 0 \oplus_1 1 \tag{\textbf{C2}} \\
        &\equiv_0 1 \oplus_0 0 \tag{\textbf{C3}} \\
        &\equiv_0 1 \oplus_{\partial(e)(\checkmark)} 0 \tag{$\partial(e)(\checkmark)=0$}\\
        &\equiv_0 1 \oplus_{\partial(e)(\checkmark)} \left( 
            \bigoplus_{d \in \emptyset } \frac{\partial(e)(d)}{\ol{\partial(e)(\checkmark)}} \cdot \exp(d)\right)\\ 
        &\equiv_0 1 \oplus_{\partial(e)(\checkmark)} \left( 
                \bigoplus_{d \in \supp(\partial(e)(\checkmark))\setminus \{\checkmark\} } \frac{\partial(e)(d)}{\ol{\partial(e)(\checkmark)}} \cdot \exp(d)\right)\\ 
    \end{align*}
    
    The remaining case when $\supp(\partial(e)) \neq 0$ holds by \Cref{lem:sum_unrolling} and the fact that $\ex(\checkmark) = 1$.
\end{proof}
\fundamentaltheorem*
\begin{proof}
{
	\allowdisplaybreaks
    \begin{itemize}
        \item[]\fbox{$e = \zero$}
        \begin{align*}
            \zero &\equiv_0 \bigoplus_{d \in \emptyset} \partial(\zero)(d) \cdot \exp(d) \tag{\cref{prop:properties_of_positive_convex_algebras}}\\
            &\equiv_0 \bigoplus_{d \in \supp(\partial(\zero))} \partial(\zero)(d) \cdot \exp(d) \tag{$\supp(\partial(\zero))=\emptyset$}\\
        \end{align*}
        \item[]\fbox{$e=\one$}
        \begin{align*}
            \one &\equiv_0 \ex(\checkmark) \equiv_0 \bigoplus_{d \in \supp(\partial(\one))} \partial(\one)(d) \cdot \ex (d) \tag{$\partial(\one) = \delta_{\checkmark} $}
        \end{align*}
        \item[]\fbox{$e = a$}
        \begin{align*}
            a &\equiv_0 a \seq \one \equiv_0 \ex((a,\checkmark)) \equiv_0 \bigoplus_{d \in \supp(\partial(a))} \partial(a)(d) \cdot \ex (d) \tag{$\partial(a) = \delta_{(a, \checkmark)}$}
        \end{align*}
    \end{itemize}
    For the inductive steps, we have the following.
    \begin{itemize}
        \item[]\fbox{$e = f \oplus_p g $}
        \begin{align*}
            f \oplus_p g &\equiv_0 \left( \bigoplus_{d \in \supp(\partial(f))} \partial(f)(d) \cdot \ex(d) \right)\\&\quad\quad\quad\quad \oplus_p  \left( \bigoplus_{d \in \supp(\partial(g))} \partial(g)(d) \cdot \ex(g) \right) \tag{Induction hypothesis}\\
            &\equiv_0 \left( \bigoplus_{d \in \supp(\partial(f \oplus_p g))}\partial(f)(d) \cdot \ex(d) \right)\\&\quad\quad\quad\quad  \oplus_p  \left( \bigoplus_{d \in \supp(\partial(f \oplus_p g))} \partial(g)(d) \cdot \ex(g) \right) \tag{\cref{prop:properties_of_positive_convex_algebras}}\\
            &\equiv_0 p \cdot \left( \bigoplus_{d \in \supp(\partial(f \oplus_p g))} \partial(f)(d) \cdot \ex(d) \right)\\&\quad\quad\quad\quad \oplus \bar{p} \cdot  \left( \bigoplus_{d \in \supp(\partial(f \oplus_p g))} \partial(g)(d) \cdot \ex(g) \right) \\
            &\equiv_0 \bigoplus_{d \in \supp(\partial(f \oplus_p g))} \left(p\partial(f)(d) + \ol{p}\partial(g)(d)\right) \cdot \ex(d) \tag{Barycenter axiom}\\
             &\equiv_0 \bigoplus_{d \in \supp(\partial(f \oplus_p g))} \partial(f \oplus_p g)(d) \cdot \ex(d)
        \end{align*}
        \item[]\fbox{$e = f\seq g $}
        \begin{align*}
            f \seq g &\equiv_0 \left(\bigoplus_{d \in \supp(\partial(f))} \partial(f)(d) \cdot \ex(d) \right) \seq g \tag{Induction hypothesis} \\
            &\equiv_0 \left(\one \oplus_{\partial(f)(\checkmark)} \left(\bigoplus_{d \in \supp(\partial(f)) \setminus \{\checkmark\}} \frac{\partial(f)(d)}{\ol{\partial(f)(\checkmark)}} \cdot \ex(d) \right) \right) \seq g \tag{\cref{lem:safe_unrolling_lemma}}\\
            &\equiv_0 \left(\one\seq g \oplus_{\partial(f)(\checkmark)} \left(\bigoplus_{d \in \supp(\partial(f)) \setminus \{\checkmark\}} \frac{\partial(f)(d)}{\ol{\partial(f)(\checkmark)}} \cdot \ex(d)  \right)\seq g \right) \tag{\textbf{D1}}\\
            &\equiv_0 g \oplus_{\partial(f)(\checkmark)} \left(\bigoplus_{d \in \supp(\partial(f)) \setminus \{\checkmark\}} \frac{\partial(f)(d)}{\ol{\partial(f)(\checkmark)}} \cdot \ex(d)\seq g  \right)\\ \tag{\cref{lem:generalised_right_distributivity} and \textbf{1S}} \\
            &\equiv_0 \left(\bigoplus_{d \in \supp(\partial(g))}\partial(g)(d) \cdot \ex(d)\right)\\&\quad\quad\quad\quad \oplus_{\partial(f)(\checkmark)} \left(\bigoplus_{d \in \supp(\partial(f)) \setminus \{\checkmark\}} \frac{\partial(f)(d)}{\ol{\partial(f)(\checkmark)}} \cdot \ex(d)\seq g  \right) \tag{Induction hypothesis}\\
            &\equiv_0 \partial(f)(\checkmark) \cdot \left( \bigoplus_{d \in \supp(\partial(g))}\partial(g)(d) \cdot \ex(d)\right)\\&\quad\quad\quad\quad \oplus \ol{\partial(f)(\checkmark)} \cdot \left(\bigoplus_{d \in \supp(\partial(f)) \setminus \{\checkmark\}} \frac{\partial(f)(d)}{\ol{\partial(f)(\checkmark)}} \cdot \ex(d)\seq g \right)\\
            &\equiv_0 \partial(f)(\checkmark) \cdot \left( \bigoplus_{d \in \supp(\partial(f \seq g))}\partial(g)(d) \cdot \ex(d)\right)\\&\quad\quad\quad\quad \ol{\partial(f)(\checkmark)} \cdot \left(\bigoplus_{d \in \supp(\partial(f)) \setminus \{\checkmark\}} \frac{\partial(f)(d)}{\ol{\partial(f)(\checkmark)}} \cdot \ex(d)\seq g \right) \tag{\cref{prop:properties_of_positive_convex_algebras}}
        \end{align*}
        Now, we simplify the subexpression on the right part of the convex sum. Define
        $n : F \Exp \to [0,1]$ to be: 
        $$n(d) = \begin{cases}
            \partial(f)(a, f') & d = (a, f' \seq g) \\
            0 & \text{otherwise}
        \end{cases}$$
        Using \cref{prop:properties_of_positive_convex_algebras} and above definition, we have that:
        \begin{align*}
            \bigoplus_{d \in \supp(\partial(f)) \setminus \{\checkmark\}} \frac{\partial(f)(d)}{\ol{\partial(f)(\checkmark)}} \cdot \ex(d)\seq g &\equiv_0 \bigoplus_{d \in \supp(\partial(f \seq g))} \frac{n(d)}{\ol{\partial(f)(\checkmark)}} \cdot \ex(d)
        \end{align*}
        Combining it with the previous derivation, using the barycenter axiom, we can show that:
        $$f \seq g \equiv_0 \bigoplus_{d \in \supp(\partial(f \seq g))} \left(\partial(f)(\checkmark)\partial(g)(d) + n(d)\right) \cdot \ex(d) $$
        Observe, that for $d = (a, f' \seq g)$, we have that:
        $$\partial(f)(\checkmark)\partial(g)(d) + n(d) = \partial(f)(\checkmark)\partial(g)(a, f' \seq g) + \partial(f)(a, f') = \partial(f \seq g)(d)$$
        When $d = \checkmark$, we have that:
        $$\partial(f)(\checkmark)\partial(g)(d) + n(d) = \partial(f)(\checkmark)\partial(g)(d) = \partial(f \seq g)(d)$$
        In the remaining cases both functions assign $0$ to $d$. Hence, we have that:
        $$f \seq g \equiv_0 \bigoplus_{d \in \supp(\partial(f \seq g))} \partial(f \seq g)(d) \cdot \ex(d)$$
        which completes this case.
        \item[] \fbox{$e = f^{[p]}$}
        
        First, consider the situation when $\partial(f)(\checkmark) = 1$ and $p = 1$. 
        \begin{align*}
            f^{[p]} &\equiv_0 \left(\bigoplus_{d \in \supp(\partial(f))} \partial(f)(d) \cdot \ex(d) \right)^{[1]} \tag{Induction hypothesis}\\
            &\equiv_0 \one^{[1]} \tag{$\partial(f)(\checkmark)=1$}\\
            &\equiv_0 \zero \tag{\textbf{Div}}\\
            &\equiv_0 \bigoplus_{d \in \supp \left(\partial\left(f^{[1]}\right)\right)} \partial\left(f^{[1]}\right)(d) \cdot \ex(d)
        \end{align*}
        Otherwise, we start by applying the tightening axiom to the loop body in the following way:
        \begin{align*}
            f^{[p]} &\equiv_0 \left(\bigoplus_{d \in \supp(\partial(f))} \partial(f)(d)\cdot \ex(d)\right)^{[p]} \tag{Induction hypothesis} \\
            &\equiv_0 \left(\one \oplus_{\partial(f)(\checkmark)} \left( \bigoplus_{d \in \supp(\partial(f)) \setminus \{\checkmark\}} \frac{\partial(f)(d)}{\ol{\partial(f)(\checkmark)}}\cdot \ex(d)\right)\right)^{[p]} \tag{\cref{lem:safe_unrolling_lemma}}\\
            &\equiv_0 \left(\left( \bigoplus_{d \in \supp(\partial(f)) \setminus \{\checkmark\}} \frac{\partial(f)(d)}{\ol{\partial(f)(\checkmark)}}\cdot \ex(d)\right) \oplus_{\ol{\partial(f)(\checkmark)}} \one\right)^{[p]} \\
            &\equiv_0 \left(\bigoplus_{d \in \supp(\partial(f)) \setminus \{\checkmark\}} \frac{\partial(f)(d)}{\ol{\partial(f)(\checkmark)}} \cdot \ex(d)\right)^{\left[\frac{\ol{\partial(f)(\checkmark)}p}{1-\partial(f)(\checkmark)p}\right]} \tag{\textbf{Tight}}\\
        \end{align*}
        As a shorthand, we will write $g^{[r]}$ for the expression obtained through the above derivation. We continue, by applying the \textbf{Unroll} axiom.
        \begin{align*}
            g^{[r]}&\equiv_0 \left(\bigoplus_{d \in \supp(\partial(f)) \setminus \{\checkmark\}} \frac{\partial(f)(d)}{\ol{\partial(f)(\checkmark)}} \cdot \ex(d)\right)\seq g^{[r]} \oplus_{\frac{\ol{\partial(f)(\checkmark)}p}{1-\partial(f)(\checkmark)p}} \one \tag{\textbf{Unroll}}\\
            &\equiv_0 \left(\bigoplus_{d \in \supp(\partial(f)) \setminus \{\checkmark\}} \frac{\partial(f)(d)}{\ol{\partial(f)(\checkmark)}} \cdot \ex(d)\right)\seq f^{[p]} \oplus_{\frac{\ol{\partial(f)(\checkmark)}p}{1-\partial(f)(\checkmark)p}} \one \tag{$f^{[p]} \equiv_0 g^{[r]}$}\\
            &\equiv_0 \left(\bigoplus_{d \in \supp(\partial(f)) \setminus \{\checkmark\}} \frac{\partial(f)(d)}{\ol{\partial(f)(\checkmark)}} \cdot \ex(d)\seq f^{[p]}\right) \oplus_{\frac{\ol{\partial(f)(\checkmark)}p}{1-\partial(f)(\checkmark)p}} \one \tag{\cref{lem:generalised_right_distributivity}} \\
            &\equiv_0 \frac{\ol{\partial(f)(\checkmark)}p}{1-\partial(f)(\checkmark)p} \cdot \left(\bigoplus_{d \in \supp(\partial(f)) \setminus \{\checkmark\}} \frac{\partial(f)(d)}{\ol{\partial(f)(\checkmark)}} \cdot \ex(d)\seq f^{[p]}\right)\\&\quad\quad\quad\quad\oplus  \frac{1-p}{1-\partial(f)(\checkmark)p} \cdot \one \\
            &\equiv_0  \frac{\ol{\partial(f)(\checkmark)}p}{1-\partial(f)(\checkmark)p} \cdot \left(\bigoplus_{d \in \supp(\partial(f)) \setminus \{\checkmark\}} \frac{\partial(f)(d)}{\ol{\partial(f)(\checkmark)}} \cdot \ex(d)\seq f^{[p]}\right)\\&\quad\quad\quad\quad\oplus  \frac{1-p}{1-\partial(f)(\checkmark)p} \cdot \left(\bigoplus_{d \in \supp \left(\partial\left(f^{[p]}\right)\right)} \delta_{\checkmark}(d) \cdot \ex(d)\right)
        \end{align*}
        Now, we simplify the left hand side of the binary convex sum. Let $n : F \Exp \to [0,1]$ be defined by the following:
        $$n(d) = \begin{cases}
            {\partial(f)(a,f')} & d = \left(a, f' \seq f^{[p]}\right) \\
            0 & \text{otherwise}
        \end{cases}$$
        Using \cref{prop:properties_of_positive_convex_algebras} and above definition, we have that:
        $$
            \bigoplus_{d \in \supp(\partial(f)) \setminus \{\checkmark\}} \frac{\partial(f)(d)}{\ol{\partial(f)(\checkmark)}} \cdot \ex(d)\seq f^{[p]} \equiv_0 \bigoplus_{d \in \supp \left(\partial\left(f^{[p]}\right)\right)} \frac{n(d)}{\ol{\partial(f)(\checkmark)}} \cdot \ex(d)
        $$
        Now, we can combine the above with the previous derivation and use the barycenter axiom to show that:
        $$
        f^{[p]} \equiv_0 \bigoplus_{d \in \supp \left(\partial\left(f^{[p]}\right)\right)} \left(\frac{pn(d)}{1-\partial(f)(\checkmark)p} + \frac{(1-p)\delta_{\checkmark}(d)}{1-\partial(f)(\checkmark)p}\right) \cdot \ex(d)
        $$
        Observe that, if $d = \checkmark$, then:
        $$\frac{pn(d)}{1-\partial(f)(\checkmark)p} + \frac{(1-p)\delta_{\checkmark}(d)}{1-\partial(f)(\checkmark)p} = \frac{1-p}{1-\partial(f)(\checkmark)p} = \partial \left(f^{[p]}\right)(d)$$
        When, $d = \left(a, f' \seq f^{[p]}\right)$
        $$\frac{pn(d)}{1-\partial(f)(\checkmark)p} + \frac{(1-p)\delta_{\checkmark}(d)}{1-\partial(f)(\checkmark)p} = \frac{p\partial(f)(a,f')}{1-\partial(f)(\checkmark)p}  = \partial \left(f^{[p]}\right)(d)$$ 
        In the reamining cases, it holds that:
        $$\frac{pn(d)}{1-\partial(f)(\checkmark)p} + \frac{(1-p)\delta_{\checkmark}(d)}{1-\partial(f)(\checkmark)p} = 0 = \partial \left(f^{[p]}\right)(d)$$
        Therefore, we have the following:
        $$f^{[p]}\equiv_0 \bigoplus_{d \in \supp \left(\partial \left(f^{[p]}\right)\right)} \partial\left(f^{[p]}\right)(d) \cdot \ex(d)$$
        which leads to the desired result.
    \end{itemize}
    }
\end{proof}
\begin{corollary}[Productive loop]\label{cor:productive_loop}
    Let $e \in \Exp$ and $p \in [0,1]$. We have that $e^{[p]} \equiv_0 f^{[r]}$ for some $f \in \Exp$ and $r \in [0,1]$, such that $E(f)=0$
\end{corollary}
\begin{proof}
    If $\partial(e)(\checkmark)=1$ and $p=1$, then we have that:
    \begin{align*}
        e^{[1]} &\equiv_0 \left(\bigoplus_{d \in \supp(\partial(e))} \partial(e)(d) \cdot \ex(d)\right)^{[1]} \tag{\cref{thm:fundamental_theorem}}\\
        &\equiv_0 \one^{[1]} \\
        &\equiv_0 \zero \tag{\textbf{Div}}\\
        &\equiv_0 \zero\seq \zero \oplus_1 \one \tag{\textbf{0S} and \textbf{C2}}\\
        &\equiv_0 \zero^{[1]} \tag{\textbf{Unique} fixpoint rule and $E(\zero)=0$}
    \end{align*}
    Hence $e^{[1]}=\zero^{[1]}$. In such a case we have that $E(\zero)=0$.
    In the remaining cases, we have the following:
    \begin{align*}
        e^{[p]} &\equiv_0 \left(\bigoplus_{d \in \supp(\partial(e))} \partial(e)(d) \cdot \ex(d)\right)^{[p]} \tag{\cref{thm:fundamental_theorem}}\\
        &\equiv_0 \left(\one \oplus_{\partial(e)(\checkmark)} \left(\bigoplus_{d \in \supp(\partial(e)) \setminus \{\checkmark\}} \frac{\partial(e)(d)}{\ol{\partial(e)(\checkmark)}} \cdot \ex(d)\right)\right)^{[p]} \tag{\Cref{lem:safe_unrolling_lemma}} \\
        &\equiv_0 \left(\left(\bigoplus_{d \in \supp(\partial(e)) \setminus \{\checkmark\}} \frac{\partial(e)(d)}{\ol{\partial(e)(\checkmark)}} \cdot \ex(d)\right)\oplus_{\ol{\partial(e)(\checkmark)}}  \one\right)^{[p]} \tag{\textbf{C3}}\\
        &\equiv_0 \left(\left(\bigoplus_{(a,e') \in \supp(\partial(e))} \frac{\partial(e)(a,e')}{\ol{\partial(e)(\checkmark)}} \cdot a \seq e'\right)\oplus_{\ol{\partial(e)(\checkmark)}}  \one\right)^{[p]} \tag{Def. of $\ex$}\\
        &\equiv_0 \left(\bigoplus_{(a,e') \in \supp(\partial(e))} \frac{\partial(e)(a,e')}{\ol{\partial(e)(\checkmark)}} \cdot a \seq e' \right)^{\left[\frac{p\ol{\partial(e)(\checkmark)}}{1-p\partial(e)(\checkmark)}\right]} \tag{\textbf{Tight}}
    \end{align*}
    Observe that body of the loop above is an $n$-ary probabilistic sum involving the terms being in the form $a \seq e'$ (such that $a \in A, e' \in \Exp$) for which $E(a \seq e')=0$. Looking at definition of $n$-ary sum (\Cref{prop:binary}) and the termination operator $E(-)$ (\Cref{fig:axioms}) allows us to immediately that the loop body is mapped to $0$ by $E(-)$, which concludes the proof.

\end{proof}
\begin{restatable}{corollary}{quotientcoalgebrainverse}\label{cor:quotient_coalgebra_inverse}
    The quotient coalgebra structure map $\ol{\partial} : {\Exp}/{\equiv_0} \to \distf F {\Exp}/{\equiv_0}$ is an isomorphism
\end{restatable}
\begin{proof}
    Given $\nu \in \distf F {{\Exp}/{\equiv_0}}$ define a function $\ol{\partial}^{-1} : \distf F {{\Exp}/{\equiv_0}} \to {{\Exp}/{\equiv_0}} $ given by:
    $$
    \ol{\partial}^{-1}(\nu) = \left[\nu(\checkmark) \cdot \one \oplus \bigoplus_{(a,[e']_{\equiv_0}) \in \supp(\nu)} \nu(a,[e']_{\equiv_0}) \cdot a \seq e'\right]_{\equiv_0}
    $$
    First, observe that for arbitrary $\nu \in \distf F {{\Exp}/{\equiv_0}}$ we have that:
    \begin{align*}
    (\ol{\partial} \o \ol{\partial}^{-1})(\nu)(\checkmark) &=  \ol{\partial}\left(\left[\nu(\checkmark) \cdot \one \oplus \bigoplus_{(a,[e']_{\equiv_0}) \in \supp(\nu)} \nu(a,[e']_{\equiv_0}) \cdot a \seq e'\right]_{\equiv_0}\right)(\checkmark)\\
    &= (\distf F [-]_{\equiv_0} \o \partial) \left(\nu(\checkmark) \cdot \one\right. \\&\left.\quad\quad\quad\quad\oplus \bigoplus_{(a,[e']_{\equiv_0}) \in \supp(\nu)} \nu(a,[e']_{\equiv_0}) \cdot a \seq e'\right)(\checkmark) \tag{$\ol{\partial}$ is a coalgebra homomorphism} \\
    &= \partial \left(\nu(\checkmark) \cdot \one \oplus \bigoplus_{(a,[e']_{\equiv_0}) \in \supp(\nu)} \nu(a,[e']_{\equiv_0}) \cdot a \seq e'\right) (\checkmark) \\
    &= \nu(\checkmark) \tag{Def. of $\partial$}
    \end{align*} 
    Similarly, for any $(b, [f']_{\equiv_0}) \in \supp(\nu)$, we have that:
    \begin{align*}
    &(\ol{\partial} \o \ol{\partial}^{-1})(\nu)(b, [f']_{\equiv_0})\\ &\quad=  \ol{\partial}\left(\left[\nu(\checkmark) \cdot \one \oplus \bigoplus_{(a,[e']_{\equiv_0}) \in \supp(\nu)} \nu(a,[e']_{\equiv_0}) \cdot a \seq e'\right]_{\equiv_0}\right)(b, [f']_{\equiv_0}) \\
    &\quad= (\distf F [-]_{\equiv_0} \o \partial) \left(\nu(\checkmark) \cdot \one\right. \\&\left.\quad\quad\quad\quad\oplus \bigoplus_{(a,[e']_{\equiv_0}) \in \supp(\nu)} \nu(a,[e']_{\equiv_0}) \cdot a \seq e'\right)(b, [f']_{\equiv_0})  \tag{$\ol{\partial}$ is a coalgebra homomorphism} \\
    &\quad= \sum_{g \equiv f'} \partial\left(\nu(\checkmark) \cdot \one \oplus \bigoplus_{(a,[e']_{\equiv_0}) \in \supp(\nu)} \nu(a,[e']_{\equiv_0}) \cdot a \seq e'\right)(b,g) \\
    &\quad= \nu(b, [f']_{\equiv_0}) \tag{Def. of $\partial$}
    \end{align*}
    For the second part of the proof, let $ e \in \Exp$. As a consequence of \cref{thm:fundamental_theorem}, we have that:
    \begin{align*}
        e &\equiv_0 \bigoplus_{d \in \supp(\partial(e))} \partial(e)(d) \cdot \ex(d) \tag{\cref{thm:fundamental_theorem}} \\
        &\equiv_0 \partial(e)(\checkmark) \cdot \one \oplus \bigoplus_{(a,e') \in \supp(\partial(e))} \partial(e)(a,e') \cdot a \seq e'\\
        &\equiv_0 \partial(e)(\checkmark) \cdot \one \oplus \bigoplus_{(a,[e']_{\equiv_0}) \in A \times {\Exp}/{\equiv_0}} \left(\sum_{g \equiv e'} \partial(e)(a,g)\right) \cdot a \seq e' \tag{\cref{prop:properties_of_positive_convex_algebras}}
    \end{align*}
    Now, observe that:
    \begin{align*}
        &(\ol{\partial}^{-1} \o \ol{\partial})[e]_{\equiv_0}\\ &\quad= (\ol{\partial}^{-1} \o \distf F [-]_{\equiv_0} \o {\partial})(e)\\
        &\quad= \left[ \partial(e)(\checkmark) \cdot \one \oplus\right.\\&\quad\quad\quad\left. \bigoplus_{(a,[e']_{\equiv_0}) \in \supp((\distf F [-]_{\equiv_0} \o \partial)(e))}(\distf F [-]_{\equiv_0} \o \partial)(e)(a,[e']_{\equiv_0}))\cdot a\seq e'\right]_{\equiv_0}\\
        &\quad= \left[ \partial(e)(\checkmark) \cdot \one \oplus \bigoplus_{(a,[e']_{\equiv_0}) \in A \times {\Exp}/{\equiv_0}}(\distf F [-]_{\equiv_0} \o \partial)(e)(a,[e']_{\equiv_0}))\cdot a\seq e'\right]_{\equiv_0} \tag{\cref{prop:properties_of_positive_convex_algebras}}\\
        &\quad= \left[ \partial(e)(\checkmark) \cdot \one\oplus \bigoplus_{(a,[e']_{\equiv_0}) \in A \times {\Exp}/{\equiv_0}} \left(\sum_{g \equiv e'} \partial(e)(a,g) \right) \cdot a \seq e'\right]_{\equiv_0}\\
        &\quad = [e]_{\equiv_0}
    \end{align*}
    which completes the proof.
\end{proof}
\section{Algebra structure}
Define ${\alpha}_{\equiv_0} : \distf {\Exp}/{\equiv_0} \to {\Exp}/{\equiv_0} $ to be given by following composition of morphisms:
\[\begin{tikzcd}
	{\distf({\Exp}/{\equiv_0})} & {\distf\distf F{\Exp}/{\equiv_0}} & {\distf F{\Exp}/{\equiv_0}} & {{\Exp}/{\equiv_0}}
	\arrow["{\distf\ol{\partial}}", from=1-1, to=1-2]
	\arrow["\mu_{F {\Exp}/{\equiv}}", from=1-2, to=1-3]
	\arrow["{{\ol{\partial}}^{-1}}", from=1-3, to=1-4]
\end{tikzcd}\]
This algebra structure is particularly well-behaved. Namely,
\begin{restatable}{lemma}{expressionsemalgebra}\label{lem:quotient_is_an_em_algebra}
    $({\Exp}/{\equiv_0}, \alpha_{\equiv_0})$ is an Eilenberg-Moore algebra for the finitely supported subdistribution monad.
\end{restatable}
\begin{proof}
    We first verify that $\alpha_{\equiv_0} \o \eta_{{\Exp}/{\equiv_{0}}} = \id_{{\Exp}/{\equiv_{0}}}$.
    \begin{align*}
        \alpha_{\equiv_0} \o \eta_X &= \ol{\partial}^{-1} \o \mu_{F {\Exp}/{\equiv_{0}}} \o \distf {\ol{\partial}} \o \eta_{{\Exp}/{\equiv_{0}}} \\
        &=\ol{\partial}^{-1} \o \mu_{F {\Exp}/{\equiv_{0}}} \o \eta_{\distf F{\Exp}/{\equiv_{0}}} \o {\ol{\partial}} \tag{$\eta$ is natural} \\
        &=\ol{\partial}^{-1} \o \ol{\partial} \tag{Monad laws} \\
        &=\id_{{\Exp}/{\equiv_{0}}}
    \end{align*}
    Then, we show that $\alpha_{\equiv_0} \o \distf \alpha_{\equiv_0} = \alpha_{\equiv_0} \o \mu_{{\Exp}/{\equiv_0}}$
    \begin{align*}
        \alpha_{\equiv_0} \o \distf \alpha_{\equiv_0} &= \ol{\partial}^{-1} \o \mu_{F {\Exp}/{\equiv_{0}}} \o \distf {\ol{\partial}} \o \distf\ol{\partial}^{-1} \o \distf\mu_{F {\Exp}/{\equiv_{0}}} \o \distf^2 {\ol{\partial}} \\
        &=\ol{\partial}^{-1} \o \mu_{F {\Exp}/{\equiv_{0}}}  \o \distf\mu_{F {\Exp}/{\equiv_{0}}} \o \distf^2 {\ol{\partial}} \tag{Monad laws}\\
        &=\ol{\partial}^{-1} \o \mu_{F {\Exp}/{\equiv_{0}}} \o \distf {\ol{\partial}} \o \mu_{{\Exp}/{\equiv_0}} \tag{$\mu$ is natural}\\ 
        &= \alpha_{\equiv_{0}} \o \mu_{{\Exp}/{\equiv_0}}
    \end{align*}
\end{proof}
\begin{lemma}\label{lem:generalised_sum_congruence}
    Let $I$ be a finite set, $\{p_i\}_{i \in I}$, $\{e_i\}_{i \in I}$ and $\{f_i\}_{i \in I}$, such that for all $i \in I$, $p_i \in [0,1]$, $e_i, f_i \in \Exp$, $e_i \equiv f_i$ and $\sum_{i \in I} p_i \leq 1$. We have that 
    $$
    \bigoplus_{i \in I} p_i \cdot e_i \equiv \bigoplus_{i \in I} p_i \cdot f_i
    $$
\end{lemma}
\begin{proof}
    By induction.
    \begin{itemize}
        \item If $I = \emptyset$, then 
        $
        \bigoplus_{i \in I} p_i \cdot e_i \equiv \zero \equiv  \bigoplus_{i \in I} p_i \cdot f_i
        $.
        \item If there is some $j \in I$, such that $p_j = 1$, we have that 
        $
        \bigoplus_{i \in I} p_i \cdot e_i \equiv e_j \equiv f_j \equiv   \bigoplus_{i \in I} p_i \cdot f_i
        $.
        \item Otherwise, for some $j \in J$, we have that:
        \begin{align*}
            \bigoplus_{i \in I} p_i \cdot e_i &\equiv e_j \oplus_{p_j} \left(\bigoplus_{i \in I \setminus \{j\}} \frac{p_i}{\ol{p_j}} \cdot e_i \right) \tag{Def. of $n$-ary convex sum}\\
            &\equiv f_j \oplus_{p_j} \left(\bigoplus_{i \in I \setminus \{j\}} \frac{p_i}{\ol{p_j}} \cdot e_i \right) \tag{$e_j \equiv f_j$} \\
            &\equiv f_j \oplus_{p_j} \left(\bigoplus_{i \in I \setminus \{j\}} \frac{p_i}{\ol{p_j}} \cdot f_i \right) \tag{Induction hypothesis} \\
            &\equiv \bigoplus_{i \in I} p_i \cdot f_i \tag{Def. of $n$-ary convex sum}\\
        \end{align*}
    \end{itemize}

\end{proof}
\section{Soundness of axioms for probabilistic language equivalence}
Let ${\uaequiv} \subseteq {\Exp \times \Exp}$ be the least congruence relation such that for all $e,f \in \Exp$, $p \in [0,1]$ and $a \in A$ closed under the following rules:
\begin{enumerate}
    \item If $e \equiv_0 f$, then $e \uaequiv f$
    \item $a \seq (e \oplus_p f) \uaequiv a\seq e \oplus_p a \seq f$
    \item $a \seq \zero \uaequiv \zero$
\end{enumerate}
\begin{lemma}\label{lem:simpler_equivalence}
    For all $e,f \in \Exp$, $e \equiv f$ if and only if $e \uaequiv f$
\end{lemma}
\begin{proof}
    We split the proof into two cases.
    \begin{itemize}
        \item[] \fbox{$e \uaequiv f \implies e \equiv f$}

        This implication holds almost immediately. If $e \equiv_0 f $, then $e \equiv f$. Rules that $a \seq (e \oplus_p f) \uaequiv a\seq e \oplus_p a \seq f$ and $a \seq \zero \uaequiv \zero$ are special cases of \textbf{D2} and \textbf{S0} axioms of $\equiv$ specialised to single actions.

        \item[] \fbox{$e \equiv f \implies e \uaequiv f$}
        
        Axioms of $\equiv$ are either axioms of $\equiv_0$, which are already contained in $\uaequiv$ or are \textbf{D2}/\textbf{S0} axioms. It suffices to show that latter two are derivable in $\uaequiv.$ First, we show by induction that for all $e \in \Exp$, $e \seq \zero \uaequiv \zero$

        \begin{itemize}
            \item[]\fbox{$e = \zero$} 
            \begin{align*}
                e \seq \zero &\uaequiv \zero \seq \zero \tag{$e = \zero$}\\
                &\uaequiv \zero \tag{\textbf{0S}}
            \end{align*}

            \item[]\fbox{$e=\one$}
            \begin{align*}
                e \seq \zero &\uaequiv \one \seq \zero \tag{$e=\one$}\\
                &\uaequiv \zero \tag{\textbf{1S}}
            \end{align*}   

            \item[]\fbox{$e = a$}
            \begin{align*}
                e \seq \zero &\uaequiv a \seq \zero \tag{$e = a$}\\
                &\uaequiv \zero \tag{Def. of $\uaequiv$}
            \end{align*}

            \item[]\fbox{$e = f \oplus_p g$}
            \begin{align*}
                e \seq \zero &\uaequiv (f \oplus_p g) \seq \zero \tag{$e = f \oplus_p g$}\\
                &\uaequiv f \seq \zero \oplus_p g \seq \zero \tag{\textbf{D1}} \\
                &\uaequiv \zero \oplus_p \zero \tag{Induction hypothesis} \\
                &\uaequiv \zero \tag{\textbf{C1}}
            \end{align*}

            \item[]\fbox{$e = f\seq g$}
            \begin{align*}
                e \seq \zero &\uaequiv (f \seq g) \seq \zero \tag{$e = f \seq g$}\\
                &\uaequiv f \seq (g \seq \zero) \tag{\textbf{S}} \\
                &\uaequiv f \seq \zero\tag{Induction hypothesis} \\
                &\uaequiv \zero \tag{Induction hypothesis}
            \end{align*}

            \item[]\fbox{$e = f^{[p]}$}

            First, by \cref{cor:productive_loop} we know that $f^{[p]} \uaequiv g^{[r]}$, such that $E(g)=0$.

            \begin{align*}
                \zero &\uaequiv \zero \oplus_r \zero \tag{\textbf{C1}} \\
                &\uaequiv g \seq \zero \oplus_r \zero \tag{Induction hypothesis}
            \end{align*}
            Since $E(g)= 0$, we can use unique fixpoint axiom and obtain:
            \begin{align*}
                \zero &\uaequiv g^{[r]}\seq \zero \tag{\textbf{Unique}}\\
                &\uaequiv f^{[p]} \seq \zero \tag{\cref{cor:productive_loop}}
            \end{align*}
        \end{itemize}
        Secondly, we show by induction that for all $e,f,g \in \Exp$ and $p \in [0,1]$ we have that $e \seq (f \oplus_p g) \uaequiv e \seq f \oplus_p e \seq g$
        \begin{itemize}
            \item[] \fbox{$e = \zero$}
            \begin{align*}
                e \seq (f \oplus_p g ) &\uaequiv \zero \seq (f \oplus_p g ) \tag{$e = \zero$} \\
                &\uaequiv \zero \tag{\textbf{0S}} \\
                &\uaequiv \zero \oplus_p \zero \tag{\textbf{C1}} \\
                &\uaequiv \zero\seq f \oplus_p \zero \seq g \tag{\textbf{0S}}\\
            \end{align*}

            \item[] \fbox{$e=\one$}
            \begin{align*}
                e \seq (f \oplus_p g) &\uaequiv \one \seq (f \oplus_p g) \tag{$e=\one$}\\
                &\uaequiv f \oplus_p g \tag{\textbf{1S}} \\
                &\uaequiv \one \seq f \oplus_p \one \seq g \tag{\textbf{1S}}\\
            \end{align*}

            \item[] \fbox{$e=a$}
            \begin{align*}
                e \seq (f \oplus_p g) &\uaequiv a \seq (f \oplus_p g) \tag{$e = a$}\\
                &\uaequiv a \seq f \oplus_p a \seq g \tag{Def. of $\uaequiv$} \\
            \end{align*}

            \item[] \fbox{$e=h \oplus_r i$}
            \begin{align*}
                e \seq (f \oplus_p g) &\uaequiv (g \oplus_r h) \seq (f \oplus_p g) \tag{$e = h \oplus_r i$}\\
                &\uaequiv h\seq(f \oplus_p g) \oplus_r i \seq (f \oplus_p g) \tag{\textbf{D1}} \\
                &\uaequiv (h\seq f \oplus_p h\seq g) \oplus_r (i \seq f \oplus_p i \seq g) \tag{Induction hypothesis} \\
                &\uaequiv (h \seq f \oplus_r i \seq f) \oplus_p (h \seq g \oplus_r i \seq g) \tag{\cref{lem:binary_sum_properties}}\\
                &\uaequiv (h \oplus_r i)\seq f \oplus_p (h \oplus_r i) \seq g \tag{\textbf{D1}} \\
            \end{align*}

            \item[] \fbox{$e = h^{[r]}$}

             First, by \cref{cor:productive_loop} we know that $h^{[r]} \uaequiv i^{[q]}$, such that $E(i)=0$.

             Now, we derive the following:
             \begin{align*}
                 i^{[q]} \seq f \oplus_p i^{[q]} \seq g  &\uaequiv (i \seq i^{[q]} \oplus_q \one) \seq f \oplus_p  (i \seq i^{[q]} \oplus_q \one) \seq g \tag{\textbf{Unroll}}\\
                 &\uaequiv (i \seq i^{[q]} \seq f \oplus_q \one \seq f)  \oplus_p  (i \seq i^{[q]} \seq g \oplus_q \one \seq g) \tag{\textbf{D1}}\\
                 &\uaequiv (i \seq i^{[q]} \seq f \oplus_q f)  \oplus_p  (i \seq i^{[q]} \seq g \oplus_q g) \tag{\textbf{0S}} \\
                 &\uaequiv (i \seq i^{[q]} \seq f \oplus_p i \seq i^{[q]} \seq g) \oplus_q (f \oplus_p g) \tag{\cref{lem:binary_sum_properties}}\\
                 &\uaequiv i \seq (i^{[q]} \seq f \oplus_p i^{[q]} \seq g) \oplus_q (f \oplus_p g) \tag{Induction hypothesis}
             \end{align*}
            Since $E(i)=0$, we can use \textbf{Unique} fixpoint axiom to derive:
            \begin{align*}
                 i^{[q]} \seq f \oplus_p i^{[q]} \seq g \uaequiv i^{[q]} \seq (f \oplus_p g)
            \end{align*}
            Since $h^{[r]} \uaequiv i^{[q]}$ and $\uaequiv$ is a congruence, we have that:
            \begin{align*}
                 h^{[r]} \seq f \oplus_p h^{[r]} \seq g \uaequiv h^{[r]} \seq (f \oplus_p g)
            \end{align*}
            which completes a proof.
        \end{itemize}
    \end{itemize}
\end{proof}
\coalgebramapalgebrahomomorphism*
\begin{proof}
    We show that the following diagram commutes.
        \[\begin{tikzcd}
		{\distf{\Exp}/{\equiv_0}} && {{\Exp}/{\equiv_0}} \\
		& {\distf\distf F{\Exp}/{\equiv_0}} & {\distf F{\Exp}/{\equiv_0}} \\
		&& {G\distf {\Exp}/{\equiv_0}} \\
		{G\distf{\Exp}/{\equiv_0}} && {G{\Exp}/{\equiv_0}}
		\arrow["\ol{\partial}"', from=1-3, to=2-3]
		\arrow["\gamma_{{\Exp}/{\equiv}_0}"', from=2-3, to=3-3]
		\arrow["{Ga_{\equiv_0}}"', from=3-3, to=4-3]
		\arrow["(\gamma_{{\Exp}/{\equiv_0}} \o \ol{\partial} )^\star"', from=1-1, to=4-1]
		\arrow["{\alpha_{\equiv_0}}", from=1-1, to=1-3]
		\arrow["{G\alpha_{\equiv_0}}", from=4-1, to=4-3]
		\arrow["{\distf\ol{\partial}}", from=1-1, to=2-2]
		\arrow["\mu_{F {\Exp}/{\equiv_0}}", from=2-2, to=2-3]
	\end{tikzcd}\]
 For the top inner part, we have the following:
 \begin{align*}
     \ol{\partial} \o \alpha_{\equiv_0} &= \ol{\partial} \o \ol{\partial}^{-1} \o \mu_{F {\Exp}/{\equiv_{0}}} \o \distf {\ol{\partial}} \\
     &= \mu_{F {\Exp}/{\equiv_{0}}} \o \distf {\ol{\partial}}
 \end{align*}
 For bottom part, since $\alpha_{{\Exp}/{\equiv_0}}$ is the structure map of an Eilenberg-Moore algebra, we have that
$
    \gamma_{{\Exp}/{\equiv_0}} \o \ol{\partial} \o \alpha_{\equiv_0} \o \eta_{{\Exp}/{\equiv_0}}= \gamma_{{\Exp}/{\equiv_0}} \o \ol{\partial} 
$.
 Since $(\gamma_{{\Exp}/{\equiv_0}} \o \ol{\partial} )^\star$ is the unique map such that $(\gamma_{{\Exp}/{\equiv_0}} \o \ol{\partial} )^\star \o \eta_{{\Exp}/{\equiv_0}}=\gamma_{{\Exp}/{\equiv_0}} \o \ol{\partial}$ we have that 
$(\gamma_{{\Exp}/{\equiv_0}} \o \ol{\partial} )^\star = \gamma_{{\Exp}/{\equiv_0}} \o \ol{\partial} \o \alpha_{\equiv_0} $
and thus:
$$G\alpha_{\equiv_0} \o (\gamma_{{\Exp}/{\equiv_0}} \o \ol{\partial} )^\star = G\alpha_{\equiv_0} \o \gamma_{{\Exp}/{\equiv_0}} \o \ol{\partial} \o \alpha_{\equiv_0} $$
which proves that the bottom part of the diagram above commutes.
\end{proof}
\quotienttrace*
\begin{proof}
    By induction on the length of derivation of $\equiv$. We use the diagonal fill in lemma to show the existence of $d$, so it suffices to show that: $$\ker([-]_\equiv) \subseteq \ker(G[-]_\equiv \o G\alpha_{\equiv_0} \o \gamma_{{\Exp}/{\equiv_0}} \o \ol{\partial})$$
    
    If $e \equiv_0 f$, then $\ol{\partial}([e]_{\equiv_0}) = \ol{\partial}([f]_{\equiv_0})$, which implies the property. Since $\uaequiv = \equiv$ (\cref{lem:simpler_equivalence}), it is enough to show the desired property for the simpler version of the axiomatisation. First, observe the following
    \begin{align*}
        G[-]_\equiv \o G\alpha_{\equiv_0} \o \gamma_{{\Exp}/{\equiv_0}} \o \ol{\partial} \o [-]_{\equiv_0} &= G\alpha_{\equiv} \o G \distf[-]_\equiv \o \gamma_{{\Exp}/{\equiv_0}} \o \ol{\partial} \o [-]_{\equiv_0} \tag{\cref{lem:coarser_quotient_is_a_pca}}\\
        &= G\alpha_{\equiv} \o G \distf[-]_\equiv \o \gamma_{{\Exp}/{\equiv_0}} \o \distf F [-]_{\equiv_0} \o \partial \tag{$[-]_{\equiv_0}$ is a coalgebra homomorphism}\\
        &= G\alpha_{\equiv} \o G \distf[-]_\equiv \o G \distf [-]_{\equiv_0} \o \gamma_{{\Exp}} \o \partial \tag{$\gamma$ is a natural transformation} \\
        &= G\alpha_{\equiv} \o G \distf[-] \o \gamma_{{\Exp}} \o \partial \tag{Definition of $[-]$}\\
        &= G\alpha_{\equiv} \o \gamma_{{\Exp}/{\equiv_0}} \o \distf F[-] \o \partial \tag{$\gamma$ is a natural transformation}\\
    \end{align*}
    \fbox{$a \seq (e \oplus_p f) \uaequiv a \seq e \oplus_p a \seq f$} Let $a \in A$, $p \in [0,1]$ and $e,f \in \Exp$.
    We have that: $$(\distf F[-] \o \partial) (a \seq (e \oplus_p f)) = \delta_{(a, [e \oplus_p f])} $$
    After applying $\gamma_{{\Exp}/{\equiv_0}}$ we obtain
    $$(\gamma_{{\Exp}/{\equiv_0}} \o \distf F[-] \o \partial) (a \seq (e \oplus_p f)) = \langle 0, l_b \rangle$$
    where $l_b = \delta_{[e \oplus_p f]}$ if $a= b$ or $l_b(x) = 0$ otherwise. For the right hand side of the (simplified version of) left distributivity (\textbf{D2}) axiom, we have that
    $$(\distf F[-] \o \partial) (a \seq e \oplus_p a \seq f) = p\delta_{(a, [e])} + (1-p)\delta_{(a, [f])} $$
    After applying $\gamma_{{\Exp}/{\equiv_0}}$ we obtain
    $$(\gamma_{{\Exp}/{\equiv_0}} \o \distf F[-] \o \partial) (a \seq e \oplus_p a \seq f) = \langle 0, r_b \rangle$$
    where $r_b = p\delta_{[ e]} + (1-p)\delta_{[ f]}$ if $a= b$ or $r_b(x) = 0$ otherwise. Hence, we are left with showing that 
    $\alpha_{\equiv} \left(\delta_{[e \oplus_p f]}\right)=\alpha_{\equiv} \left(p \delta_{[e]} + (1-p)\delta_{[f]}\right)$.
    We will use the isomorphism between $\Set^\distf$ and $\pca$. First, we consider the case when $e \not\equiv f$.
    \begin{align*}
        \bigboxplus_{d \in \supp\left(\delta_{[e \oplus_p f]}\right)} \delta_{[e \oplus_p f]}(d) \cdot d &= [e \oplus_p f] \\
        &= \bigboxplus_{d \in \supp\left(p\delta_{[e]} + (1-p)\delta_{[f]}\right)} \left(p\delta_{[e]} + (1-p)\delta_{[f]}\right)(d) \cdot d\\
    \end{align*}
    In the case when $e \equiv f$, we have the following:
        \begin{align*}
        \bigboxplus_{d \in \supp\left(\delta_{[e \oplus_p f]}\right)} \delta_{[e \oplus_p f]}(d) \cdot d &= [e \oplus_p f] \\
        &= [e] \tag{\textbf{C1}} \\
        &= \bigboxplus_{d \in \supp\left(p\delta_{[e]} + (1-p)\delta_{[f]}\right)} \left(p\delta_{[e]} + (1-p)\delta_{[f]}\right)(d) \cdot d\\
    \end{align*}
    \fbox{$a \seq 0 \uaequiv 0$}
    Now, consider the second axiom we need to check. We have that: $$(\distf F[-] \o \partial) (a \seq \zero) = \delta_{(a, [\zero])} $$
    After applying $\gamma_{{\Exp}/{\equiv_0}}$ we obtain
    $(\gamma_{{\Exp}/{\equiv_0}} \o \distf F[-] \o \partial) (a \seq \zero) = \langle 0, l_b \rangle$
    where $l_b = \delta_{[\zero]}$ if $a= b$ or $l_b(x) = 0$ otherwise. For the right hand side of the right annihilation axiom (\textbf{S0}), we have that:
    $$(\gamma_{{\Exp}/{\equiv_0}} \o \distf F[-] \o \partial) (\zero) = \langle 0, r_b \rangle$$
    where $r_b(x) = 0$. Hence, we are left with showing that 
    $\alpha_{\equiv}(\delta_{[\zero]})=\alpha_{\equiv}(r_b)$.
    We have the following:
    \begin{align*}
        \bigboxplus_{d \in \supp \left(\delta_{[\zero]}\right)} \left(\delta_{[\zero]}\right)(d) \cdot d = [\zero] = \bigboxplus_{d \in r_b} r_b(d) \cdot d \tag{$\supp(r_b) = \emptyset$}
    \end{align*}
    which completes the proof.
\end{proof} 
    
    \begin{restatable}{lemma}{factoringlemma}\label{lem:factoring_lemma}
     The unique final homomorphism $\dagger (\gamma_{\Exp} \o \partial)^\star : (\distf \Exp, \gamma_{\Exp} \o \partial) \to \nu G$ from the determinisation of $(\Exp, \partial)$ into the final $G$-coalgebra satisfies the following:
     $$\dagger (\gamma_{\Exp} \o \partial)^\star = \dagger d \o [-]_{\equiv} \o \alpha_{\equiv_0} \o \distf [-]_{\equiv_0}$$
    \end{restatable}
    \begin{proof}
    Since $\equiv_0$ is a bisimulation equivalence on $\distf F$-coalgebra $(\Exp, \partial)$, the map $[-]_{\equiv_0} : \Exp \to {\Exp}/{\equiv_0}$ is a $\distf F$-coalgebra homomorphism to $({\Exp}/{\equiv_0}, \ol{\partial})$. Because of~\cite[Lemma~15.1]{Rutten:2000:Universal}, $[-]_{\equiv_0}$ is also $G \distf$-coalgebra homomorphism from $(\Exp, \gamma_{\Exp} \o \partial)$ to $({\Exp}/{\equiv_0}, \gamma_{{\Exp}/{\equiv_0}} \o \ol{\partial})$. 

    Finally, by~\cite[Theorem~4.1]{Silva:2010:Generalizing} we have that $\distf [-]_{\equiv_0}$ is $G$-coalgebra homomorphism from $(\distf \Exp, (\gamma_{\Exp} \o \partial)^\star)$ to $(\distf {\Exp}/{\equiv_0}, (\gamma_{{\Exp}/{\equiv_0}} \o \ol{\partial})^\star)$. 
    
    Since $\dagger(\gamma_{{\Exp}/{\equiv}} \o \ol{\partial})^\star \o \distf [-]_{\equiv_0}$ is a $G$-coalgebra homomorphism from  $(\distf \Exp, (\gamma_{\Exp} \o \partial)^\star)$ to the final $G$-coalgebra, which by finality must be equal to $\dagger (\gamma_{{\Exp}} \o \partial)^\star$

    By similar line of reasoning involving a finality argument, as a consequence of \cref{lem:algebra_map_coalgebra_homomorphism}, we have that: $$
    \dagger (\gamma_{{\Exp}/{\equiv_0}} \o \ol{\partial})^\star = \dagger (G \alpha_{\equiv_0} \o \gamma_{{\Exp}/{\equiv_0}} \o \ol{\partial}) \o \alpha_{\equiv_0}
    $$
    Finally, via identical line of reasoning we can use \cref{lem:coalg_on_quotient_left_dist} to show:
    $$
    \dagger (G \alpha_{\equiv_0} \o \gamma_{{\Exp}/{\equiv_0}} \o \ol{\partial}) = \dagger d \o [-]_{\equiv}
    $$
    Composing the above equalities yields
    $\dagger (\gamma_{\Exp} \o \partial)^{\star} = \dagger d \o [-]_{\equiv} \o \alpha_{\equiv_0} \o \distf [-]_{\equiv_0}$.
    \end{proof}

\section{Systems of equations}
\uniquesolutions*
\begin{proof}
    Let $\mathcal{M}=(M, p, b, r)$ be an arbitrary left-affine system of equations.
    We proceed by induction on the size of $Q$. Since $Q$ is finite, we can safely assume that $Q=\{q_1, \dots q_n\}$. We will write $M_{i, j}$ instead of $M_{q_i, q_j}$ and similarly $p_i$ for $p_{q_i}$, $b_i$ for $b_{q_i}$ and $r_i$ for $r_{q_i}$.
    
    If $Q=\{q_1\}$, then we set $h(q_1) = {M_{1,1}}^{[p_{1,1}]}\seq b_1$. To see that it is indeed the $\equiv$-solution, observe the following:
    \begin{align*}
        h(q_1) &= {M_{1,1}}^{[p_{1,1}]}\seq b_{1} \\
        &\equiv \left(M_{1,1} \seq {M_{1,1}}^{[p_{1,1}]} \oplus_{p_{1,1}} \one \right) \seq b_{1} \tag{\textbf{Unroll}}\\
        &\equiv M_{1,1} \seq {M_{1,1}}^{[p_{1,1}]} \seq b_{1} \oplus_{p_{1,1}} b_{1} \tag{\textbf{D1}}\\
        &\equiv M_{1,1} \seq h(q_{1}) \oplus_{p_{1,1}} b_{1}\\
        &\equiv p_{1,1} \cdot M_{1,1} \seq h(q_{1}) \oplus \ol{p_{1,1}} \cdot b_1 \\
        &\equiv p_{1,1} \cdot M_{1,1} \seq h(q_{1}) \oplus {r_1} \cdot b_1 \tag{$r_1 = 1 - p_{1,1}$}\\
    \end{align*}
    Given an another $\equiv$-solution $g : Q \to \Exp$, we have that:
    \begin{align*}
        g(q_1) &\equiv  p_{1,1} \cdot{M_{1,1}} \seq g(q_1) \oplus r_1 \cdot b_1 \\
        &\equiv  p_{1,1} \cdot{M_{1,1}} \seq g(q_1) \oplus \ol{p_{1,1}} \cdot b_1  \tag{$r_1 = 1 - p_{1,1}$} \\
        &\equiv  {M_{1,1}} \seq g(q_1) \oplus_{p_{1,1}} b_1\\  
        &\equiv {M_{1,1}}^{[p_{1,1}]}\seq b_1 \tag{\textbf{Unique} and $E(M_{1,1})=0$}\\
        &\equiv h(q_1)
    \end{align*}
    For the induction step, assume that all systems of the size \(n\) admit a unique $\equiv$-solution. First, we show that problem of finding $\equiv$-solution to the system of \(n+1\) unknowns can be reduced to the problem of finding $\equiv$-solution to the system with $n$ unknowns. Let $Q=\{q_1, \dots, q_{n+1}\}$. One of the equivalences that need to hold for $h : Q \to \Exp$ to be a $\equiv$-solution to the system $\mathcal{M}=(p, M, r, b)$ of size $n+1$ is the following:
    $$
    h(q_{n+1}) \equiv \left(\bigoplus_{i=1}^{n+1} p_{{n+1}, i} \cdot M_{{n+1}, i}\seq h(q_{i}) \right) \oplus r_{{n+1}} \cdot b_{{n+1}} 
    $$
    We can unroll the $n+1$-ary sum (\Cref{lem:sum_unrolling}) to obtain:
    \begin{align*}
         h(q_{n+1}) &\equiv M_{n+1, n+1} \seq h(q_{n+1}) \oplus_{p_{n+1, n+1}} \\ &\quad\quad \left(\left(\bigoplus_{i=1}^{n} \frac{p_{{n+1}, i}}{\ol{p_{n+1, n+1}}} \cdot M_{{n+1}, i}\seq h(q_{i}) \right) \oplus \frac{r_{{n+1}}}{\ol{p_{n+1, n+1}}} \cdot b_{{n+1}} \right)
    \end{align*}
    Since $E(M_{n+1, n+1}) = 0$ we can apply the \textbf{Unique} fixpoint axiom and obtain:
    \begin{align*}
        h(q_{n+1}) &\equiv M_{n+1, n+1}^{[p_{n+1, n+1}]}\seq\left(\left(\bigoplus_{i=1}^{n} \frac{p_{{n+1}, i}}{\ol{p_{n+1, n+1}}} \cdot M_{{n+1}, i}\seq h(q_{i}) \right) \oplus \frac{r_{{n+1}}}{\ol{p_{n+1, n+1}}} \cdot b_{{n+1}} \right)
    \end{align*}
    Observe that the above only depends on $h(q_1), \dots h(q_n)$ and we can substitute it in the equations for $h(q_1), \dots, h(q_n)$. For an arbitrary $q_j$, such that $ j \leq n$, we have that:
    \begin{align*}
         h(q_{j}) &\equiv p_{j, n+1}\cdot M_{j, n+1} \seq h(q_{n+1}) \oplus \left(\bigoplus_{i=1}^{n} {p_{{j}, i}} \cdot M_{{j}, i}\seq h(q_{i}) \right) \oplus {r_{{j}}} \cdot b_{{j}} 
    \end{align*}
    Now, substitute $h(q_{n+1})$ to obtain the following:
        \begin{align*}
         h(q_{j}) &\equiv \left(\bigoplus_{i=1}^{n} {p_{{j}, i}} \cdot M_{{j}, i}\seq h(q_{i}) \right) \oplus {r_{{j}}} \cdot b_{{j}}\\& \quad \oplus p_{j, n+1}\cdot M_{j, n+1} \seq M_{n+1, n+1}^{[p_{n+1, n+1}]}\seq\left(\left(\bigoplus_{i=1}^{n} \frac{p_{{n+1}, i}}{\ol{p_{n+1, n+1}}} \cdot M_{{n+1}, i}\seq h(q_{i}) \right) \oplus \frac{r_{{n+1}}}{\ol{p_{n+1, n+1}}} \cdot b_{{n+1}} \right)\\
    \end{align*}
    We can apply \Cref{lem:generalised_left_distributivity} to rearrange the above into:
    \begin{align*}
         h(q_{j}) &\equiv \left(\bigoplus_{i=1}^{n} {p_{{j}, i}} \cdot M_{{j}, i}\seq h(q_{i}) \right)\\& \quad \oplus \left(\bigoplus_{i=1}^{n} \frac{p_{j, n+1}p_{{n+1}, i}}{\ol{p_{n+1, n+1}}} \cdot M_{j, n+1} \seq M_{n+1, n+1}^{[p_{n+1, n+1}]}\seq M_{{n+1}, i}\seq h(q_{i}) \right) \\
         &\quad \oplus {r_{{j}}} \cdot b_{{j}} \oplus \frac{p_{j, n+1}r_{{n+1}}}{\ol{p_{n+1, n+1}}} \cdot M_{j, n+1} \seq M_{n+1, n+1}^{[p_{n+1, n+1}]}\seq b_{{n+1}}\\
    \end{align*}
    To simplify the expression above, we introduce the following shorthands:
    \begin{gather*}
        s_{j,i} = p_{j,i} + \frac{p_{j, n+1} p_{n+1, i}}{\ol{p_{n+1, n+1}}}\\
        N_{j,i} = \frac{p_{{j}, i}}{s_{j,i}} \cdot M_{{j}, i} \oplus \frac{p_{j, n+1}p_{{n+1}, i}}{s_{j,i}\ol{p_{n+1, n+1}}} \cdot M_{j, n+1} \seq M_{n+1, n+1}^{[p_{n+1, n+1}]}\seq M_{{n+1}, i}\\
        t_j = r_j + \frac{p_{j, n+1}r_{{n+1}}}{\ol{p_{n+1, n+1}}}\\
        c_j = \frac{r_j}{t_j}\cdot b_j \oplus \frac{p_{j, n+1}r_{{n+1}}}{t_j\ol{p_{n+1, n+1}}} \cdot M_{j, n+1} \seq M_{n+1, n+1}^{[p_{n+1, n+1}]}\seq b_{{n+1}}
    \end{gather*}
    Observe that $E(N_{j,i})=0$.
    \noindent
    Because of \cref{lem:grouping_probabilities} and \Cref{lem:generalised_right_distributivity}, we have that:
    \begin{align*}
        h(q_j) = \left(\bigoplus_{i=1}^n s_{j, i} \cdot N_{j,i} \seq h(q_i)\right) \oplus t_j \cdot c_j
    \end{align*}
    In other words, $h$ restricted to $\{q_1, \dots, q_n\}$ must be a $\equiv$-solution to the left affine system $\mathcal{T}=(N, s, c, t)$, which by induction hypothesis admits a unique solution. We can extend the solution to $\mathcal{T}$ to the solution to the whole system $\mathcal{S}$, by setting $h({q_n})$ to be: 
    \begin{align*}
        h(q_{n+1}) &\equiv M_{n+1, n+1}^{[p_{n+1, n+1}]}\seq\left(\left(\bigoplus_{i=1}^{n} \frac{p_{{n+1}, i}}{\ol{p_{n+1, n+1}}} \cdot M_{{n+1}, i}\seq h(q_{i}) \right) \oplus \frac{r_{{n+1}}}{\ol{p_{n+1, n+1}}} \cdot b_{{n+1}} \right)
    \end{align*}
    By using \textbf{Unroll} axiom and performing analogous axiomatic manipulation as before, one can show that it is an indeed a $\equiv$-solution to $\mathcal{S}$.

    For the uniqueness, assume that we have $g : Q \to \Exp$, an another $\equiv$-solution to system $\mathcal{M}$. Because of \textbf{Unique} fixpoint axiom, we have that:
    \begin{align*}
        g(q_{n+1}) &\equiv M_{n+1, n+1}^{[p_{n+1, n+1}]}\seq\left(\left(\bigoplus_{i=1}^{n} \frac{p_{{n+1}, i}}{\ol{p_{n+1, n+1}}} \cdot M_{{n+1}, i}\seq g(q_{i}) \right) \oplus \frac{r_{{n+1}}}{\ol{p_{n+1, n+1}}} \cdot b_{{n+1}} \right)
    \end{align*}
     Substituting it to equations for $h(q_1), \dots, h(q_n)$ and performing the steps as earlier would lead to requiring that for all $1 \leq j \leq n$, we have that:
    \begin{align*}
        g(q_j) = \left(\bigoplus_{i=1}^n s_{j, i} \cdot N_{j,i} \seq g(q_i)\right) \oplus t_j \cdot c_j
    \end{align*}
    By induction hypothesis left-affine system of equations $\mathcal{T}$, admits a unique $\equiv$-solution. Therefore for all $1 \leq j \leq n$, we have that
    $g(q_j)\equiv h(q_j)$.
    Because of this, we have that:
    \begin{align*}
        g(q_{n+1}) &\equiv M_{n+1, n+1}^{[p_{n+1, n+1}]}\seq\left(\left(\bigoplus_{i=1}^{n} \frac{p_{{n+1}, i}}{\ol{p_{n+1, n+1}}} \cdot M_{{n+1}, i}\seq g(q_{i}) \right) \oplus \frac{r_{{n+1}}}{\ol{p_{n+1, n+1}}} \cdot b_{{n+1}} \right)\\
        &\equiv M_{n+1, n+1}^{[p_{n+1, n+1}]}\seq\left(\left(\bigoplus_{i=1}^{n} \frac{p_{{n+1}, i}}{\ol{p_{n+1, n+1}}} \cdot M_{{n+1}, i}\seq h(q_{i}) \right) \oplus \frac{r_{{n+1}}}{\ol{p_{n+1, n+1}}} \cdot b_{{n+1}} \right)\\
        &\equiv h(q_{n+1})
    \end{align*}
    which completes the proof.
\end{proof}
\begin{restatable}{lemma}{simplersolutionslemma}\label{lem:simpler_solution_lemma}
    Let $(X, \beta)$ be a finite-state \acro{GPTS}. A map $h : X \to \Exp$ is a $\equiv$-solution to the system $\mathcal{S}(\beta)$ if and only if for all $x \in X$
    $$
    h(x) \equiv \left(\bigoplus_{(a, x') \in A \times X} \beta(x)(a,x') \cdot a \seq h(x')\right) \oplus \beta(x)(\checkmark) \cdot \one
    $$
\end{restatable}
\begin{proof}
    Fix an arbitrary $x \in X$.
    Recall that, if $p^\beta_{x,x'}=0$, then
    $M^\beta_{x,x'}=\zero\equiv \bigoplus_{a \in A} 0 \cdot a$.
    Because of that, we can safely assume that $M^\beta_{x,x'}$ can be always written out in the following form:
    $$
        M^\beta_{x,x'} \equiv \bigoplus_{a \in A} s^a_{x,x'} \cdot a
    $$
    where $s^a_{x,x'} \in [0,1]$. By definition of $\equiv$-solution we the following:
    \begin{align*}
        h(x) &\equiv \left(\bigoplus_{x' \in X} p^{\beta}_{x, x'} \cdot M^\beta_{x,x'}\seq h(x')\right) \oplus r^{\beta}_x \cdot b^{\beta}_x \\
        &\equiv  \left(\bigoplus_{x' \in X} p^\beta_{x,x'} \cdot \left(\bigoplus_{a \in A} s^a_{x,x'} \cdot a\right)\seq h(x') \right)\oplus r^\beta_x \cdot b^\beta_x \\
        &\equiv  \left(\bigoplus_{x' \in X} p^\beta_{x,x'} \cdot \left(\bigoplus_{a \in A} s^a_{x,x'} \cdot a\seq h(x')\right) \right)\oplus r^\beta_x \cdot b^\beta_x \tag{\Cref{lem:generalised_right_distributivity}}\\
        &\equiv \left(\bigoplus_{(a , x') \in A \times X} p^\beta_{x,x'}s^a_{x,x'} \cdot a \seq h(x')\right) \oplus \beta(x)(\checkmark) \cdot \one \tag{\cref{lem:flattening_convex_sums}}\\
        &\equiv \left(\bigoplus_{(a , x') \in A \times X} \beta(x)(a,x') \cdot a \seq h(x')\right) \oplus \beta(x)(\checkmark) \cdot \one
    \end{align*}
    The last step of the proof relies on the observation that for both cases of how $M^\beta_{x,x'}$ is defined, we have that
    $
        p^\beta_{x,x'} s^{a}_{x,x'} = \beta(x)(a,x')
    $.
    Similarly, in the passage between the third and fourth line, we have used the observation that:
    $
        r^\beta_x \cdot b^\beta_x \equiv \beta(x)(\checkmark) \cdot \one
    $
    for both possible cases.
\end{proof}
\section{Completeness}
\coraserquotientisapca*
\begin{proof}
    The signature of positive convex algebras is an endofunctor $\Sigma : \Set \to \Set$. $\pca$ is isomorphic to the full subcategory of $\alg \Sigma$ containing only $\Sigma$-algebras which satisfy the axioms of positive convex algebras.

    Let $I$ be a finite set, $\{p_i\}_{i \in I}$, $\{e_i\}_{i \in I}$ and $\{f_i\}_{i \in I}$, such that for all $i \in I$, $p_i \in [0,1]$, $e_i, f_i \in \Exp$, $e_i \equiv f_i$ and $\sum_{i \in I} p_i \leq 1$. As a consequence of \cref{lem:generalised_sum_congruence}, we have that:
    $$
    \bigoplus_{i \in I} p_i \cdot e_i \equiv \bigoplus_{i \in I} p_i \cdot f_i 
    $$

    In other words, we have that $\ker(\Sigma [-]_\equiv) \subseteq \ker([-]_{\equiv} \o \alpha_{\equiv_0})$, where $\alpha_{\equiv_0} : \Sigma {\Exp}/{\equiv_0} \to {\Exp}/{\equiv_0}$ is a $\pca$ structure map on ${\Exp}/{\equiv_0}$ (isomorphic to Eilenberg-Moore algebra $\alpha_{\equiv_0} : \distf {\Exp}/{\equiv_0} \to {\Exp}/{\equiv_0}$ structure on ${\Exp}/{\equiv_0}$). Combining that with the fact that $\Sigma [-]_\equiv$ is surjective (since $\Set$ endofunctors preserve epis), we can use the diagonal fill-in lemma to conclude that there is a unique map $\alpha_{\equiv} : \Sigma {\Exp}/{\equiv} \to {\Exp}/{\equiv}$ satisfying $\alpha_{\equiv} \o \Sigma [-]_\equiv = [-]_{\equiv} \o \alpha_{\equiv_0}$. Hence, $[-]_\equiv : {\Exp}/{\equiv_0} \to {\Exp}/{\equiv}$ is a $\pca$ homomorphism and $({\Exp}/{\equiv}, \alpha_{\equiv})$ is a positive convex algebra.
    The positive convex algebra structure on ${\Exp}/{\equiv}$ we have obtained is concretely given by:
    $$\bigboxplus_{i \in I} p_i \cdot [e_i] = \left[\bigoplus_{i \in I} p_i \cdot e_i\right] $$
    We can verify the above by checking that $[-]_{\equiv} : {\Exp}/{\equiv_0} \to {\Exp}/{\equiv}$ is a positive convex algebra homomorphism and the above concrete formula corresponds to the unique $\pca$ structure on ${\Exp}/{\equiv}$ which makes $[-]_{\equiv}$ into a homomorphism.
    \begin{align*}
        \left[\bigboxplus_{i \in I} p_i \cdot [e_i]_{\equiv_0}\right]_{\equiv} &= \left[\left[\bigoplus_{i \in I} p_i \cdot e_i \right]_{\equiv_0}\right]_{\equiv} \\
        &= \left[\bigoplus_{i \in I} p_i \cdot e_i \right]\\
        &= \bigboxplus_{i \in I} p_i \cdot[e_i] \\
        &= \bigboxplus_{i \in I}^n p_i \left[[e_i]_{\equiv_0}\right]_{\equiv} \\
    \end{align*}
\end{proof}    
\smappcahomom*
    \begin{proof}
        We need to show that:
        $$d \left(\bigboxplus_{i \in I} p_i \cdot [e_i]\right) = \bigboxplus_{i \in I} p_i \cdot d(e_i)$$
        As a consequence of \cref{thm:fundamental_theorem}, we can safely assume that 
        $e_i \equiv q^i \cdot \one \oplus \bigoplus_{j \in J} r^i_j \cdot a^i_j \seq e^i_j $
        and hence:
        $d(e_i) = \left\langle q^i , \lambda a. \left[\bigoplus_{a = a^i_j} r^i_j \cdot e^i_j\right] \right\rangle$.
        We show the following:
        \begin{align*}
            d \left(\bigboxplus_{i \in I} p_i \cdot [e_i]\right) &= d \left(\left[\bigoplus_{i \in I} p_i \cdot e_i\right]\right) \\
            &=d \left(\left[ \bigoplus_{i \in I} p_i \cdot \left(q^i \cdot \one \oplus \bigoplus_{j \in J} r^i_j \cdot a^i_j \seq e^i_j\right) \right]\right) \\
            &= d \left( \bigoplus_{i \in I} p_i q^i \cdot \one \oplus \bigoplus_{(i,j) \in I \times J} p_ir^i_j \cdot a^i_j \seq e^{i}_j\right) \tag{\cref{lem:flattening_convex_sums}} \\
            &= \left\langle \sum_{i \in I} p_i q^i, \lambda a . \left[\bigoplus_{a = a^i_j} p_ir^i_j \cdot e^i_j\right]\right\rangle\\
            &= \left\langle \sum_{i \in I} p_i q^i, \lambda a . \left[\bigoplus_{i \in I} p_i \cdot \left(\bigoplus_{a = a^i_j} r^i_j \cdot e^i_j\right)\right]\right\rangle \tag{\cref{lem:flattening_convex_sums}} \\
            &= \left\langle \sum_{i \in I} p_i q^i, \lambda a . \bigboxplus_{i \in I} p_i \cdot \left[\bigoplus_{a = a^i_j} r^i_j \cdot e^i_j\right]\right\rangle \\
            &= \bigboxplus_{i \in I} p_i \left\langle q^i , \lambda a. \left[\bigoplus_{a = a^i_j} r^i_j \cdot e^i_j\right] \right\rangle \\
            &= \bigboxplus_{i \in I} p_i \cdot d(e_i)
        \end{align*}
        
    \end{proof}
\onetoone*
\begin{proof}
    Since $X$ is finite,  we can assume that $X = \{s_i\}_{i \in I}$ for some finite set $I$.
    Because of the free-forgetful adjunction between $\pca$ and $\Set$, we have the following correspondence of maps:
    $$\mprset{fraction={===}}
        \inferrule{\zeta  : X \to G \distf X \text{ on } \Set}{\xi : (\distf X, \mu_X) \to G(\distf X, \mu_X) \text{ on } \pca}
    $$
    First, we show that for all $x \in X$, we have that $\gamma_X \o \beta(x) \in \hat{G} \distf X$. Pick an arbitrary $x \in X$. 
    For every $a \in A$ define $p_i^a = \beta(x)(a,s_i)$. This implies that $(\pi_2 \o \gamma_X \o \beta)(x)(a)(x_i) = p_i^a$ if $x_i \in S$. Therefore, we have:
    $$(\gamma_X \o \beta)(x) = \left\langle\beta(x)(\checkmark), \lambda a . \sum_{i \in I} p_i^a \delta_{x_i^a} \right\rangle$$
    Going by isomorphism between $\pca$ and $\Set^\distf$, we can rewrite the above as:
    $$(\gamma_X \o \beta)(x) = \left\langle\beta(x)(\checkmark), \lambda a . \bigboxplus_{i \in I} p_i^a \cdot {x_i^a} \right\rangle$$
    where $\boxplus$ denotes a structure map of a $\pca$ isomorphic to $(\distf X, \mu_X)$, a free Eilenberg-Moore algebra on $X$. 
    From the well-definedness of $\beta(x)$ it follows that:
    $$\sum_{a \in A} \sum_{i=0}^n p_i^a \leq 1 - \beta(x)(\checkmark)$$
    which proves that $\gamma_X \o \beta(x) \in \hat{G} \distf X$.

    For the converse, we know that every arrow $\xi : \distf X \to \hat{G} \distf X$ arises as an extension of $\zeta : X \to \hat{G} D X$. Furthermore, note that any $\zeta : X \to \hat{G} \distf X$ can be factored as composition $\gamma_X \circ \beta$, where $\beta(x)(\checkmark) = \pi_1 \circ \zeta(x)$ and $\beta(x)(a,x') = \pi_2 \circ \zeta(x)(a)(x')$. We can deduce that $\beta(x)$ is a well-defined subdistribution using the fact that $\zeta(x) \in \hat{G}\distf X$. Note that since $\gamma_X$ is monic, $\beta$ is unique.
\end{proof}
\begin{lemma}\label{lem:iamge_of_natural_transformation}
    Let $(X, \alpha)$ be a $\pca$. Then for every $\zeta \in \distf F X$ we have that $G\alpha \o \gamma_X (\zeta)\in \hat{G} X$
\end{lemma}
\begin{proof}
    Let $\zeta \in \distf F X$. Recall that
    $\gamma_X(\zeta) = \langle \zeta(\checkmark), \lambda a . \lambda x . \zeta(a,x)\rangle$.
    Let
    $$S = \{x \in X \mid \exists a \in A . (a,x) \in \supp(\zeta)\}$$
    Without loss of generality, we can assume that $S = \{s_i\}_{i \in I}$ for some finite set $I$. For every $a \in A$ define $p_i^a = \zeta(a,s_i)$. This implies that $(\pi_2 \o \gamma_X)(\zeta)(a)(x_i) = p_i^a$ if $x_i \in S$ or  $(\pi_2 \o \gamma_X)(\zeta)(a)(x_i) =0$ otherwise. Therefore, we have the following:
    $$
    (G\alpha \o \gamma_X)(\zeta) = \left\langle
    \zeta(\checkmark), \lambda a. \bigboxplus_{i\in I} p_i^a \cdot s_i
    \right\rangle
    $$
    Finally, since $\zeta \in \distf F X$ we have that 
    $\sum_{a \in A} \sum_{i \in I} p_i^a \leq 1 - \zeta(\checkmark)$
    which proves that indeed the image of $G \alpha \o \gamma_X$ belongs to $\hat{G} X$.
\end{proof}
\begin{lemma}\label{lem:closed_under_quotients}
    $\hat{G}$-coalgebras are closed under quotients.
\end{lemma}
\begin{proof}
 Let $((X, \alpha_X), \beta_X)$ be a $\hat{G}$-coalgebra,     $((Y, \alpha_Y), \beta_Y)$ be a $G$-coalgebra and let $e : X \to Y$ be a surjective homomorphism $e : ((X, \alpha_X), \beta_X) \to ((Y, \alpha_Y),\beta_Y)$. We need to show that for all $y \in Y$, $\beta(y) \in \hat{G} X$. 
 Pick an arbitrary $y \in Y $. Since $e: X \to Y $ is surjective, we know that $y=e(x)$. Let $\beta_X(x) = \langle o, f\rangle$. We have that: 
 $$\beta_Y(y) = (\beta_Y \o e) (x) = (G e \o \beta_X)(x) = \langle o, e \o f \rangle $$
Since $\beta_X(x) \in \hat{G}X$, we have that for all $a \in A$ there exist $p_a^i \in [0,1]$ and $x_a^i \in X$ such that $f(a)=\bigboxplus_{i \in I} p_a^i x_a^i$ and $\sum_a \sum_{i \in I} p_a^i \leq 1 - o$. Let $e(x_a^i)=y_a^i$
For all $a \in A$, we have that:
$$(e \o f)(a) = e \left( \bigboxplus_{i \in I} p_a^i \cdot x_a^i \right) = \bigboxplus_{i \in I} p_a^i \cdot e(x_a^i)= \bigboxplus_{i \in I} p_a^i \cdot y_a^i$$
Therefore, for all $a \in A$ there exist $p_a^i \in [0,1]$ and $y_a^i \in X$ such that $f(a)=\bigboxplus_{i \in I} p_a^i y_a^i$ and $\sum_a \sum_{i \in I} p_a^i \leq 1 - o$, which completes the proof.
\end{proof}
\quotientcoalgebraforsubfunctor*
\begin{proof}
    $({({\Exp}/{\equiv}, \alpha_{\equiv})},d)$ is a quotient coalgebra of $({\Exp}/{\equiv_0}, G\alpha_{\equiv_0} \o \gamma_{{\Exp}/{\equiv_0}} \o \ol{\partial})$, which by \cref{lem:iamge_of_natural_transformation} is a $\hat{G}$-coalgebra. Because of \cref{lem:closed_under_quotients} so is $({{\Exp}/{\equiv}}, d)$.
\end{proof}
\begin{restatable}{lemma}{quotientiso}\label{lem:isomorphism_quotient}
    $d : ({\Exp}/{\equiv}, \alpha_{\equiv}) \to \hat{G} ({\Exp}/{\equiv}, \alpha_{\equiv})$ is an isomorphism
\end{restatable}
\begin{proof}
    We construct a map $d^{-1} : \hat{G} {\Exp}/{\equiv} \to {\Exp}/{\equiv}$ and show that $d \o d^{-1} = \id = d^{-1} \o d$ (which is enough to show that $d^{-1}$ is a $\pca$ homomorphism, because the forgetful functor $U: \pca \to \Set$ is conservative).

    Let $\langle o , f \rangle \in \hat{G} X$, such that for any $a \in A$, $f(a) = \bigboxplus_{i \in I} p_a^i [e_a^i]$ for some $p_a^i \in [0,1]$ and $[e_a^i] \in {\Exp}/{\equiv}$. We define:
    $$
    d^{-1} \langle o, f\rangle = \left[ o \cdot \one \oplus \left(\bigoplus_{(a,i) \in A\times I} p_a^i \cdot a \seq e_a^i\right) \right]
    $$
    The expression inside the brackets is well defined as $\sum_{a} \sum_{i \in I} p_a^i \leq 1-o$. To show that $d^{-1}$ is well-defined, assume that we have $g : A \to {\Exp}/{\equiv}$, such that $f(a)=g(a) = \boxplus_{i \in I} q_a^i [h_a^i]$ for all $a \in A$ and some $q_a^i \in [0,1]$ and $[h_a^i] \in {\Exp}/{\equiv}$. First of all, because of the way $\pca$ structure on ${\Exp}/{\equiv}$ (\Cref{lem:coarser_quotient_is_a_pca}) is defined we have that:
    $$
    \left[\bigoplus_{i \in I} p_a^i \cdot e_a^i\right] = \bigboxplus_{i \in I} p_a^i [e_a^i] = \bigboxplus_{i \in I} q_a^i [h_a^i]  =  \left[\bigoplus_{i \in I} q_a^i \cdot h_a^i\right]
    $$
    Using the above, we can show:
    \begin{align*}
        d^{-1}\langle o, f\rangle &= \left[ o \cdot \one \oplus \left( \bigoplus_{(a,i) \in A \times I} p^i_a \cdot a  \seq e^i_a\right) \right]\\
        &= \left[ \one \oplus_o \left( \bigoplus_{(a,i) \in A \times I} p^i_a \cdot a  \seq e^i_a \right) \right] \tag{\Cref{prop:binary}}\\
        &= \left[ \one \oplus_o \left( \bigoplus_{(a,i) \in A \times I} q^i_a \cdot a  \seq h^i_a \right) \right] \tag{$\equiv$ is a congruence}\\
        &= \left[ o \cdot \one \oplus \left( \bigoplus_{(a,i) \in A \times I} q^i_a \cdot a  \seq h^i_a\right) \right]\tag{\Cref{prop:binary}}\\
        &= d^{-1}\langle o, g \rangle
    \end{align*}
    which proves well-definedness of $d^{-1}$. For the first side of an isomorphism, consider the following:
    \begin{align*}
        (d \o d^{-1}) (\langle o, f\rangle) &= d \left(\left[o \cdot \one \oplus \left(\bigoplus_{(a,i) \in A\times I} p_a^i \cdot a \seq e_a^i\right)\right]\right) \\
        &= \left\langle o, \lambda a . \left[\bigoplus_{i \in I} p_a^i \cdot e_a^i\right]\right\rangle \\
        &= \left\langle o, \bigboxplus_{i \in I} p_a^i \cdot [e_a^i] \right\rangle \tag{\Cref{lem:coarser_quotient_is_a_pca}}\\
        &= \langle o, f \rangle
    \end{align*}
    \noindent
    For the second part, let $[e] \in {\Exp}/{\equiv}$. By \cref{thm:fundamental_theorem}, we can safely assume that: 
    $$
    e \equiv o \cdot \one \oplus \left( \bigoplus_{(a,i) \in A \times I} p_a^i \cdot a \seq e_a^i \right)
    $$
    Next, observe that:
    \begin{align*}
        (d^{-1} \o d) ([e]) &= d^{-1}\left\langle o , \lambda a . \left[ \bigoplus_{i \in I} p_a^i \cdot e_a^i\right] \right\rangle \\
        &= d^{-1}\left\langle o , \lambda a . \bigboxplus_{i \in I} p_a^i \cdot \left[ e_a^i\right] \right\rangle \\
        &= \left[o \cdot \one \oplus \left(\bigoplus_{(a,i) \in A \times I} p_a^i \cdot a \seq e_a^i\right)\right] \\
        &= d([e])
    \end{align*}
\end{proof}
\solutionshomomorphisms*
\begin{proof}

First, assume that $h : X \to \Exp$ is a solution to the system $\mathcal{S}(\beta)$ associated with $\distf F$-coalgebra $(X, \beta)$. We show that $([-]\o h)^\star : \distf X \to {\Exp}/{\equiv}$ is a $\hat{G}$-coalgebra homomorphism. Because of \cref{lem:simpler_solution_lemma}, we have that for all $x \in X$:
$$
h(x) \equiv \left(\bigoplus_{(a,x') \in A \times X} \beta(x)(a,x') \cdot a \seq h(x')\right) \oplus \beta(x)(\checkmark) \cdot \one
$$
Let $\nu \in \distf X$. The convex extension of the map $[-] \o h : X \to {\Exp}/{\equiv}$ is given by the following:
$$([-] \o h)^\star(\nu) = \left[\bigoplus_{x \in \supp(\nu)} \nu(x)\cdot h(x)\right]$$
Similarly, given $\beta : X \to \distf F X$ we have that:
$$(\gamma_X \o \beta)^\star(\nu) = \left\langle\sum_{x \in \supp(\nu)} \nu(x) \beta(x)(\checkmark),\lambda a . \lambda x'.\sum_{x \in \supp(\nu)}\nu(x)\beta(x)(a, x')\right\rangle$$
We now show the commutativity of the diagram below.
\[\begin{tikzcd}
	X \\
	& {\distf X} &&& {{\Exp}/{\equiv}} \\
	{\distf F X} \\
	{\hat{G}\distf X} &&&& {\hat{G} {\Exp}/{\equiv}}
	\arrow["{\eta_X}"', from=1-1, to=2-2]
	\arrow["{([-]\o h)^\star}"', from=2-2, to=2-5]
	\arrow["{[-]\o h}", from=1-1, to=2-5]
	\arrow["\beta"', from=1-1, to=3-1]
	\arrow["{\gamma_X}"', from=3-1, to=4-1]
	\arrow["{(\gamma_X \o \beta )^\star}", from=2-2, to=4-1]
	\arrow["{\hat{G}([-]\o h)^\star}", from=4-1, to=4-5]
	\arrow["d", from=2-5, to=4-5]
\end{tikzcd}\]
For any distribution $\nu \in \distf X$:
\begin{align*}
&d \o ([-] \o h)^\star(\nu)\\ &= d \left[\bigoplus_{x \in \supp(\nu)}\nu(x) \cdot \left( \left(\bigoplus_{(a,x') \in A \times X} \beta(x)(a,x') \cdot a \seq h(x')\right) \oplus \beta(x)(\checkmark) \cdot \one\right) \right] \\
&= d \left[\left(\bigoplus_{(a, x') \in A \times X} \left( \sum_{x \in \supp(\nu)} \nu(x)\beta(x)(a,x').\right)\cdot a \seq h(x')\right)\right. \\&\quad\quad\quad\quad\left.\oplus \left(\sum_{x \in \supp(\nu)} \nu(x)\beta(x)(\checkmark)\right)\cdot \one\right] \tag{Barycenter axiom}\\
&= \left\langle \sum_{x \in \supp(\nu)} \nu(x)\beta(x)(\checkmark), \lambda a. \left[  
\bigoplus_{x' \in X} \left(\sum_{x \in \supp(\nu(x))}\nu(x)\beta(x)(a,x')\right)\cdot h(x')\right] \right\rangle
\end{align*}
Now, consider $\hat{G}([-] \o h)^\star \o (\gamma_X \o \beta)^\star (\nu)$. We have the following:
\begin{align*}
    &\hat{G}([-] \o h)^\star \o (\gamma_X \o \beta)^\star (\nu)\\
    &= \hat{G}([-] \o h)^\star \left\langle\sum_{x \in \supp(\nu)} \nu(x)\beta(x)(\checkmark), \lambda a. \lambda x'. \sum_{x \in \supp(\nu)} \nu(x)\beta(x)(a,x')\right\rangle\\
    &= \left\langle\sum_{x \in \supp(\nu)} \nu(x)\beta(x)(\checkmark), \lambda a. ([-] \o h)^\star \left(\lambda x'. \sum_{x \in \supp(\nu)} \nu(x)\beta(x)(a,x')\right)\right\rangle\\
    &=\left\langle\sum_{x \in \supp(\nu)} \nu(x)\beta(x)(\checkmark), \lambda a.\left[\bigoplus_{x' \in X} \left( \sum_{x \in \supp(\nu)} \nu(x)\beta(x)(a,x')\right)\cdot h(x')\right]\right\rangle
\end{align*}
and therefore:
$$d \o ([-] \o h)^\star(\nu) = \hat{G}([-] \o h)^\star \o (\gamma_X \o \beta)^\star$$
which shows that $([-] \o h)^\star$ is indeed a $\hat{G}$-coalgebra homomorphism.

For the converse, let $((\distf X, \mu_X), (\gamma_X \circ \beta)^\star)$ be a $\hat{G}$-coalgebra and let $m : \distf X \to {\Exp}/{\equiv}$ be a $\hat{G}$-coalgebra homomorphism from $((\distf X, \mu_X), (\gamma_X \circ \beta)^\star)$ to $(({\Exp}/{\equiv}, \alpha_{\equiv}), d)$. $m$ arises uniquely as a convex extension of some map $\ol{h} : X \to {\Exp}/{\equiv}$. This map can be factorised as $\ol{h} = [-] \o h$, for some $h \to \Exp$. Observe that any two such factorisations determine the same $\equiv$-solution.

Since $m$ is the homomorphism, the inner square in the diagram above commutes. Since the triangle diagrams commute, we have that outer diagram commutes.
Recall \cref{lem:isomorphism_quotient} stating that $d : {\Exp}/{\equiv} \to \hat{G} {\Exp}/{\equiv} $ is an isomorphism. Because of that we have the following for all $x \in X$:
\begin{align*}
    ([-] \o h)(x) &= (d^{-1} \o \hat{G}([-] \o h)^\star \o \gamma_x \o \beta) (x) \\
    &= (d^{-1} \o \hat{G}([-] \o h)^\star) \left\langle\beta(x)(\checkmark), \lambda a . \lambda x' . \beta(x)(a,x') \right\rangle\\
    &= d^{-1} \left\langle\beta(x)(\checkmark), \lambda a. ([-] \o h)^\star(\lambda x' . \beta(x)(a,x')) \right\rangle\\
    &= d^{-1} \left\langle\beta(x)(\checkmark), \lambda a. \left[\bigoplus_{x'\in X } \beta(x)(a,x') \cdot h(x')\right]\right\rangle\\
    &= d^{-1} \left\langle\beta(x)(\checkmark), \lambda a. \bigboxplus_{x'\in X } \beta(x)(a,x') \cdot \left[h(x')\right]\right\rangle \\
    &= \left[\left(\bigoplus_{(a,x')\in A \times X} \beta(x)(a,x')\cdot a \seq h(x')\right)\beta(x)(\checkmark) \cdot \one \right]
\end{align*}
Hence, for all $x \in X$, we have that:
$$h(x) \equiv \left(\bigoplus_{(a,x')\in A \times X} \beta(x)(a,x')\cdot a \seq h(x')\right)\beta(x)(\checkmark) \cdot \one $$
which by \cref{lem:simpler_solution_lemma} implies that $h : X \to \Exp$ is a $\equiv$-solution to the system $\mathcal{S}(\beta)$ associated with the $\distf F$-coalgebra $(X, \beta)$.
\end{proof}
\kleene*
\begin{proof}
	Let $h : X \to \Exp$ be the solution to the system $\mathcal{S}(\beta)$ associated with a $\distf F$-coalgebra $(X, \beta)$ existing because of \Cref{thm:unique_solutions}. For each $x \in X$, set $e_x = h(x)$. Because of \Cref{lem:solutions_homomorphisms}, $([-] \o h)^{\star} : \distf X \to {\Exp}/{\equiv}$ is a $\hat{G}$-coalgebra homomorphism from $(\distf X, (\gamma_X \o \beta)^{\star})$ to $({\Exp}/{\equiv}, d)$
	
		For any $x \in X$, we have the following:
		\begin{align*}
			\llbracket e_x \rrbracket &= \llbracket h(x) \rrbracket \tag{Def. of $e_x$} \\
			&= (\dagger d \o [-]) (h(x)) \tag{\Cref{lem:factoring_lemma2}} \\
			&= (\dagger d \o ([-] \o h)^{\star}) (\eta_{X}(x)) \tag{Universal property} \\
			&=  (\dagger ([-] \o h)^{\star} \o \eta_{X})(x) \tag{Finality}
		\end{align*}
		which is precisely the probabilistic language of the state $x$, given by the generalised determinisation construction.
\end{proof}
\lfplemma*
\begin{proof}
    We establish the simpler conditions of \cite[Definition~III.7]{Milius:2010:Sound}. Recall that in $\pca$ locally presentable and locally generated objects coincide~\cite{Sokolova:2015:Congruences}. Because of that, it suffices to show that every finitely generated subalgebra is contained in finitely generated subcoalgebra. Let $(Y, \alpha_Y)$ be a finitely generated subalgebra of $({\Exp}/{\equiv}, \alpha_{\equiv})$ generated by $[e_1], \dots, [e_n] \in {\Exp}/{\equiv}$ where $1 \leq i \leq n$. We will construct a finitely generated subalgebra $(Z, \alpha_Y)$ of $({\Exp}/{\equiv}, \alpha_{\equiv})$ such that $[e_i] \in Z$ (hence containing $(Y, \alpha_Y)$ as subalgebra) that is subcoalgebra as well.

    First, we define the closure of $[e] \in {\Exp}/{\equiv}$. Recall that by \Cref{thm:fundamental_theorem} every expression can be written in the following normal form: 
    $$
    e \equiv o \cdot \one \oplus \left(\bigoplus_{j \in J} p_j \cdot a_j \seq e'_j\right)
    $$
    Because of that, we will define a closure of $[e]$ to be the least set satisfying:
    $$
    \cl \left(o \cdot \one \oplus \left(\bigoplus_{j \in J} p_j \cdot a_j \seq e'_j\right) \right) = \{[e'_j] \mid j \in J\} \cup \bigcup_{j \in J} \cl([e'_j])
    $$
    Now, define the closure of the set $\{[e_1], \dots, [e_n]\}$ by
    $$
    \cl\{[e_1], \dots, [e_n]\} = \bigcup_{i = 1}^n \cl([e_i])
    $$ 
    Let $(Z, \alpha_Z)$ be the subalgebra of $({\Exp}/{\equiv}, \alpha_{\equiv})$ generated by $\{[a \seq e'] \mid e' \in \cl\{[e_1], \dots, [e_n]\}, a \in A\} \cup \{\one\}$.
    Note that this set is finite, since derivatives are finite (\Cref{lem:locally_finite}) and $A$ is the finite set. Moreover, it is not hard to observe that generators of $(Y, \alpha_Y)$ are contained in $Z$ thus making it into the subalgebra of $(Z, \alpha_Z)$. 

    We proceed to showing that $Z$ is closed under the transitions of the coalgebra structure $d$.
    Let $z = p \cdot [\one] \boxplus \bigboxplus_{j \in J} p_j \cdot [a_j \seq e'_j] \in Z$.
    We will show that $d(z) \in \hat{G}(Z)$. Since $d$ is a $\pca$ homomorphism, we have that
    $$
    d(z) = p \cdot d[\one] \boxplus \bigboxplus_{j \in J} p_j \cdot d([a_j \seq e'_j]) 
    $$
    Observe the following:
    \begin{gather*}
        d([a_k \seq e'_j]) = \langle 0 , f_j \rangle \text{ with } f_j(a)=[e'_j] \text{ if } a = a_j \text{ or otherwise } f(a)=[\zero]\\
        d([\one]) = \langle 1, f\rangle \text{ with } f(a)=[\zero] \text{ for all } a \in A
    \end{gather*}
    Hence, we have that:
    $$
    d(z) = \left\langle p , \lambda a . \left[ \bigoplus_{a_j = a} p_j \cdot e'_j \right] \right\rangle
    $$
    As a consequence of \cref{thm:fundamental_theorem}, we have that
    $
    [e'_j] = \left[ q_j \cdot \one \oplus \bigoplus_{i \in I} p_{i,j} \cdot a_{i,j} \seq e'_{i,j}\right]
    $
    where $[e'_{i,j}] \in \cl\{e_1, \dots, e_n\}$ and hence $[a_{i,j}\seq e'_{i,j}] \in Z$. Hence:
    \begin{align*}
            d(y) &= \left\langle p , \lambda a . \left[ \bigoplus_{a_j = a} p_j \cdot \left(q_j \cdot \one \oplus \bigoplus_{i \in I} p_{i,j} \cdot a_{i,j} \seq e'_{i,j}\right) \right] \right\rangle\\
            &=\left\langle p , \lambda a .\left[p_jq_j \cdot \one \oplus \bigoplus_{\substack{(i , j) \in I \times J \\a_j=a}} p_jp_{i,j} \cdot a_{i,j} \seq e'_{i,j}\right]  \right\rangle \tag{\cref{lem:flattening_convex_sums}}\\
            &=\left\langle p , \lambda a. \left(p_jq_j \cdot \left[\one\right] \boxplus \bigboxplus_{\substack{(i , j) \in I \times J \\a_j=a}} p_jp_{i,j} \cdot \left[ a_{i,j} \seq e'_{i,j}\right]\right)  \right\rangle \tag{\Cref{lem:coarser_quotient_is_a_pca}}\\
    \end{align*}
    Therefore, we know that $d(z) \in GZ$. To complete the proof and show that $d(z) \in \hat{G}Z$, observe the following for all $a \in A$:
    $$
    \sum_{\substack{i,j \in I\\a_j = a}}  p_j (p_{i,j} + q_j) \leq \sum_{j \in J} p_j \leq 1 - p
    $$
\end{proof}

\end{document}